\documentclass[a4paper,10pt,twoside]{article}

\usepackage{soul}

\usepackage{exscale}%
\usepackage{amsfonts}%

\DeclareMathAlphabet{\mathcal}{OMS}{lmsy}{m}{n}

\usepackage{amsmath,amssymb}
\usepackage{enumerate}

\usepackage{epsfig}
\usepackage{bbm}
\usepackage{bm}

\makeatletter
\newcommand\@b@gtimes[1]{%
  \vcenter{\hbox{#1$\m@th\mkern-2mu\times\mkern-2mu$}}}
\newcommand\@bigtimes{%
  \mathchoice{\@b@gtimes\huge}         
             {\@b@gtimes\LARGE}        
             {\@b@gtimes{}}            
             {\@b@gtimes\footnotesize} 
}
\newcommand\bigtimes{\mathop{\@bigtimes}\displaylimits}
\makeatother

\let\epsilon\varepsilon

\let\theta\vartheta

\DeclareMathOperator{\bbbone}{\mathbbm{1}}
\newcommand{\p}[1]{^{\vphantom{bg}\smash{(#1)}}}
\newcommand{\ext}[1][]{^{\mathrm{x}#1}}
\newcommand{\fin}{_\mathrm{fin}}
\newcommand{\sym}{_{\mathrm{sym}}}

\newcommand{\cCo}{\cC}

\providecommand{\abs}[2][]{#1\lvert#2#1\rvert}
\providecommand{\norm}[2][]{#1\lVert#2#1\rVert}

\providecommand{\inv}{^{-1}}

\newcommand{\eps}{{\varepsilon}}        
\newcommand{\vphi}{{\varphi}}           
 
\newcommand{\la}{\langle} \newcommand{\ra}{\rangle}

\newcommand{\bok}{\mathbf{k}}         

\newcommand{\bov}{\mathbf{v}}

\newcommand{\cA}{\mathcal{A}}
\newcommand{\cB}{\mathcal{B}}
\newcommand{\cC}{\mathcal{C}}
\newcommand{\cD}{\mathcal{D}}
\newcommand{\cE}{\mathcal{E}}
\newcommand{\cF}{\mathcal{F}}
\newcommand{\cG}{\mathcal{G}}
\newcommand{\cH}{\mathcal{H}}

\newcommand{\cJ}{\mathcal{J}}
\newcommand{\cK}{\mathcal{K}}

\newcommand{\cT}{\mathcal{T}}
\newcommand{\cU}{\mathcal{U}}
\newcommand{\cV}{\mathcal{V}}

\newcommand{\RR}{\mathbb{R}}            
\newcommand{\NN}{\mathbb{N}}            
\newcommand{\CC}{\mathbb{C}}            
\newcommand{\TT}{\mathbb{T}}

\newcommand{\bTT}{\overline{\TT}}
\newcommand{\uk}{{\underline k}}

\newcommand{\tA}{\widetilde{A}}
\newcommand{\tB}{\widetilde{B}}
\newcommand{\tC}{\widetilde{C}}

\newcommand{\tH}{\widetilde{H}}
\newcommand{\tJ}{\widetilde{J}}
\newcommand{\tS}{\widetilde{S}}

\newcommand{\ta}{\tilde{a}}             

\newcommand{\tv}{\tilde{v}}
\newcommand{\tw}{\widetilde{w}}

\newcommand{\tcC}{{\widetilde{\cC}}}
\newcommand{\tcD}{\widetilde{\cD}}

\newcommand{\tpsi}{\tilde{\psi}}

\newcommand{\trho}{\widetilde{\rho}}

\newcommand{\tvarphi}{\tilde{\varphi}}

\newcommand{\gothh}{\mathfrak{h}}
\newcommand{\huelle}{\mathop{\mathrm{span}}} 
\newcommand{\Tr}{\mathop{\mathrm{Tr}}}

\newcommand{\supp}{\mathop{\mathrm{supp}}}

\newcommand{\ess}{\mathrm{ess}}

\usepackage[amsmath,thmmarks]{ntheorem}

\theoremseparator{.}
\theoremstyle{plain}
\newtheorem{theorem}{Theorem}[section]         

\newtheorem{lemma}[theorem]{Lemma}             
\newtheorem{proposition}[theorem]{Proposition} 
\theorembodyfont{\normalfont}
\newtheorem{definition}[theorem]{Definition}   

\theorembodyfont{\normalfont}
\newtheorem{condition}[theorem]{Condition}
\newtheorem{remark}[theorem]{Remark}         

\theoremstyle{nonumberplain}
\theoremheaderfont{\itshape}
\theoremsymbol{\enspace\ensuremath{\Box}}
\newtheorem{proof}{Proof}

\newcommand{\sess}{\sigma_{\mathrm{ess}}}

\newcommand{\spp}{\sigma_{\mathrm{pp}}}
\newcommand{\ssc}{\sigma_{\mathrm{sc}}}

\newcommand{\loc}{\mathrm{loc}}

\newcommand{\im}{\mathrm{Im}}
\newcommand{\re}{\mathrm{Re}}

\newcommand{\wgt}[1]{\langle #1 \rangle}

\newcommand{\Ortho}{\mathrm{O}}             

\newcommand{\set}[2]{\{#1\,|\,#2\}}
\newcommand{\bigset}[2]{\bigl\{#1\,\big|\,#2\bigr\}}
\newcommand{\Bigset}[2]{\Bigl\{#1\,\Big|\,#2\Bigr\}}

\usepackage{mathtools}
\mathtoolsset{centercolon}

\usepackage[T1]{fontenc}

\usepackage[english]{babel}
\usepackage
{enumitem}
\setlist{itemsep=0pt}

\newcommand{\LLP}{\mathrm{LLP}}

\newcommand{\x}{\mathrm{x}}
\newcommand{\hph}{{\gothh_{\mathrm{bo}}}}
\newcommand{\symL}{L_{\mathrm{sym}}}
\newcommand{\rmsym}{\mathrm{sym}}
\newcommand{\cGam}{\check{\Gamma}}

\newcommand{\one}{\bbbone}
\newcommand{\ri}{\mathrm{i}}
\newcommand{\D}{\mathrm{d}}
\newcommand{\e}{\mathrm{e}}
\newcommand{\vacuum}{\vert 0\ra}

\newcommand{\cross}{\mathcal{X}}  
\newcommand{\shells}{\mathcal{S}} 
\newcommand{\pp}{\mathrm{pp}}
\newcommand{\iso}{\mathrm{iso}}

\newcommand{\thrC}{\cT_\parallel^{(1)}}
\newcommand{\thrN}{\cT_{\not\parallel}^{(1)}}
\newcommand{\cKJ}[2]{\cK_{#2}}

\newcommand{\Annul}{\cA}

\newcommand{\Exc}{\mathrm{Exc}}
\newcommand{\coup}{g}

\newcommand{\epsone}{\eps^{(1)}}
\newcommand{\epstwo}{\eps^{(2)}}
\newcommand{\epsthree}{\eps^{(3)}}
\newcommand{\epsfour}{\eps^{(4)}}
\newcommand{\delone}{\delta'}

\newcommand{\Spectrum}{\Sigma}
\newcommand{\Shell}{S}
\newcommand{\thr}{\cT}

\newcommand{\exchange}{\cE}

\newcommand{\remarkQED}{$\hfill\diamond$}
\newcommand{\myquote}[1]{`#1'}

 \numberwithin{equation}{section}

\usepackage{fix-cm}
\usepackage[T1]{fontenc}

\usepackage[colorlinks]{hyperref}

  \providecommand{\olig}{\overset{o}{=}}
  \DeclareMathOperator{\adjungeret}{ad}

\begin{document}
\bibliographystyle{amsplain}


\title{The translation Invariant Massive Nelson Model: II. \\ The Continuous Spectrum
Below the Two-boson Threshold}

\author{Jacob Schach M{\o}ller\\
Department of Mathematics\\ Aarhus University \\ Denmark\\ email: jacob@imf.au.dk
\\ \\
Morten Grud Rasmussen\footnote{Partially supported by the Lundbeck and Carlsberg Foundations}
\\
Department of Engineering\\
Aarhus University \\ Denmark \\ email: mgr@ase.au.dk }
\maketitle

\begin{abstract}
In this paper we continue the study of the energy-momentum spectrum of
a class of translation invariant, linearly coupled, and massive Hamiltonians
from non-relativistic quantum field theory. The class contains the Hamiltonians 
of E.~Nelson \cite{Ne} and H.~Fr{\"o}hlich \cite{FrH}. In
\cite{MAHP,MRMP} one of us previously investigated the structure of
the ground state mass shell and the bottom of the continuous
energy-momentum spectrum. Here we study the continuous energy-momentum
spectrum itself up to the two-boson threshold, the threshold for energetic
support of two-boson scattering states. We prove that non-threshold
embedded mass shells have finite multiplicity and can accumulate only
at thresholds. We furthermore establish the non-existence of singular continuous
energy-momentum spectrum. Our results hold true for all values of the particle-field
coupling strength but only below the two-boson threshold. The proof
revolves around the construction of a certain relative velocity vector
field used to construct a conjugate operator in the sense of Mourre.
\end{abstract}

\thispagestyle{empty}

\newpage

\tableofcontents
\setcounter{page}{0}
\thispagestyle{empty}\strut
\newpage

\section{Introduction and Results} \label{Chap-Intro}

The present paper is a sequel to \cite{MAHP}, where the ground state
mass shell and the bottom of the continuous energy-momentum spectrum
of the translation invariant massive Nelson model was studied. The massive Nelson model was
introduced in \cite{Ne} as a toy model for nucleon-meson interactions.
It is similar in structure to the large polaron model of H.~Fr{\"o}hlich \cite{FrH}, describing
electrons in polar crystals, interacting with longitudinal
optical  phonons. In \cite{MRMP}, the results
 of \cite{MAHP} were in fact extended to cover a larger class of models
encompassing both the massive Nelson model and the large polaron model.

With the pioneering works of H{\"u}bner-Spohn \cite{HuSp1,HuSp2} and Bach-Fr{\"o}hlich-Sigal \cite{BFS1,BFS2} in the mid 90's,
it became apparent that many of the techniques developed to deal with spectral and scattering problems
for many-body quantum mechanics were in fact also applicable to models of quantized matter
interacting with second quantized fields, sometimes called non-relativistic QFT.
These models range from finite level systems interacting with scalar fields, e.g. the spin-boson model,
to models of atoms and molecules minimally coupled to a second quantized Maxwell field.
They include models from solid state physics
describing electrons interacting with vibrational modes of crystals, i.e. phonons. We recall that acoustic phonons
are modeled by massless fields, and optical phonons by massive fields.

The purpose of this sequel to \cite{MAHP,MRMP} is to study the structure of the
continuous energy-momentum regime. More precisely the region
supporting at most one asymptotic boson, i.e. the region below the
threshold for energetic support of states with two (or more) asymptotic
bosons. This is what is meant with \myquote{two-boson threshold} in the title. 
In particular we prove that fiber Hamiltonians in this energy regime 
have isolated thresholds, non-threshold eigenvalues have finite
multiplicity and can only accumulate at thresholds, and there is no
singular continuous spectrum. Our results do not depend on the
strength of the particle-field coupling.
 The main tool is the construction of an
energy-momentum dependent relative velocity field, describing at fixed total momentum the
difference of velocities of a single asymptotically free boson and an
interacting effective particle (e.g. polaron). This velocity field
goes into the construction of a modified generator of dilation,
which induces a second quantized conjugate operator in the sense of
Mourre, admitting a positive commutator with the fiber Hamiltonian.
Our work can be seen as a fusion of the spectral theory part of
\cite{DGe1} by Derezi{\'n}ski-G{\'e}rard and the paper \cite{GeNi1} by
G{\'e}rard-Nier on analytically fibered operators. A simpler version
of the construction and results of this paper formed a part of the Ph.D.
thesis of the second author \cite{RaThesis}.

We distinguish between the bare particle entering into the model via
its dispersion relation $\Omega\colon \RR^\nu\to [0,\infty)$, and the
effective particles described by the interacting model. In the polaron
model, this is even hammered home by the word \myquote{polaron} used
to refer to the effective particle associated with the ground state,
and the word \myquote{electron} reserved for the (bare) band electron 
entering into the Fr{\"o}hlich Hamiltonian.
While only the bare quantities enter into the construction of the
Hamiltonian, for an observer the bare particle is a mythical entity
which never appears in a scattering experiment. Only effective
particles are manifest as identifiable quantities. 

In relativistically invariant field theories, the particle content of
a theory is determined by eigenvalues of the mass operator. That only
the mass characterizes the effective particles is due to Poincar{\'e} 
invariance, which ensures that the dispersion relations of the
effective particles are forced to be of the form $\sqrt{k^2 + M^2}$, 
where $M$ is an eigenvalue of the mass operator. In our setup, the
model is non-relativistic and only invariance under translations and 
spatial rotations remains. This means that the dispersion relations 
of the effective particles are not determined by a single number,
but a priori by a function on (a subinterval of) the half-line. Very little is known 
about the general structure of the effective dispersion
relations, or mass shells, even for the ground state. This is
a source of complications since we have to take into account the
following features: \textbf{(A)} There may be multiple species of effective
particles, i.e. mass shells. \textbf{(B)} The effective dispersions may not be convex, nor are
they a priori forced to be radially increasing. \textbf{(C)} Excited mass shells may
cross making it ambiguous how to assign velocity to a state
constructed by energy-momentum localizations.

Let us give a heuristic explanation for the role of the effective dispersions  in the analysis of 
the continuous energy-momentum spectrum. The continuous spectrum pertains to 
scattering states of the Hamiltonian, and scattering states should at large times look like
(superpositions of) an interacting bound effective particle plus a number of free
asymptotic bosons, with the sum of momenta and energies of the
constituents summing up to the total momentum and energy of the initial
state. The dynamics for such compound asymptotic systems at total
momentum $\xi$ is governed by a kinetic energy of the form
\begin{equation}\label{IntroEffDisp}
\Shell\p{n}(\xi;k_1,\dots,k_n) :=\Shell(\xi-k_1-\cdots -k_n) +
\omega(k_1)+\cdots +\omega(k_n),
\end{equation}
 where $S$ is the dispersion relation  for an effective particle, $\omega$
is that of the bosons, and $k_1,\cdots,k_n$ labels the momenta of $n$
asymptotic bosons. We observe from this expression that we get
thresholds at energies where $(k_1,\dots,k_n)\to S\p{n}(\xi;k_1,\dots,k_n)$ has
critical points. Computing the gradient we see that this happens when $\nabla
\Shell(\xi-k_1-\cdots -k_n) = \nabla\omega(k_j)$ for all $j$,
i.e. when all the asymptotic bosons have group velocity equal to the
group velocity
of the effective particle. As a first step we ensure that the
threshold set is small. Secondly starting from a given
total momentum  and energy $(\xi,E)$ below the two-boson threshold, we ensure that we can
unambiguously assign a non-vanishing relative velocity field to the scattering
states sitting in a small energy-momentum region near $(\xi,E)$.
To translate non-vanishing of the relative velocity field into a positive commutator estimate,  
we develop a method to extract from the Hamiltonian, expressed
in terms of the bare particle dispersion only, the effective dispersions
used to construct the relative velocity field.

To fix ideas, let us first discuss the form of the vector field under the simplifying assumption
that the ground state mass shell is the only (isolated) mass shell, and that it extends 
to infinity in total momentum. Denote by $\RR^\nu\ni \xi\to \Spectrum_0(\xi)\in\RR$, the 
ground state mass shell, and by $\Spectrum_0\p{1}(\xi;k)$ the associated effective 
dispersion relation, cf. \eqref{IntroEffDisp}, pertaining to one-boson scattering.
In this simple setup one can simply pick the relative velocity 
vector field to be $v_\xi(k) =\nabla_k \Spectrum_0\p{1}(\xi;k)$. 
In the presence of multiple mass shells one has to stitch the associated vector fields together, using an 
energy localization to ensure that for a given $k$ only one shell is available for decay due to energy-momentum conservation.
This is feasible in the absence of level crossings. If mass shells do cross, 
there will be multiple channels accessible and one  has 
to go through a somewhat more painful analysis to ensure that regardless of which channel is chosen, the system breaks up
with positive velocity.
 
We remark that in the weak coupling regime
there are results in the literature about the 
structure of the continuous energy-momentum spectrum 
both for massive and massless bosons. Angelescu-Minlos-Zagrebnov and
Minlos \cite{AMZ,Mi2} study polaron type models
and  prove absence of 
embedded mass shells below the two-boson threshold at small coupling using 
a Feshbach reduction to a (generalized) Friedrichs model also studied
in \cite{Ak}. For what appears to be technical reasons only, the papers
\cite{Ak,AMZ,Mi2} cover neither the Nelson nor the
polaron  model, due to field energy and form factor restrictions,
respectively. 
The Friedrichs model itself, corresponding to  cutting the
Fock space down to the vacuum and one-particle sectors, was studied for all couplings in \cite{GJY,GMR}.
For massless bosons (photons) De Roeck-Fr{\"o}hlich-Pizzo \cite{DRFP} 
show that in the weak coupling regime, (necessarily)
interacting \myquote{large and regular} embedded mass shells must lie close to the bottom of 
the continuous energy-momentum spectrum, and outside
a natural cylinder around zero total momentum. In a narrower cylinder Chen-Faupin-Fr{\"o}hlich-Sigal
establish in \cite{CFFS} the absence of singular continuous spectrum. For sufficiently small energies and momenta, 
these results were previously established in \cite{FGSch3} 
under the additional assumption that soft bosons are non-interacting. 

We stress that our results are valid also outside a weak coupling regime, which necessitates 
-- to put it in somewhat poetic terms -- a final goodbye to the electron and a full embrace of the polaron.

Together with Wojciech Dybalski, the authors are currently working on applying
the constructions of this paper to prove asymptotic completeness below
the two-boson threshold for models of the type considered in this paper. 

 In the remaining part of Section~\ref{Chap-Intro} we introduce the
 Hamiltonian and its energy-momentum spectrum, formulate our main
 results, and at the end we give a geometric picture
 describing our central construction in a nutshell. In Section~\ref{Chap-Regularity} we
 introduce the $C^k(A)$ classes of self-adjoint operators,
 cf.~\cite{ABG}, and prove
 that the fiber Hamiltonians of our model is of class $C^2(A)$,
 whenever $A$ is a second quantized (modified) generator of dilation.
 In Section~\ref{Chap-Mourre} we prove our main theorems. We begin
 with an analysis of the threshold set, followed by a geometric analysis of
 level crossings needed to patch together the relative velocity fields of 
 potentially several effective particle species. Finally we prove a
 Mourre estimate for fiber Hamiltonians, first for a comparison
 Hamiltonian describing an interacting system plus a free boson, and subsequently
 for the Hamiltonian itself.

\vspace{5mm}

\noindent\textbf{Acknowledgments:} 

J.~S.~Møller thanks Denis Ch{\'e}niot for
some clarifying discussions about stratifications we had while visiting Dokuz
Eyl{\"u}l University in Izmir, Turkey.
Both authors thank Wojciech Dybalski for locating a number
of bugs and misprints.

\subsection{The Hamiltonian and its Energy-Momentum Spectrum }

We consider a bare quantum particle,  moving in $\RR^\nu$ and linearly coupled to a scalar field of massive
bosons. The particle Hilbert space is
\begin{align*}
  \cK&:=L^2(\RR_y^\nu)
\end{align*}
where $y$ is the particle position. The bare particle Hamiltonian is
$\Omega(D_y)$, where $D_y:=-\ri\nabla_y$.

The one-boson Hilbert space is
\begin{align*}
  \hph &:=L^2(\RR_k^\nu)
\end{align*}
where $k$ denotes boson momentum, and the one-boson dispersion
relation is $\omega(k)$. See Condition~\ref{Cond:MC} below for the
conditions we impose on the dispersion relations $\Omega$ and $\omega$.

The Hilbert space for the field is the bosonic Fock space
\begin{align}
  \cF&=\Gamma(\hph):=\bigoplus_{n=0}^\infty\cF\p
  n,\quad\textup{where}\\
  \cF\p n&=\Gamma\p n(\hph):=\hph^{\otimes_s n}.
\end{align}
Here $\hph^{\otimes_s n}$ is the symmetric tensor product of $n$
copies of $\hph$. We write $\vacuum=(1,0,0,\dotsc)$ for the vacuum state. The
creation and annihilation operators $a^*(k)$ and $a(k)$ satisfy the
following distributional form identities, known as the canonical
commutation relations.
\begin{align}
  \begin{split}
    [a^*(k),a^*(k')]&=[a(k),a(k')]=0,\\
    [a(k),a^*(k')]&=\delta(k-k')\quad\textup{and}
  \end{split}\\
  a(k)\vacuum &=0.\nonumber
\end{align}
The free field energy is the second quantization of the one-boson
dispersion relation,
\begin{align}
  \D\Gamma(\omega)&=\int_{\RR^\nu}\omega(k)a^*(k)a(k)\,\D k.
\end{align}
The Hilbert space of the combined system is
\begin{align}
  \cH&:=\cK\otimes\cF.
\end{align}
The free and coupled Hamiltonians for the combined system are
\begin{align}
  H_0&:=\Omega(D_y)\otimes\bbbone_\cF+{\bbbone_\cK}\otimes{\D\Gamma(\omega)}\quad\textup{and}\\
  H&:=H_0+V
\end{align}
where the interaction $V$ is given by
\begin{align}
  V&:=\int_{\RR^\nu}\bigl(\e^{-\ri k\cdot y}\,\coup(k)\,{\bbbone_\cK}\otimes
  {a^*(k)}+\e^{\ri k\cdot y}\,\overline{\coup(k)}\,{\bbbone_\cK}\otimes {a(k)}\bigr)\,\D k.
\end{align}
Here $\coup\in\hph=L^2(\RR^\nu)$ is a coupling function.

The total momentum of the combined system is given by
\begin{equation}\label{TotalMomentum}
  P = D_y\otimes{\one_\cF}+{\one_\cK}\otimes{\D\Gamma(k)}.
\end{equation}
The operators $H_0$ and $H$ commute with $P$, i.e.\ $H_0$ and $H$ are
translation invariant. This implies that $H_0$ and $H$ are fibered
operators. Using the unitary transform $I_{\LLP}$ first
introduced by Lee-Low-Pines in \cite{LLP} and given by
\begin{align}\label{LLP}
  I_{\LLP}&:={({F}\otimes{\bbbone_\cF})}\circ\Gamma(\e^{-\ri k\cdot y})
\end{align}
we can identify the fibers of $H_0$ and $H$, respectively. Here $F$ is the
Fourier transform and $\Gamma$ the second quantization functor. We get
\begin{align*}
  I_{\LLP}\,H_0\,I_{\LLP}^*&=\int_{\RR^\nu}^\oplus
  H_0(\xi)\,\D\xi\quad\textup{and}\\
  I_{\LLP}\,H\,I_{\LLP}^*&=\int_{\RR^\nu}^\oplus H(\xi)\,\D\xi,
\end{align*}
where $H_0(\xi)$ and $H(\xi)$ are operators on $\cF$ and given by
\begin{align*}
  H_0(\xi)&=\D\Gamma(\omega)+\Omega(\xi-\D\Gamma(k))\quad\textup{and}\\
  H(\xi)&=H_0(\xi)+\phi(\coup).
\end{align*}
Here $\phi(\coup)$ is the field operator evaluated at $y=0$
\begin{align*}
  \phi(\coup)&=\int_{\RR^\nu}\bigl(\coup(k)\,a^*(k)+\overline{\coup(k)}\,a(k)\bigr)\,\D k.
\end{align*}
See also \cite{RS1,RS2} and \cite{DGe1} for general constructions related
to bosonic Fock space.

\begin{remark}\label{Rem-SecQuant1} Above we introduced the unitary operator
$\Gamma(\e^{-\ri k\cdot y})$ on $\cH$. This is in fact a slight
abuse of notation since the functor $\Gamma$ a priori only maps
contractions on $\hph$ to contractions on $\cF$. Here $q =
\e^{-\ri k\cdot y}$ is a contraction on $\cK\otimes\hph$.

Suppose now that $q$ is a contraction on $\cK\otimes\gothh$ (with
$\cK$ and $\gothh$ Hilbert spaces).
Introduce for each $n\geq 2$ and $j=1,\dots,n$ a unitary operator
$\cE\p{n}_j$ on $\cG\p{n} := \cK\otimes \gothh^{\otimes n}$ (full
$n$-fold tensor product) by the following prescription on simple tensors 
\[
\cE\p{n}_j \bigl(f\otimes u_1\otimes \cdots \otimes u_j \otimes \cdots
\otimes u_n\bigr)
= f\otimes u_j \otimes u_1\otimes \cdots \otimes u_{j-1}\otimes
u_{j+1} \cdots \otimes u_n.
\]
Note that $\cE\p{n}_1 = \one_{\cG\p{n}}$. We extend $q$ to contractions
on $\cG\p{n}$ setting $q_j = \cE_j^{(n)*} q\otimes
\one_{\gothh^{\otimes n-1}} \cE\p{n}_j$. Using this construction we
can define a contraction $G\p{n}(q) = q_1\cdots q_n$  on
$\cG\p{n}$. Let $G\p{0}(q)=\one_\cK$, $G\p{1}(q) = q$, and construct 
the direct sum $G(q) = \oplus_{n=0}^\infty  G\p{n}(q)$ to get a
contraction on $\cG = \cK\otimes \oplus_{n=0}^\infty \gothh^{\otimes n}$.
 If $q$ is unitary, the contraction $G(q)$ is in fact unitary on $\cG$.

 Letting $P_s$ denote the projection onto the symmetric Fock space
$\cF$ inside $\oplus_{n=0}^\infty \gothh^{\otimes n}$ we can now define
$\Gamma(q) = (\one_{\cK}\otimes P_s) G(q) (\one_\cK\otimes P_s)$ as a contraction on
  $\cH$. We warn the reader that for unitary $q$, if $q_i$ and
$q_j$ do not commute, the contraction $\Gamma(q)$ may not be unitary!
\remarkQED\end{remark}

The following minimal conditions will be imposed on the dispersion relations and
coupling function throughout the paper, and often without explicit reference.
We will in particular formulate and use results from the literature under these minimal conditions
although they may in fact hold true under weaker assumptions. The reader is asked to consult the literature for
optimal formulations of known results. 
The notation $\wgt{k}$ is an abbreviation of the function  $\sqrt{1+|k|^2}$. We will use the same notation for
numbers, vectors, and self-adjoint operators.

\begin{condition}[Minimal Conditions]\label{Cond:MC} There exist
  $s_\Omega\in [0,2]$ and $C>0$  
such that the dispersion relations and coupling satisfy:
\begin{enumerate}[label=\textbf{\textup{(MC\arabic*)}},ref=\textup{(MC\arabic*)},leftmargin=*]
\item $\omega\in C(\RR^\nu)$, $\Omega\in C^2(\RR^\nu)$ and $\coup\in L^2(\RR^\nu)$.
\item $m:= \inf_{k\in\RR^\nu} \omega(k) >0$.
\item\label{Item:GrowthOfOmegas} $\forall k\in\RR^\nu$ we have $\omega(k) \leq C \wgt{k}$, $\Omega(k) \geq C^{-1}\wgt{k}^{s_\Omega} - C$.
\item\label{Item:BasicDerOfOmega} $|\partial_\eta^\alpha \Omega(\eta)|\leq C\wgt{k}^{s_\Omega-|\alpha|}$, for all multi-indices $\alpha$ with $0\leq |\alpha|\leq 2$.
\item\label{Item:Subadditivity} $\forall k_1,k_2\in\RR^{\nu}$ we have $\omega(k_1+k_2) < \omega(k_1)+\omega(k_2)$.
\item\label{Item:NoHolesInSpec} Either $\lim_{|k|\to\infty} \omega(k)= \infty$ or: $\sup_{k\in\RR^\nu} \omega(k)<\infty$ and $\lim_{|k|\to\infty}\Omega(k)=\infty$.
\end{enumerate}
\end{condition}

Since $\Omega$ is bounded from below by \ref{Item:GrowthOfOmegas}, we
can assume without loss of generality that $\Omega\geq 0$.

\begin{remark} The translation invariant massive Nelson model as well as Fr{\"o}hlich's polaron model
satisfy the above conditions, both with non-relativistic and
relativistic electron dispersion relation.
 
We recall that the physical  interactions $\coup$,
up to a constant multiple, are
$1/\sqrt{\omega(k)}$, with $\omega(k) = \sqrt{k^2 + m^2}$, for the Nelson model, and $1/|k|$ for the polaron model
in three dimensions. The phonon dispersion relation in the polaron
model is taken to be a positive (material dependent) constant function.  

For both models we are required to impose a UV cutoff on the physical
interaction. 
However, there does not seem to be a fundamental obstacle
to consider also the UV-renormalized models (if $\Omega(\eta) =
\eta^2$) as in \cite{Am}, although an extension to the model without
a UV cutoff is likely to be a delicate task.
\remarkQED\end{remark}

In the remaining part of this section we list a number of known properties of $H$ and its energy-momentum spectrum
\begin{equation}\label{EMSpectrum}
\Spectrum = \bigset{(\xi,E)\in\RR^\nu\times\RR}{ E\in\sigma(H(\xi))}.
\end{equation}
These properties have a long history, see e.g. \cite{Ca,FrJ2,GeLo,Mi1,MAHP,MRMP,Ne,Sp2}, 
with the most complete results in \cite{MRMP}, 
where the reader can also find a comprehensive discussion of the literature on the subject.

Let
\begin{equation}\label{BasicCore}
\cCo:=\Gamma\fin(C_0^\infty(\RR^\nu))\subset\cF,
\end{equation}
where $\Gamma\fin(\cV)$ denotes the algebraic direct sum of the algebraic tensor products
$\cV^{\otimes_\mathrm{s} n}$, where $\cV\subset \hph$. In fact,
when tensor products appear between spaces not all of which are
complete (as Hilbert spaces) an algebraic tensor product is implicitly understood.
The operator $H_0$ is essentially self-adjoint on $C_0^\infty(\RR^\nu)\otimes \cC$, and $V$ is an infinitesimally small
perturbation in the sense of Kato-Rellich. Hence $H$ is also essentially self-adjoint on $C_0^\infty(\RR^\nu)\otimes \cC$,
and the domain of the closures, which we as usual denote by the same symbols, coincide.

Similarly $H_0(\xi)$ is essentially self-adjoint on $\cC$ and $\phi(\coup)$ is an infinitesimally small perturbation,
hence $H(\xi)$ is also essentially self-adjoint on $\cC$. Not only do their domains coincide, 
they are independent of total momentum $\xi$, and we denote the common domain of self-adjointness by
\begin{equation}\label{CommonDomain}
\cD := \cD(H_0(\xi)) = \cD(H(\xi)).
\end{equation}

One can easily verify that $\xi\to (H(\xi)-\ri)^{-1}$ is norm continuous, and hence we observe
by a norm resolvent convergence argument, cf \cite[Theorem~VIII.23]{RS1}, that $\Sigma$ is a closed set.

We pause to introduce some notation. We denote the bottom of the spectrum
of the fiber Hamiltonians by
\begin{equation}\label{Sigma0xi}
  \Spectrum_0(\xi):=\inf\sigma(H(\xi)),
\end{equation}
and the bottom of the spectrum of the full operator by
\begin{equation}\label{Sigma0}
  \Spectrum_0 :=\inf_{\xi\in\RR^\nu}\Sigma_0(\xi)>-\infty.
\end{equation}
Let $n\in\NN$ be some positive integer and
$\uk=(k_1,\dotsc,k_n)\in\RR^{n\nu}$. We introduce the least energy
of a composite system consisting of a copy of an interacting
system at momentum $\xi-\sum_{j=1}^n k_j$ and $n$ non-interacting
photons with momenta $k_j$
\begin{align}\label{n-boson-energy}
   \Spectrum_0\p{n}(\xi;\uk)&:=\Sigma_0(\xi-\textstyle\sum_{j=1}^n k_j)+\sum_{j=1}^n\omega(k_j).
\end{align}
The following functions are the so-called $n$-boson thresholds, i.e. the least energy needed
to support an interacting state and $n$ free bosons at a given total momentum
\begin{equation}\label{n-boson-threshold}
  \Spectrum_0\p{n}(\xi):=\inf_{\uk\in\RR^{n\nu}}\Sigma_0\p{n}(\xi;\uk).
\end{equation}
Abusing notation, we write $\Spectrum_0\p{n}$ both for the function and for its graph.
We should warn the reader that the terminology \myquote{threshold}
carries a dual meaning. 
The use of \myquote{\emph{the} $n$-boson threshold} to describe
$\Sigma\p{n}_0(\xi)$ refers to its literal meaning as the lowest energy
supporting an interacting system and $n$ free bosons.
It is in fact also \myquote{\emph{an} $n$-boson threshold} in the physical
sense of the word threshold as an energy at which the system 
can form an interacting bound state plus $n$ free bosons, with zero breakup velocity. 
We stress that these are in general not the only (physical) thresholds of the system.

With the above notation the HVZ Theorem takes the form
\begin{equation}\label{HVZ}
\sess(H(\xi)) = [\Spectrum_0\p{1}(\xi),\infty),
\end{equation}
and below $\Spectrum_0\p{1}(\xi)$ the spectrum of $H(\xi)$ consists of locally finitely many eigenvalues all of
finite multiplicity, that may only accumulate at $\Spectrum_0\p{1}(\xi)$. We will often write
$\Spectrum_\ess(\xi) = \Spectrum_0\p{1}(\xi)$ to emphasize the role of 
the one-boson threshold as the bottom of the essential energy-momentum spectrum.
We remark that the assumption \ref{Item:NoHolesInSpec} ensures that
the essential energy-momentum spectrum does not have holes.

Due to the subadditivity assumption \ref{Item:Subadditivity} on $\omega$, the $n$-boson thresholds are increasing in $n$, i.e.
\begin{equation}\label{MonotoneThresholds}
\forall n > m: \quad \Spectrum_0\p{n}(\xi) \geq \Spectrum_0\p{m}(\xi).
\end{equation}
If $\lim_{|k|\to\infty} \omega(k) = \infty$ the inequality is strict. If $M=\sup_{k\in\RR^\nu} \omega(k)<\infty$,
then the inequality remains  strict under the extra assumption $2\liminf_{|k|\to\infty}\omega(k) > M$, satisfied obviously by the constant
polaron dispersion, cf. \cite{MRMP}. This can be considered a remark
on non-triviality of our results, since we
work in the energy-momentum region between the graphs of
$\Spectrum_0\p{1}$ and $\Spectrum_0\p{2}$.

Finally we remark that isolated ground states of $H(\xi)$ are
non-degenerate, in particular the ground state mass shell does not cross any possibly
existing isolated excited mass shells. Very little is known about the
structure of the discrete spectrum when we are away from the weak coupling regime.
In fact, we rely only on some symmetry observations and Kato's general
analytic perturbation theory, which applies to the family $\{H(\xi)\}_{\xi\in\RR^\nu}$.
In the weak coupling regime one can compare with the uncoupled model and derive stronger results \cite{AMZ}.
It is still an open problem to produce a verifiable condition under which an
excited mass shell exists, cf. however \cite{Mi1}.

It is a curious fact that in dimensions $\nu=1,2$ the ground
state energy $\Spectrum_0(\xi)$ is an isolated eigenvalue for all
$\xi$ \cite{MAHP,MRMP,Sp2}.
In dimensions $\nu=3$ and higher the ground state mass shell
is expected to vanish
into the continuous energy-momentum spectrum at some critical
momentum, something only known as a fact in the weak coupling regime 
\cite{AMZ,Mi1,Mi2} (not including the physical combinations of dispersions
and couplings).

We identify distinct mass shells, as functions of total momentum
$\xi\to \Shell(\xi)$, with effective particles with dispersion
relation given by $S$. In the case of the polaron model, it is the
ground state which in the literature is referred to as the
Fr{\"o}hlich polaron.

\subsection{Extended Objects}\label{sec-ext}

In this subsection we introduce a new Hamiltonian which
plays the role of the generator of the dynamics for a system of one interacting 
particle and a number of free bosons. The interacting particle and the free bosons
are not coupled. Operators of this type were also used in \cite{Am,DGe1,DGe2,FGSch2,FGSch3,MAHP,MRMP}.
This is a natural object in the context of scattering theory, where one
expects scattering states to decay into interacting bound states
under emission of asymptotically free bosons.

We abbreviate
\[
\cF\ext = \cF\otimes \cF \quad \textup{and} \quad \cH\ext = \cH\otimes \cF= \cK\otimes \cF\ext .
\]
For a self-adjoint operator $a$ on $\hph$, we extend the second quantization operation to $\cF\ext$ by
the construction
\[
\D\Gamma\ext(a) = \D\Gamma(a) \otimes \one_\cF + \one_\cF\otimes\,\D\Gamma(a).
\]
Note that $\D\Gamma\ext(a)$ is essentially self-adjoint on
$\Gamma\ext\fin(D) = \Gamma\fin(D)\otimes\Gamma\fin(D)$, if
$D\subset \hph$ is a domain of essential self-adjointness for
$a$, cf. \cite{RS1}.
We can now define the Hamiltonian
describing an interacting system together with free (asymptotic) bosons.
It is given by
\begin{equation}\label{xHamiltonian}
H\ext =  H\otimes\one_\cF + \one_\cH\otimes\, \D\Gamma(\omega) = \Omega(D_y)\otimes \one_{\cF\ext} + \one_\cK\otimes \,\D\Gamma\ext(\omega) +  V\otimes\one_\cF
\end{equation}
as an operator on the Hilbert space $\cH\ext$. The free operator, with
$\coup=0$, is essentially self-adjoint on 
\begin{equation}\label{xBasicCore}
\cCo\ext = \cCo\otimes\Gamma\fin\bigl(C_0^\infty(\RR^\nu)\bigr),
\end{equation}
and so is $H\ext$ by a
Kato-Rellich argument.

We adopt the terminology from \cite{DGe1} and call $H\ext$ the \emph{extended Hamiltonian}
and $\cH\ext$ the \emph{extended Hilbert space}.
We remark that $H\ext$ commutes with the extended total momentum operator
\begin{equation}\label{xTotalMomentum}
P\ext = P\otimes\one_\cF+ \one_\cH\otimes\,\D\Gamma(k) = D_y \otimes \one_{\cF\ext} + \one_\cK\otimes\, \D\Gamma\ext(k).
\end{equation}

We extend the functor $\Gamma$ from Remark~\ref{Rem-SecQuant1} as
follows. Denote by $\exchange$ the exchange involution on $\cH\ext$ defined
on simple tensors as $\exchange(f\otimes\psi\otimes\varphi) = f\otimes
\varphi\otimes\psi$, where $f\in\cK$ and $\psi,\varphi\in\cF$.
For a contraction $q$ on $\cK\otimes\hph$ we define 
\[
\Gamma\ext(q) = \bigl(\Gamma(q)\otimes\one_\cF\bigr)\exchange\bigl(\Gamma(q)\otimes\one_\cF\bigr)\exchange.
\]
This is only a good definition if the $q_i$'s commute,
cf. Remark~\ref{Rem-SecQuant1}. Denote by $P\p{n}$ the projection of
$\cH\ext$ onto $\cH\otimes\cF\p{n}$ and observe that
$P\p{n}\Gamma\ext(q) = \Gamma\ext(q)P\p{n}$. Abbreviate
$\Gamma\p{n}(q) = P\p{n} \Gamma\ext(q)P\p{n}$ as a contraction on $\cH\otimes\cF\p{n}$.

 We now build the extended Hamiltonian $H\ext$, cf. \eqref{xHamiltonian}, from the inside out as an explicitly fibered operator.
Recall that $H\ext$ commutes with the extended total momentum $P\ext$, cf. \eqref{xTotalMomentum}.
First we introduce fiber Hamiltonians for an interacting system at total momentum $\xi$,
and $n$ asymptotically free bosons with momenta $\uk = (k_1,\dots,k_n)$. These are self-adjoint operators on $\cF$ given by
\[
H\p{n}(\xi;\uk) = H(\xi-\textstyle\sum_{j=1}^n k_j)+\bigl(\sum_{j=1}^n\omega(k_j)\bigr)\one_\cF.
\]
From these operators we construct self-adjoint fiber operators on
$L^2\sym(\RR^{n\nu};\cF)\simeq \cF\otimes\cF\p{n}$ by the direct integral construction
\[
H\p{n}(\xi) = \int^\oplus_{\RR^{n\nu}} H\p{n}(\xi;\uk)\,\D k.
\]
Here the subscript `sym' indicates that the functions are symmetric
under permutation of the $n$ variables, reflecting Bose statistics.
Finally, by another direct integral construction and an application of an extended version of $I_\LLP$, cf. \eqref{LLP},
we can define
\[
H\p{n} = I_\LLP^{(n)^*} \Bigl( \int^\oplus_{\RR^{\nu}} H\p{n}(\xi)\,\D\xi \Bigr) I_\LLP\p{n},
\]
as an operator on the Hilbert space $\cH\p{n} = \cH\otimes\cF\p{n}$.
Here $I_\LLP\p{n} = (F\otimes\one_{\cF\otimes \cF\p{n}}) \Gamma\p{n}(\e^{\ri k\cdot y})$.
The full extended Hamiltonian can now be expressed as a direct sum
\[
H\ext = H\oplus\Bigl(\bigoplus_{n=1}^\infty H\p{n}\Bigr)
\]
as an operator on the extended Hilbert space $\cH\ext = \cH\otimes \cF = \cH\oplus(\oplus_{n=1}^\infty \cH\p{n})$.
Similarly we can introduce fiber operators
\[
H\ext(\xi) = H(\xi)\oplus\Bigl(\bigoplus_{n=1}^\infty H\p{n}(\xi)\Bigr)
\]
as an operator on $\cF\ext$. From this construction we can directly identify $H\ext(\xi)$ as the fiber operators of $H\ext$
and we have the fibration
\[
H\ext = I_\LLP^{\mathrm{x}^*} \Bigl( \int^\oplus_{\RR^\nu} H\ext(\xi) \,\D\xi\Bigr) I_\LLP\ext,
\]
where $I_\LLP\ext = (F\otimes\one_{\cF\ext}) \Gamma\ext(\e^{\ri k\cdot y}) = I_\LLP\oplus (\oplus_{n=1}^\infty I_\LLP\p{n})$.

Note that these constructions tie in well with the notion of  thresholds for supporting states with free bosons, i.e. the functions
$\Spectrum_0\p{n}(\xi;\uk)$ and $\Spectrum_0\p{n}(\xi)$ introduced in \eqref{n-boson-energy} and \eqref{n-boson-threshold}.
More precisely we have
\begin{equation}\label{GroundStateOfHn}
\inf\sigma\bigl(H\p{n}(\xi;\uk)\bigr) = \Spectrum_0\p{n}(\xi;\uk) \quad \textup{and} \quad \inf\sigma\bigl(H\p{n}(\xi)\bigr) = \Spectrum_0\p{n}(\xi).
\end{equation}

\subsection{The Results}

To formulate our main results on the structure of the
energy-momentum spectrum below the two-boson threshold we need an extra
set of assumptions. The condition below depends on a natural number
$n_0$ encoding the amount of control required. The condition will be
used with $n_0=0$ for our result on the structure of the threshold
set, with $n_0=1$ for our result on the structure of embedded point spectrum,
and with $n_0=2$ for our result on absence of singular continuous spectrum.

\begin{condition}[Spectral Theory]\label{Cond:MT}  Let $n_0\in\NN$. We impose
\begin{enumerate}[label=\textbf{\textup{(ST\arabic*)}},ref=\textup{(ST\arabic*)},leftmargin=*]
\item\label{Item:RealAnal} $\omega$ and $\Omega$ are real analytic functions.
\item\label{Item:DerOfv} $\coup$ admits $n_0$ distributional derivatives with 
$\partial^\alpha_k \coup\in L^2_\loc(\RR^\nu\backslash\{0\})$, for all $1\leq |\alpha|\leq n_0$.
\item\label{Item:RotInv}  For all orthogonal matrices
  $O\in\Ortho(\nu)$ 
and all $k\in \RR^\nu$ we have  $\omega(Ok) = \omega(k)$, $\Omega(Ok)=\Omega(k)$
and $\coup(Ok) = \coup(k)$ almost everywhere.
\item $\sup_{k\in\RR^\nu}|\partial^\alpha\omega(k)|<\infty$ for
  $|\alpha|\geq 1$ and $\sup_{\eta\in\RR^\nu}|\partial^\beta
  \Omega(\eta)|<\infty$ for $|\beta|\geq 2$.
\end{enumerate}
\end{condition}

\begin{remark} The assumptions of real analyticity \ref{Item:RealAnal} 
and rotation invariance \ref{Item:RotInv} serve a combined purpose.
The rotation invariance ensures that the energy-momentum spectrum
$\Spectrum$ (and all its components, i.e. pure point, absolutely and
singular continuous spectrum), are
rotation invariant, i.e. $(\xi,E)\in\Sigma$ and $O\in\Ortho(\nu)$
implies $(O\xi,E)\in \Sigma$. 
In particular, the $n$-boson thresholds
$\Spectrum_0\p{n}$ are rotation invariant, cf. \eqref{n-boson-threshold}.
The functions $\Spectrum_0\p{n}(\xi;\uk)$, cf. \eqref{n-boson-energy}, however, only retain invariance under
simultaneous rotation of all $k_j$'s around the $\xi$ axis.

From the point of view of the models discussed so far, these are
reasonable assumptions. 
However, one should keep in mind that
dispersion relations in solid state physics are material dependent
functions and more 
realistic ones are not likely to carry any more symmetry
than discrete symmetries of an underlying lattice. 
We do not consider \ref{Item:RotInv} to be an essential assumption, 
cf. the discussion in Subsection~\ref{Sec-Strata}.
\remarkQED\end{remark}

The above remark, together with Kato's analytic perturbation theory
\cite{Ka}, 
enables a precise description of the isolated part of the energy momentum spectrum
\begin{equation}\label{IsolatedSpectrum}
\Spectrum_\iso = \bigset{(\xi,E)\in\Spectrum}{ E< \Spectrum_\ess(\xi)},
\end{equation}
as a collection of real analytic mass shells and level crossings. The
set $\Spectrum_\iso$ forms a subset of the full pure-point
energy-momentum spectrum
\begin{equation}\label{PointSpectrum}
\Spectrum_\pp = \bigset{(\xi,E)\in\Spectrum}{ E\in\sigma_\pp(H(\xi))}.
\end{equation}
While the general analytic structure of $\Spectrum_\iso$ is understood,
the only thing we can a priori say about $\Spectrum_\pp$ is that it is a Borel subset of $\RR^{\nu+1}$,
cf. Appendix~\ref{App-Fiber}.

We introduce the set of level crossings for isolated mass shells:
\begin{equation}\label{crossings}
\cross :=\bigset{(\xi,E)\in\Spectrum_\iso}{\forall n\in\NN: \
  \Spectrum_\iso\cap B((\xi,E);1/n) \ \textup{is not a graph}},
\end{equation}
where $B(a,r)$ denotes the open ball of radius $r$, centered at $a$.
The connected components of $\cross$ are $S^{\nu-1}$-spheres of the form
$\partial B(0;R)\times \{E\}$, or as a possibly degenerate case, of the form $\{0\}\times \{E\}$.
 The spheres forming the connected components of $\cross$ will also be called \emph{level crossings}. They are
isolated $S^{\nu-1}$-spheres, possibly accumulating either at
infinity or at the bottom
of the essential energy-momentum spectrum
$\Spectrum_\ess$.
In particular, elements $(\xi,E)\in\cross$  represent
eigenvalues $E$ of $H(\xi)$ with a given finite multiplicity.
The connected components of $\cross$ are connected in $\Spectrum_\iso$
by real analytic manifolds, each carrying a finite multiplicity,
in such a way that the sum of the multiplicities of
shells emanating from the same crossing, should equal the multiplicity of the crossing.
We denote the collection of such real analytic manifolds by $\shells$. To be more
precise, by a shell we understand a pair $(\Annul,\Shell)\in\shells$, where $\Annul$ is
an open annulus $\set{\xi\in\RR^\nu}{r< |\xi| <R}$, with $0\leq r < R$, or an open ball centered at $0$.
The function $\Shell\colon \Annul\to\RR$ is real analytic and rotation invariant, with $\Spectrum_0(\xi)\leq\Shell(\xi) <
\Spectrum_\ess(\xi)$
and such that the graphs of the shells together with the level crossings cover
the entire isolated spectrum in energy-momentum space.
For $(\Annul,\Shell)\in\shells$, denote by
\begin{equation}\label{Graph-Of-Shell}
\cG_\Shell = \bigset{(\xi,\Shell(\xi))}{\xi\in \Annul}
\end{equation}
its graph in energy-momentum space. We have
$\cG_\Shell\cap \cG_{\Shell'} = \emptyset$, for all distinct shells $(\Annul,\Shell)\neq
(\Annul',\Shell')$,
and $\cG_\Shell\cap \cross=\emptyset$ for all
$(\Annul,\Shell)\in\shells$. In addition, to ensure we have all shells
covered, we demand that
\[
\Spectrum_\iso = \cross\cup \Bigl(\bigcup_{(\Annul,\Shell)\in\shells} \cG_\Shell \Bigr).
\]
We remark that due to rotation invariance, the mass shells continue
analytically through level crossings.
The reader can consult \cite{Ka} for the analytic structure of
isolated eigenvalues of holomorphic families of self-adjoint operators.
We remark that for a fixed unit vector $u\in\RR^\nu$, the map
$\kappa\to H(\kappa u)$ defines a \myquote{Type A} family of operators. See \cite{FrJ2}.

For a  given element $(\Annul,\Shell)\in\shells$, the graph $\cG_\Shell$ may have $0$, $1$ or  $2$ finite
boundaries that are $S^{\nu-1}$-spheres (perhaps of radius $0$). The case of no boundary, indicates a mass shell that
without crossings extends to infinity in total momentum.  An example of
such a shell would be the ground state mass shell in dimensions one
and two, cf. \cite{MAHP,MRMP,Sp2}.
 A boundary $S^{\nu-1}$-sphere can be one of two things. Either it is
a connected component of $\cross$, i.e. a level crossing,
or it is a subset of  $\Spectrum_\ess = \Spectrum_0\p{1}$, the boundary of the continuous
energy-momentum spectrum. 

Unless a mass shell $(\Annul,\Shell)$ is constant,
its gradient $\nabla \Shell$
can at most vanish on isolated $S^{\nu-1}$-spheres that can only
accumulate at infinity. 
We remark that we do not know the manner in which
mass shells, ground state or excited, dip into the continuous
spectrum.
One could speculate that it does so at worst as a branch of a Puiseaux series, something which
may have useful consequences. See \cite{MMPQM}.

Having discussed the structure of the isolated spectrum,
we now turn to the subset of the continuous energy-momentum spectrum below the two-boson threshold
\begin{equation}\label{EM-Region}
\cE\p{1} = \bigset{(\xi,\lambda)\in\RR^{\nu}\times\RR}{\Sigma\p{1}_0(\xi)\leq \lambda  < \Spectrum\p{2}_0(\xi)}.
\end{equation}
Write $\cE\p{1}(\xi) = [\Sigma\p{1}_0(\xi),\Spectrum\p{2}_0(\xi))$ such that
$\cE\p{1} = \set{(\xi,\lambda)\in\RR^\nu\times\RR}{\lambda\in\cE\p{1}(\xi)}$. 

Our first result is concerned with the structure of possibly embedded
point spectrum inside $\cE\p{1}$. To formulate the theorem, we need to 
carefully formalize the notion of thresholds. We should identify
energy-momenta inside $\cE\p{1}$ where emitted bosons fail to break free from
the remaining interacting system with a non-zero relative velocity, thus preventing
them from becoming asymptotically free field particles. The threshold set
pertaining to one-boson emission processes $\thr\p{1}$ has three
components which we now discuss.

The first, and perhaps most obvious, is the set of
one-boson thresholds where the remaining interacting system after
boson emission ends up inside an isolated  mass shell
$(\Annul,\Shell)\in\shells$.
We define $\Shell\p{1}(\xi;\cdot)\colon \Annul+\xi\to \RR$ by
\begin{equation}\label{EffectiveShell}
\Shell\p{1}(\xi;k) = \Shell(\xi-k) + \omega(k).
\end{equation}
This extends the construction \eqref{n-boson-energy} to (possibly
existing) excited mass shells, and is the post-emission effective dispersion
relation governing the composite interacting system plus emitted boson. The mass
shell contribution to
one-boson thresholds is
\begin{align}\label{LC-Shells}
\thr\p{1}_\shells & = \bigset{(\xi,E)\in\RR^{\nu+1}}{ E\in\thr\p{1}_\shells(\xi)},\\
\nonumber \thr\p{1}_\shells(\xi) &= \bigset{ E\in\RR}{\exists
  (\Annul,\Shell)\in\shells, k\in\Annul+\xi:\ E = \Shell\p{1}(\xi;k),
   \nabla_k \Shell\p{1}(\xi;k)=0}.
\end{align}
We emphasize that $\nabla_k \Shell\p{1}(\xi;k)=0$ is the same as
$\nabla\Shell(\xi-k) = \nabla\omega(k)$, i.e. the asymptotic boson 
and the remaining interacting system have identical velocities.
 This defines one
contribution to the one-boson threshold set.
One can similarly define $n$-body thresholds, which however will sit above the 
(lowest) two-boson threshold $\Spectrum\p{2}_0$ and therefore we disregard them here, cf. \eqref{MonotoneThresholds}.

To understand the next two contributions to the threshold set 
we need to explain the
dynamics at level crossings. Suppose we are at an energy $E$ and total
momentum $\xi$, with one
free boson at momentum $k$ such that
$(\xi-k,E-\omega(k))\in\chi$. The only direction in momentum space we
can control is where $(\xi-k,E-\omega(k))$ moves inside level
crossings, which form  $S^{\nu-1}$-spheres. Inside such spheres the energy
of the bound system stays constant, due to being
constrained to a crossing, so the effective dispersion only varies through
the contribution from $k\to \omega(k)$. The effective dispersion
therefore has
critical momenta where the tangential derivative of $\omega$, with
respect to the $S^{\nu-1}$-sphere, vanishes.  Since $\omega$ is rotation
invariant this can happen in two ways. Either $k$ is parallel to $\xi$
in which case $\nabla\omega(k)$ is normal to the sphere, or it can
happen if $\nabla\omega(k)=0$.

 The next contribution comes from the need to avoid landing on a level
 crossing with $k$ parallel to $\xi$ after emission of one boson with
 momentum $k$.
Given $\xi\in\RR^\nu$, let $u\in\RR^\nu$ be a unit vector such that $\xi = s u$ for some $s\in\RR$.
We introduce the set
\begin{equation}\label{LC-Parallel}
\thr\p{1}_\parallel(\xi) := \bigset{E\in \RR}{\exists r\in\RR:
  \quad \bigl(\xi-r u,E-\omega(r u)\bigr)\in\cross}.
\end{equation}
If $\xi=0$, the unit vector $u$ can be chosen arbitrarily and we observe,
since $\omega$ and the set $\cross$ are rotation invariant, that
\begin{equation}\label{RelaxToCross0}
\thr\p{1}_\parallel(0) = \bigset{E \in \RR}{\exists k\in\RR^\nu:
  \quad \bigl(k,E-\omega(k)\bigr)\in\cross}.
\end{equation}

The final contribution to the threshold set consists of energies at
which it is possible to emit a boson of momentum $k$ with  $\nabla\omega(k)
= 0$
and the remaining interacting system at a level crossing. 
\begin{equation}\label{LC-Angular}
\thr\p{1}_{\nparallel}(\xi) :=  \bigset{E\in \RR}{ \exists k\in\RR^\nu:
  \quad \bigl(\xi-k,E-\omega(k)\bigr)\in\cross \ \textup{and} \ \nabla\omega(k)=0 }.
\end{equation}
The reader can safely on a first reading disregard this contribution since
in typical models $\thr\p{1}_{\nparallel}(\xi)$ will be a subset of $\thr\p{1}_\parallel(\xi)$. This happens of course in dimension $1$,
if $\nabla\omega(k)\neq 0$ for $k\neq 0$, and finally in the case of the polaron model.
Note that we always have $\thr\p{1}_{\nparallel}(0)\subset \thr\p{1}_\parallel(0)$, cf. \eqref{RelaxToCross0}.

The total threshold set at total momentum $\xi$ can now be defined to be
\[
\thr\p{1}(\xi) = \thr\p{1}_\shells(\xi)\cup \thr\p{1}_\parallel(\xi)\cup\thr\p{1}_\nparallel(\xi).
\]
Finally, we introduce the following notation for threshold sets
as subsets of energy-momentum space: $\thr\p{1} = \set{(\xi,E)}{E\in\thr\p{1}(\xi)}$,
 $\thr\p{1}_\parallel =
\set{(\xi,E)}{E\in\thr\p{1}_\parallel(\xi)}$ and $\thr\p{1}_\nparallel =
\set{(\xi,E)}{E\in\thr\p{1}_\nparallel(\xi)}$. 

The first theorem we present establishes the structure of the
threshold set below the two-boson threshold.

\begin{theorem}\label{Thm-thr} Assume Conditions~\ref{Cond:MC}
  and~\ref{Cond:MT}, with $n_0=0$. Let $\xi\in\RR^\nu$. The following holds
\begin{enumerate}[label=\textup{(\roman*)},ref=(\roman*)]
\item\label{Item-pp-1} $\thr\p{1}\cap \cE\p{1}$ is a relatively closed subset of $\cE\p{1}$.
\item\label{Item-pp-2} 
The set $\thr\p{1}(\xi)\cap \cE\p{1}(\xi)$ is a discrete subset of
$\cE\p{1}(\xi)$, i.e. it is at most countable and can accumulate only
at $\Spectrum\p{2}_0(\xi)$.
\end{enumerate}
\end{theorem}

In fact Theorem~\ref{Thm-thr} holds for each of the
three types of thresholds sets individually. This is obvious for
\ref{Item-pp-2}, and follows for \ref{Item-pp-1} from its proof.

The final energy-momenta  we need to avoid come from our desire to handle the infrared singular interaction in the polaron model.
It consists simply of the set $(0,\omega(0)) + \Sigma_\iso$. 
When localizing away from $(0,\omega(0)) + \Sigma_\iso$, we cannot emit a boson with zero momentum, hence we will never meet the infrared singularity.
This contribution can be omitted if the coupling function $\coup$
behaves no worse than $|k|^{\beta}$ at zero, with $\beta>
2-\nu/2$. In order not to introduce a superfluous exceptional set we define
\begin{equation}\label{ExcSet}
\Exc = \left\{ \begin{aligned} 
& (0,\omega(0)) + \Sigma_\iso, & & \partial_{k_j} \coup\not\in
L^2_\loc(\RR^\nu), \ \textup{for some} \ j\in \{1,\dots,\nu\}\\
& \emptyset, & & \partial_{k_j} \coup\in
L^2_\loc(\RR^\nu), \ \textup{for all} \  j\in\{1,\dots,\nu\}
\end{aligned}\right.. 
\end{equation}
We write $\Exc(\xi)$ as usual for the fixed total momentum
fibers of the set $\Exc$. Observe that $\Exc(\xi)$, $\xi\in\RR^\nu$, are discrete sets and
that  $\Exc\cap\cE\p{1}$ is a
relatively closed subset of $\cE\p{1}$. The latter is a consequence of the
HVZ theorem.

Our second theorem is concerned with the structure
of the embedded pure point spectrum below the two-boson
threshold. That is, the set $\Spectrum_\pp\cap \cE\p{1}$, 
cf.~\eqref{PointSpectrum} and~\eqref{EM-Region}.

\begin{theorem}\label{Thm-pp} Assume Conditions~\ref{Cond:MC}
  and~\ref{Cond:MT}, with $n_0=1$. Let $\xi\in\RR^\nu$. The following holds
\begin{enumerate}[label=\textup{(\roman*)},ref=(\roman*)]
\item \label{Item-pp-3} All eigenvalues in $\sigma_\pp(H(\xi))\cap
  \cE\p{1}(\xi)\backslash (\thr\p{1}(\xi)\cup \Exc(\xi))$ have finite multiplicity.
\item\label{Item-pp-4}  The set $\sigma_\pp(H(\xi))\cap \cE\p{1}(\xi)$ is at most countable, with
accumulation points at most in $\thr\p{1}(\xi)\cup \Exc(\xi)\cup
\{\Spectrum_0\p{2}(\xi)\}$.
\item\label{Item-pp-5} The set $(\Spectrum_\pp\cup\thr\p{1}\cup\Exc)\cap\cE\p{1}$ is a
  relatively closed subset of $\cE\p{1}$.
\end{enumerate}
\end{theorem}

The above theorem follows from standard arguments once we have
established a so-called Mourre estimate, cf.~Theorem~\ref{thm:Mourre},
away from $\thr\p{1}$ and $\Exc$. An additional
consequence of a Mourre estimate is a limiting absorption principle and
hence in particular:

\begin{theorem}\label{Thm-sc} Assume Conditions~\ref{Cond:MC} and~\ref{Cond:MT}, with
  $n_0=2$.
Then the fiber Hamiltonians $H(\xi)$
have no singular continuous spectrum below the two-boson threshold, i.e.
\[
\forall \xi\in\RR^\nu: \quad \ssc\bigl(H(\xi)\bigr) \cap \bigl(-\infty,\Sigma\p{2}_0(\xi)\bigr) = \emptyset.
\]
\end{theorem}

\subsection{A Stratification Point of View}\label{Sec-Strata}

The paper is build around the construction of a vector field $v_\xi\in
C_0^\infty(\RR^\nu)$, from which we construct a self-adjoint 
one-body operator $a_\xi=\ri (v_\xi\cdot\nabla_k+ \nabla_k\cdot v_\xi)/2$ and a
second quantized observable $A_\xi = \D\Gamma(a_\xi)$ on $\cF$. The physical
interpretation of $v_\xi$ is that of a relative velocity field, assigning
to a momentum $k$ the difference of the velocity of a bound state at total
momentum $\xi-k$ and the velocity of an asymptotic boson at momentum $k$. In
Section~\ref{Chap-Regularity} we argue that under our assumptions 
the fiber Hamiltonians $H(\xi)$ are of class $C^k(A_\xi)$, for $k=1,2$,
provided \ref{Item:DerOfv} holds with $n_0=k$.

In Section~\ref{Chap-Mourre} we construct the vector field $v_\xi$ locally in energy
in $\cE\p{1}(\xi)$ 
and away from thresholds $\cT\p{1}(\xi)$ and exceptional energies $\Exc(\xi)$, in such a way that we can
deduce at the end of the chapter a Mourre estimate for the pair
$H(\xi)$ and $A_\xi$. From our Mourre estimate, Theorems~\ref{Thm-pp} and~\ref{Thm-sc}
will follow. Theorem~\ref{Thm-thr} will be proved in Subsection~\ref{Sec-Thr-Structure},
and ensures that the construction of $v_\xi$ can be done in a sufficiently
large energy region inside $\cE\p{1}(\xi)$.

The rest of this subsection is devoted to an explanation of the
construction of the threshold set and the vector field $v_\xi$, from the point of view of
stratifications of proper maps. We will not make any attempt to
properly introduce the notions we refer to here, which are entirely standard.
For literature on the subject we refer the reader to \cite[Section~3]{Hi} as
well as \cite{GeNi1,GPLW}. All stratifications discussed here
will satisfy Whitney's regularity condition and the so-called frontier condition: Two strata either have
disjoint closures, or one is contained in the closure of the other.

Consider a (real analytic) ambient space $X\subset \RR^{2\nu+1}$ given as the
following open set
\[
X = \bigset{(k,\xi,E)\in\RR^\nu\times\RR^\nu\times\RR}{E<\Spectrum\p{2}_0(\xi)}.
\]
Along with this we consider fibered ambient spaces
\[
X_\xi = \bigset{(k,E)\in\RR^\nu\times\RR}{(k,\xi,E)\in X}.
\]
We define (real analytic) projections $\Pi\colon X\to \RR^{\nu+1}$ and $\Pi_\xi\colon
X_\xi \to\RR$ by
\[
\Pi(k,\xi,E) = (\xi,E) \quad \textup{and} \quad \Pi_\xi(k,E) = E. 
\]
The projections in fact take values inside the (real analytic) target
spaces $Y = \set{(\xi,E)}{E<\Spectrum\p{2}_0(\xi)}$ and $Y_\xi =
(-\infty,\Spectrum\p{2}_0(\xi))$ respectively.

We now introduce what turns out to be semi-analytic subsets of the
ambient spaces just defined. Let
\[
\cA = \bigset{(k,\xi,E)\in X}{(\xi-k,E-\omega(k))\in\Spectrum_\iso)}
\]
and
\[
\cA_\xi = \bigset{(k,E)\in X_\xi}{(k,\xi,E)\in\cA}.
\]
To see that these sets are semi-analytic we first remark that
$\Sigma_\iso$ is semi-analytic as a subset of the ambient space $\set{(\xi,E)}{
  E<\Spectrum_\ess(\xi)}$. This follows from the analysis of
G{\'e}rard and Nier \cite{GeNi1}. It now follows that $\cA$ and
$\cA_\xi$ are
semi-analytic as subsets of $X$ and $X_\xi$ respectively.
Here one makes use of $E-\omega(k) < \Spectrum_\ess(\xi-k)$ provided
$(k,\xi,E)\in X$.

That the projections $\Pi$ and $\Pi_\xi$, when restricted to $\cA$ and
$\cA_\xi$ respectively, are proper (preimages of compact sets are
compact) is a consequence of \cite[Theorem~2.4]{MRMP}. See also \eqref{ClosingGap} below.

The splitting of $\Spectrum_\iso$ into graphs of mass shells
$\cG_\Shell$ and level crossings, as $S^{\nu-1}$-spheres, is a 
stratification of $\Spectrum_\iso$ with strata being graphs of real analytic
functions of total momentum. This induces a stratification of $X$ and $X_\xi$
into strata which are again graphs of real analytic functions of
$(k,\xi)$ and $k$, respectively.

The threshold sets $\thr\p{1}(\xi)\cap\cE\p{1}(\xi)$ and
$\thr\p{1}\cap \cE\p{1}$ 
can be interpreted as coming from a
Hironaka-stratification of the maps $\Pi_{|\cA}$ and $\Pi_{\xi|\cA_\xi}$ as follows.
The threshold set
$\thr\p{1}(\xi)\cap\cE\p{1}(\xi)$
are zero-strata in a stratification
of the target space $Y_\xi$, and  $\thr\p{1}\cap \cE\p{1}$ 
is the union of zero-strata and those $d$-strata, with $1\leq d\leq \nu$,  transverse
to each $\{\xi\}\times Y_\xi$ inside $Y$.

The strata of the compatible stratification of $\cA$ and $\cA_\xi$
will again be graphs of real analytic functions and the strata not
projecting into the threshold sets are exactly those for which the
function, e.g.  $S\p{1}$ from \eqref{EffectiveShell}, defining the strata has nowhere
vanishing gradient with respect to $k$. 

The vector field $v_\xi$, used at total momentum $\xi$,  will be constructed by gluing together
$k$-gradients of the functions generating  non-threshold strata in
$\cA_\xi$, which plays the physical role of a  vector field of relative
breakup velocity of a compound system consisting of an asymptotic
boson at momentum $k$, and an interacting system at momentum $\xi-k$.

In fact we expect/conjecture that a Hironaka-stratification of the
projections $\Pi$ and $\Pi_\xi$ can be used also without the
assumption \ref{Item:RotInv} on rotation invariance to construct the threshold sets,
and a subsequent analysis of the resulting Whitney-stratification of
$\cA$ and $\cA_\xi$ should make it possible, along the same lines as
employed in Chapter~\ref{Chap-Mourre} of this paper, to build a
vector field $v_\xi$ that works in a Mourre estimate. However, at this
stage where there are still many questions remaining about 
scattering theory as well as the structure  of high energy sectors of
the energy-momentum spectrum, we prefer the home-cooked and completely
explicit  stratification from Chapter~\ref{Chap-Mourre}, 
where we have full control over all the
nuts and bolts. We remark that our insistence on constructing $v_\xi$ as a
vector field necessitates some  geometrical/technical considerations
not met in \cite{GeLaBook} and \cite{GeNi1}, where $v_\xi$ was allowed to be a
more complicated object. Again, with a view towards the future, we
prefer to keep $v_\xi$ as concrete as possible.

We remark that it is a consequence of the analysis in Subsection~\ref{Sec-Thr-Structure} that
$\thr\p{1}\cap\cE\p{1}$ is a semi-analytic subset of the ambient space
$\cE\p{1}$. However, we cannot conclude that
$\Spectrum_\iso\cup(\thr\p{1}\cap\cE\p{1})$ is a semi-analytic subset of
the ambient space $\set{(\xi,E)}{E<\Spectrum_0\p{2}(\xi)}$. The reason
being that we have no control over the manner in which isolated mass
shells may hit the continuous energy-momentum spectrum. Such a
statement
together with control of possibly embedded non-threshold mass shells,
would be a natural input for investigating higher energy sectors.

Another, perhaps more serious, obstacle to analyzing the spectrum
above the two-boson threshold, is the possible existence of embedded
mass shells. Embedded mass shells below the two-boson threshold would 
give rise to one-boson scattering states
between the two- and three-boson thresholds. Controlling the induced
thresholds, in a manner similar to what is done here, necessitates that
embedded mass shells are real analytic. Proving this is well beyond
current technology \cite{FMS,HS,MW}. One solution would be to pass to a weak coupling
regime where the work of \cite{AMZ,Mi2} can be used to rule out embedded
mass shells below the two-boson threshold altogether.

\section{Regularity with Respect to a Conjugate Operator} \label{Chap-Regularity}

In this section we recall the property of a Hamiltonian being of class $C^k(A)$,
with respect to a self-adjoint conjugate operator $A$. In addition, we verify that
our fiber Hamiltonians are of class $C^2(A)$, for conjugate operators of the
general form constructed here. We remark that the particular model studied in this
paper is in fact quite singular in terms of the $C^k(A)$ classes, in that
the free operator $H_0(\xi) = \D\Gamma(\sqrt{k^2+m^2}) +
(\xi-\D\Gamma(k))^2$ is of class $C^2(A)$ but fails to be of class $C^3(A)$,
if one chooses $A$ to be e.g. a second quantized generator of dilation.
While this does not become a serious issue in the present paper,
it will be a more serious obstacle when possible embedded mass shells
are analyzed, since the most advanced results
to date only hold under a $C^2(A)$ assumption \cite{FMS,MW}. 
There are in particular no results allowing one to follow degenerate embedded
eigenvalues under perturbations without stronger regularity assumptions.

Additionally, while $H_0(\xi)$ is of class $C^1(A)$, it does not satisfy a Mourre type
regularity condition on the first commutator, which manifests itself in the fact that the group $W_t$ generated
by the generator of dilation does not preserve the domain of any positive power of $H_0(\xi)$.

The class of conjugate operators we consider in this paper are build
from one-body operators of the form
\begin{equation}\label{a-GeneralForm}
a = \frac12\bigl\{ v\cdot \ri \nabla_k + \ri\nabla_k \cdot v\bigr\}, \quad \textup{where} \ v\in C_0^\infty(\RR^\nu).
\end{equation}
If the $\partial_{k_j}\coup$'s are not all square integrable near $0$, we further require that $0\not\in\supp{v}$.
It is well-known that such $a$ are essentially self-adjoint on $C_0^\infty(\RR^\nu)$.
Furthermore, the second quantization
\begin{equation}\label{A-GeneralForm}
A = \D\Gamma(a)
\end{equation}
is essentially self-adjoint on $\cC$, cf. \eqref{BasicCore}. Being self-adjoint, the operator $a$ generates
a unitary group $w_t = \e^{\ri t a}$ which can be expressed in terms of the flow
$\psi_t$ of the autonomous ODE $\dot{\psi}_t = v(\psi_t)$, with $\psi_0(k) = k$. We have the formula
\begin{equation}\label{OneBodyGroup}
(w_t f)(k) = \sqrt{J_t(k)} f(\psi_t(k)),
\end{equation}
where $J_t$ is the determinant of the Jacobi matrix $D_k\psi_t$. By Liouville's formula
we have the equation
\begin{equation}\label{Jacobiant}
J_t(k) = \e^{\int_0^t \Tr[Dv(\psi_s(k))] \D s},
\end{equation}
which is uniformly bounded in $k$.
By the functorial properties of second quantization we find that the group $\e^{\ri t A}$ generated by $A$ is
$\Gamma(w_t)$.

Note that $\psi_t(k) = k$ for $k\not\in\supp(v)$ and by boundedness of $v$ we have
\begin{equation}\label{FinitePropSpeed}
\sup_{k\in\RR^\nu}\|\psi_t(k)-k\|\leq \sup_{k\in\RR^\nu}\int_0^t \|v(\psi_s(k))\|\,\D s \leq t\|v\|_\infty.
\end{equation}

Unfortunately we use $k$ here both as a momentum variable and
as an integer power for the class $C^k(A)$. Both are standard notation
that we prefer to adhere to and trust the reader to distinguish from the context
when $k$ denotes momentum and when it denotes an integer power.

\subsection{The $C^k(A)$ Classes of Operators}

Let $A$ be a self-adjoint operator on a complex Hilbert space $\cH$.
We recall the notion of $C^k(A)$, $k=1,2,\dotsc$, regularity from \cite{ABG}.

\begin{definition}[The $C^k(A)$ class of
  operators]\label{def:CkofA} Let $A$ be a self-adjoint operator on $\cH$, with domain $\cD(A)$.
\begin{enumerate}[label=\textup{(\roman*)},ref=(\roman*)]
\item\label{Item:BoundedCk}   Let $B\in\cB(\cH)$ be a bounded operator and $k\in\NN$. We say that
  $B\in C^k(A)$ if, for all $\phi\in\cH$, the map $\RR\ni s\mapsto
  \e^{-\ri sA}B\e^{\ri sA}\phi\in\cH$ is $k$ times continuously
  differentiable.
\item\label{Item-SelfAdjointCk} Let $H$ be a self-adjoint operator on
  $\cH$. We say that $H$ is \emph{of class $C^k(A)$} if
there exists $z\in \CC\backslash \sigma(H)$ such that $(H-z)^{-1}\in C^k(A)$.
\end{enumerate}
\end{definition}

Note that $C^k(A)$ is a subalgebra of $\cB(\cH)$, cf.~\cite{ABG,GGM1}. 
Let us make some remarks. The requirement that $A$ and $H$ be self-adjoint can be relaxed
considerably \cite{GGM1}, something we will however not need. The requirement in  \ref{Item-SelfAdjointCk}
that $(H-z)^{-1}\in C^k(A)$ for \emph{some} $z$ in the resolvent set of $H$,
is equivalent to $(H-z)^{-1}\in C^k(A)$ for \emph{all} such $z$.
Finally, we note that if the bounded operator $B$ is itself
self-adjoint then $B\in C^k(A)$ if and only if $B$ is of class
$C^k(A)$. 

The results in this section are recalled from the literature without
proofs, for which we refer the reader to \cite{ABG,FGSi,GG,GGM1}.

We remind the reader that there are several equivalent formulations
for a bounded operator $B$ to be of class $C^1(A)$.  We collect
some as a lemma.

\begin{lemma}\label{lemma:C1equiv}
  Let $B\in\cB(\cH)$. The following are equivalent.
  \begin{enumerate}[label=\textup{(\roman*)},ref=(\roman*)]
  \item\label{item:C1chardef} $B\in C^1(A)$.
  \item \label{item:C1charlim}It holds that
    $\displaystyle\liminf_{s\to0^+}\tfrac{1}{s}\norm{\e^{-\ri sA}B\e^{\ri sA}-B}<\infty$.
  \item\label{item:C1charbound} There is a constant $C$ such that for all
    $\psi,\phi\in\cD(A)$,
    \begin{align}\label{eq:C1def}
      \abs{\la A\psi,B\phi\ra -\la B\psi, A\phi\ra}&\le
      C\norm{\phi}\norm{\psi}.
    \end{align}
  \item \label{item:C1charextend} $B$ maps $\cD(A)$ into itself and $AB-BA\colon\cD(A)\to\cH$
    extends to a bounded operator on $\cH$.
  \end{enumerate}
\end{lemma}

If $B\in C^1(A)$, the commutator $[A,B]$, which is a priori only defined as a form on
$\cD(A)$, can by Lemma~\ref{lemma:C1equiv} be extended to $\cH$. We write $[A,B]^\circ$
for the unique bounded operator on $\cH$ extending the quadratic form $[A,B]$.

If $B = (H-z)^{-1}$, with $H$ being self-adjoint and of class $C^1(A)$,
then we can compute the form $[B,A]$ on $\cD(A)$ and find
that $[B,A] = -(H-z)^{-1}[H,A](H-z)^{-1}$, which is meaningful due to
 Lemma~\ref{lemma:C1equiv}~\ref{item:C1charextend}. Here $[H,A]$ is read as a form on $\cD(H)\cap \cD(A)$. Since the left-hand side
 extends by continuity to the bounded operator $[A,B]^\circ$, we observe that
 $[H,A]$ extends from $\cD(A)\cap \cD(H)$ to a bounded form on $\cD(H)$, which we can and will identify with an operator
 $[H,A]^\circ \in \cB(\cH_1;\cH_{-1})$. Here we used the standard scale of spaces associated with $H$.
 That is $\cH_s$, $|s|\leq 1$, is the completion of $\cD(H)$ with respect to the norm $\|\psi\|_s = \|(|H| + 1)^s\psi\|$.
 We remark that if $H$ is of class $C^1(A)$, then
 \begin{equation}
 \cD(H)\cap \cD(A) \ \textup{is dense in} \ \cD(H)
 \end{equation}
 and hence, the extension $[H,A]^\circ$ of the form $[H,A]$ is unique.

 We will need the following well-known lemma

 \begin{lemma}\label{Lemma:CommutatorFromGroup}
  If $H$ is a self-adjoint operator of class $C^1(A)$ and
  $W_t=\e^{\ri tA}$ is the unitary group associated to the self-adjoint
  operator $A$, then we have
  \begin{align*}
  \forall \psi,\phi\in\cD(H):\quad   \la \psi,\ri [H,A]^\circ\phi\ra &=\lim_{s\to 0}\frac{1}{s}\bigl(
    \la H\psi, W_s\phi\ra-\la \psi,W_sH\phi\ra\bigr).
  \end{align*}
\end{lemma}

\subsection{$H(\xi)$ is of Class $C^2(A)$}

In this subsection we state and prove a $C^2(A)$ regularity result for the fiber Hamiltonians
with respect to conjugate operators of the type \eqref{A-GeneralForm}. Since this is of independent interest, we formulate
precise conditions under which our results hold, conditions that are implied by a combination of
Conditions~\ref{Cond:MC} and~\ref{Cond:MT}.

\begin{condition}\label{Ck-Condition} We say that $(\omega,\Omega,\coup,v)$ satisfies a $C^k$-condition, $k=1,2$,
 if there exists $s_\Omega\in [0,2]$ such that
 \begin{enumerate}[label={\textbf{\textup{(Ck\arabic*)}}},ref={\textup{(Ck\arabic*)}},leftmargin=*]
\item\label{Item:CkRegularity} $\omega,\Omega\in C^k(\RR^\nu)$ and $v\in C^\infty_0(\RR^\nu)$.
\item\label{Item:CkMassive} $\inf_{\eta\in\RR^\nu} \omega(\eta) >0$.
\item\label{Item:CkBasicDerOfOmega} $\exists C>0$ such that
  $\Omega(\eta)\geq C^{-1} \wgt{\eta}^{s_\Omega} - C$ 
 and $|\partial^\alpha \Omega(\eta)|\leq C \wgt{\eta}^{s_\Omega-|\alpha|}$, $|\alpha|\leq k$.
\item\label{Item:CkCoup} $\coup\in L^2(\RR^\nu)$ admits $k$ distributional derivatives with $\partial^\alpha \coup \in L^2_\loc(\RR^\nu\backslash\{0\})$, $|\alpha|\le k$.
\item\label{Item:ExcCase} If $\partial_{k_j}\coup\not\in L^2_\loc(\RR^\nu)$, for some $j\in\{1,\dotsc,\nu\}$, then
$\supp{v}\subset \RR^\nu\backslash\{0\}$.
 \end{enumerate}
\end{condition}

As for Condition~\ref{Cond:MC}, we can assume without loss of generality that
$\Omega\geq 0$.
Note that due to the $\xi$- and $\coup$-independence of the domain of $H(\xi)$,
and the equivalence of the associated $\|\cdot\|_s$ norms, the scale of spaces $\cH_s$ are $\xi$- and $\coup$-independent.
To avoid ambiguity we use $H_0(0)$ to define the $s$-norms.

\begin{proposition}\label{Prop:C2} Suppose $(\omega,\Omega,\coup,v)$ satisfies a $C^k$-condition with
\begin{enumerate}[label={$k=\arabic*$:},ref= {\myquote{$k=\arabic*$}},leftmargin=*]
\item\label{Item:HisC1} Then for all $\xi\in\RR^\nu$ the fiber Hamiltonian $H(\xi)$ is of class $C^1(A)$ and
we have the explicit form of the commutator
\[
\ri [H(\xi),A]^\circ = \D\Gamma(v\cdot\nabla\omega) -\D\Gamma(v)\cdot\nabla\Omega(\xi-\D\Gamma(k)) -\phi(\ri a \coup).
\]
Furthermore
\begin{equation}\label{RACommBound}
\forall t\in \bigl[0,\tfrac12\bigr]:\quad [H(\xi),A]^\circ\in\cB\bigl(\cH_{1-t},\cH_{-\frac12-t}\bigr).
\end{equation}
\item\label{Item:HisC2} Then for all $\xi\in H(\xi)$ the fiber Hamiltonian $H(\xi)$ is of class
 $C^2(A)$.
\end{enumerate}
\end{proposition}

\begin{proof} Fix a $\xi\in\RR^\nu$.
For the purpose of this proof we abbreviate $H_0 = H_0(\xi)$ and $H=H(\xi)$.
Recall the notation $\cD$ \eqref{CommonDomain} for the common domain of $H_0(\xi)$ and $H(\xi)$, $\xi\in\RR^\nu$, and $\cC$ \eqref{BasicCore}
for the common core of $H(\xi)$ and $A$. We introduce a slightly larger
common core as the algebraic direct sum
\begin{equation}\label{BiggerBasicCore}
\tcC = \bigoplus_{n=0}^\infty C^\infty_{0,\mathrm{sym}}(\RR^{n\nu}),
\end{equation}
where the subscript \myquote{sym} indicates that the functions are
symmetric under permutations of the $n$ variables.

We begin with the case \ref{Item:HisC1}
and observe that
\begin{equation}\label{Eq-R0PreserveC}
\forall z\in\CC\backslash \sigma(H_0):\quad (H_0-z)^{-1}\cC \subset \tcC,
\end{equation}
which ensures that the following computation, for $\psi,\varphi\in\cC$,
\begin{equation*}
    \bigl\la\psi, \bigl[(H_0+1)\inv,A\bigr]\varphi\bigr\ra =-\bigl\la(H_0+1)\inv\psi,[H_0,A](H_0+1)\inv\varphi\bigr\ra
\end{equation*}
is meaningful. As a form on $\tcC$ one can easily compute that
\[
F_\tcC := \ri [H_0,A]_{|\tcC} =  \D\Gamma(v\cdot\nabla\omega) -\D\Gamma(v)\cdot\nabla\Omega(\xi-\D\Gamma(k)).
\]
Since $v\cdot\nabla\omega$ is uniformly bounded, cf. \ref{Item:CkRegularity}, we can bound the first term by a number operator, and hence due to \ref{Item:CkMassive} by $H_0$.
Likewise, we can bound $\D\Gamma(v)$ by $H_0$, and the operator $\nabla\Omega(\xi-\D\Gamma(k))$ can due to \ref{Item:CkBasicDerOfOmega}
be controlled by $H_0^{1/2}$. Recall that $s_\Omega\leq 2$. 
This yields the following bound for all $\tpsi,\tvarphi\in\tcC$
\[
\bigl|\bigl\la\tpsi,\ri [H_0,A]\tvarphi\bigr\ra\bigr| 
\leq C \Bigl(\bigl\|(H_0+1)\tpsi\bigr\|^2 + \bigl\|(H_0+1)^\frac12 \tvarphi\bigr\|^2\Bigr).
\]
Hence we find that
\[
\forall \psi,\varphi\in\cC:\quad
\bigl|\bigl\la\psi,\ri\bigl[(H_0+1)\inv,A\bigr]\varphi\bigr\ra\bigr|
\leq C \Bigl(\bigl\|\psi\bigr\|^2 + \bigl\|(H_0+1)^{-\frac12}\varphi\bigr\|^2 \Bigr).
\]
Since $\cC$ is a core for $A$, this proves that $H_0$ is of class  $C^1(A)$ and hence
$\ri [H_0,A]$ has a unique extension by continuity from $\cD\cap\cD(A)$ to a
bounded form $\ri [H_0,A]^\circ$ on $\cD$.  We now observe, appealing to the bound, that
the form $F_\tcC$ extends continuously to a bounded form $F_{\cD}$ on $\cD$, defined by the same expression.
Since $F_{\cD}$ and $\ri [H_0,A]^\circ$ coincide on $\tcC$, they must also be identical
as forms on $\cD$. Finally we observe by symmetry and interpolation that
\begin{equation}\label{R0ACommBound}
\forall t\in\bigl[0,\tfrac12\bigr]: \quad \ri [H_0,A]^\circ \in \cB\bigl(\cH_{1-t},\cH_{-\frac{1}{2}-t}\bigr).
\end{equation}

We now turn to the full fiber Hamiltonian $H$.  Since $\phi(\coup)$ is $H_0^{1/2}$ bounded, we can
choose $\lambda>0$ large enough such that $\norm{\phi(\coup)R_0(\lambda)}<1$, where $R_0(\lambda) = (H_0+\lambda)^{-1}$.
We can now write
  \begin{equation}\label{Eq-ResEquation}
    R(\lambda) := (H+\lambda)\inv= R_0(\lambda)(\one+\phi(\coup)R_0(\lambda))\inv.
  \end{equation}
  Recall that $C^1(A)$ is a subalgebra of $\cB(\cH)$, and $S\in C^1(A)$
  invertible implies $S^{-1}\in C^1(A)$ (see
  \cite[Corollary~2.10]{GGM1}). Hence it
  suffices to show that $\phi(\coup)R_0(\lambda)\in C^1(A)$ in order to prove that
  $H$ is of class $C^1(A)$.

  Using Lemma~\ref{lemma:C1equiv}~\ref{item:C1charextend} we conclude that for $\vphi\in\cC\subset\cD(A)$
  we have $R_0(\lambda)A\vphi = A R_0(\lambda)\vphi - [R_0(\lambda),A]^\circ\vphi$.
Calculate for $\psi,\varphi\in\cC$
 \begin{equation}\label{eq:C1H_0plusV}
  \begin{aligned}
    &\bigl\la\phi(\coup)\psi, R_0(\lambda)\ri A\vphi\ra-\la A\psi,\ri \phi(\coup) R_0(\lambda)\vphi\bigr\ra\\
    & \quad =\bigl\la\phi(\coup)\psi,\ri A R_0(\lambda)\vphi\ra-\la A\psi,\ri \phi(\coup)R_0(\lambda)\vphi\bigr\ra\\
    & \qquad -\bigl\la\phi(\coup)\psi,R_0(\lambda)\ri [H_0,A] R_0(\lambda)\vphi\bigr\ra\\
    & \quad = -\bigl\la\psi,\phi(\ri a \coup) R_0(\lambda)\vphi\bigr\ra-
    \bigl\la\psi,\phi(\coup)R_0(\lambda)\ri [H_0,A]^\circ R_0(\lambda)\vphi\bigr\ra.
  \end{aligned}
 \end{equation}
 Note that $a\coup\in L^2(\RR^\nu)$ due to \ref{Item:CkRegularity},  \ref{Item:CkCoup} and~\ref{Item:ExcCase}.
By using \eqref{R0ACommBound} and $H_0^{1/2}$-boundedness of $\phi(\coup)$, it follows that for all $\psi,\vphi\in\cC$
\[
\bigl|\bigl\la\psi,\phi(\coup)R_0(\lambda)\ri A\vphi\bigr\ra-\bigl\la A\psi,\ri
\phi(\coup) R_0(\lambda)\vphi\bigr\ra\bigr|
\leq C\bigl(\norm{\psi}^2+\norm{\vphi}^2\bigr),
\]
for some $C>0$. Since $\cC$ is a core for $A$ this bound extends to $\cD(A)$ and hence by Lemma~\ref{lemma:C1equiv}
we conclude that $\phi(\coup)R_0(\lambda)\in C^1(A)$. 
To verify the formula for $\ri [H,A]^\circ$ it now suffices to verify the formula
as a form on $\cC$ as we did for $[H_0,A]^\circ$. 
The perturbation contributes an $H$-bounded term, so it is \eqref{R0ACommBound}
that is the most singular contribution and hence \eqref{RACommBound} holds true.
This completes the proof for the case \ref{Item:HisC1}.

We turn to the case \ref{Item:HisC2}. Having established that
$R(\lambda)\in C^1(A)$, 
one can repeat the argument around \eqref{eq:C1H_0plusV} above
to conclude that $\phi(\coup)R(\lambda)\in C^1(A)$.
Since $(\phi(\coup)R(\lambda))^* =  \overline{R(\lambda) \phi(\coup)}$, the closure of
 $R(z)\phi(\coup)$ defined a priori on $\cD$, we get
 \begin{equation}\label{phiRC1}
 \phi(\coup)R(\lambda)\ \textup{and} \ \overline{R(\lambda) \phi(\coup)} \ \textup{are of class} \  C^1(A).
 \end{equation}
 Compute as an identity between bounded operators
\begin{align*}
\ri [R(\lambda),A]^\circ & = - R(\lambda)\bigl\{ \ri [H_0,A]^\circ - \phi(\ri a \coup)\bigr\}R(\lambda)\\
& = -R(\lambda) (H_0+\lambda)R_0(\lambda) \ri [H_0,A]^\circ  R(\lambda) + R(\lambda)\phi(\ri a \coup) R(\lambda)\\
& = -\overline{R(\lambda)\phi(\coup)}\ri[R_0(\lambda),A]^\circ (\one+\phi(\coup) R_0(\lambda))\inv +R(\lambda)\phi(\ri a\coup)R(\lambda)\\
    & \quad +\ri[R_0(\lambda),A]^\circ(\one+\phi(\coup)R_0(\lambda))\inv,
\end{align*}
where the last equality made use of \eqref{Eq-ResEquation}. We conclude that to show that $H$ is of class $C^2(A)$,
it suffices to show that $[R_0(\lambda),A]^\circ$ and $R(\lambda)\phi(\ri a\coup)R(\lambda)$ are both of class $C^1(A)$.

 We begin with $[R_0(\lambda),A]^\circ$. Compute for $\psi,\vphi\in \cC$
 \begin{equation}\label{C2FirstTerm}
  \begin{aligned}
  &\bigl\la\psi,[R_0(\lambda),A]^\circ  A\vphi\bigr\ra-\bigl\la A\psi,[R_0(\lambda),A]^\circ\vphi\bigr\ra\\
  & \quad  = -\bigl\la\psi,R_0(\lambda)[[H_0,A]^\circ,A]R_0(\lambda)\vphi\bigr\ra\\
  & \qquad +2\bigl\la\psi,R_0(\lambda)[H_0,A]^\circ R_0(\lambda)[H_0,A]^\circ R_0(\lambda)\vphi\bigr\ra,
  \end{aligned}
  \end{equation}
where we used again \eqref{Eq-R0PreserveC} and $A\cC\subset\cC$ to perform the computations.
The form $[[H_0,A]^\circ,A]$ should be understood as a form on $\tcC$,
cf.~\eqref{BiggerBasicCore}, where it can be computed to be
\begin{align*}
F'_{\tcC} & := [[H_0,A]^\circ,A]\\
 & = -\D\Gamma(\la v,(\nabla^2\omega)v\ra)-\D\Gamma(\la(\nabla v)v,\nabla\omega\ra)\\
      & \quad -\la\D\Gamma(v),\nabla^2\Omega(\xi-\D\Gamma(k))\D\Gamma(v)\ra
    +\D\Gamma((\nabla v)v)\cdot\nabla\Omega(\xi-\D\Gamma(k)).
  \end{align*}
The two first terms in $F'_{\tcC}$ are controlled by the number operator, cf. \ref{Item:CkRegularity}, and hence by $H_0$.
The third term is the most singular and require a square of the number operator to bound, cf. \ref{Item:CkBasicDerOfOmega}, and hence
is just bounded as a form on $\cD$. The fourth and final term can be controlled by $H_0^{3/2}$.
In conclusion we find the existence of a $C>0$ such that
\[
\forall \tpsi,\tvarphi\in \tcC:\quad
\bigl|\bigl\la\tpsi,F'_{\tcC}\,\tvarphi\bigr\ra\bigr| 
\leq C\Bigl(\bigl\|(H_0+\lambda)\tpsi\bigr\|^2 + \bigl\|(H_0+\lambda)\tvarphi\bigr\|^2\Bigr)
\]
We can now estimate the left-hand side in \eqref{C2FirstTerm}, cf. also \eqref{R0ACommBound}, and find that
\[
\forall\psi,\vphi\in\cC:\quad
\bigl|\bigl\la\psi,[[R_0,A]^\circ,A]\vphi\bigr\ra\bigr| 
\leq C\bigl(\|\psi\|^2 + \|\vphi\|^2\bigr),
\]
for some $C>0$.  Since $\cC$ is a core for $A$, we have thus established that $[R_0(\lambda),A]^\circ\in C^1(A)$.
Note that to control the last term in \eqref{C2FirstTerm} using \eqref{R0ACommBound}, one has to make full use of all the
free resolvents.

It remains to consider $R(\lambda)\phi(\ri a\coup)R(\lambda)$. Writing
\[
\begin{aligned}
R(\lambda)\phi(\ri a\coup)R(\lambda) &= R_0(\lambda)\phi(\ri a\coup)R_0(\lambda) + 2\re\big\{R(\lambda)\phi(\coup) R_0(\lambda)\phi(\ri a\coup)R_0(\lambda)\} \\
&\quad + R(\lambda)\phi(\coup) R_0(\lambda)\phi(\ri a\coup)R_0(\lambda)\phi(\coup) R(\lambda),
\end{aligned}
\]
we appeal to \eqref{phiRC1} and conclude that it suffices to show that $R_0(\lambda)\phi(\ri a\coup)R_0(\lambda)$ is of class $C^1(A)$.
Here we can compute as a form on $\cC$ for one last time
\[
\begin{aligned}
[R_0(\lambda)\phi(\ri a\coup)R_0(\lambda),A] & = -\ri R_0(\lambda) \phi(a^2 \coup) R_0(\lambda)\\
& \quad + [R_0(\lambda),A]^\circ \phi(\ri a\coup)R_0(\lambda) + R_0(\lambda)\phi(\ri a\coup)[R_0(\lambda),A]^\circ.
\end{aligned}
\]
By \ref{Item:CkRegularity}, \ref{Item:CkCoup} and~\ref{Item:ExcCase}, $a\coup,a^2\coup\in L^2(\RR^\nu)$ and the right-hand side clearly extends to a bounded operator and we are done.
\end{proof}

\subsection{Extended Operators}\label{Sec-Extended}

Below we will make use of the following two simple observations, the proofs of which are left to the reader.

\begin{lemma}\label{Lemma:SumsOfAs} Let $H,A_1,A_2,A_3$ be self-adjoint operators such that $H \in C^1(A_j)$, $j=1,2$.
Suppose furthermore that there exists a dense set $\cD_0$, with the following properties:
\begin{enumerate}[label=\textup{(\roman*)},ref=(\roman*)]
\item  $\cD_0\subset \cD(A_j)$, $j=1,2,3$.
\item $\cD_0$ is a core for $A_3$.
\item  For all $\psi\in\cD_0$ we have  $A_3 \psi = A_1\psi + A_2\psi$.
\end{enumerate}
Then $H$ is of class $C^1(A_3)$ and  $[H,A_3]^\circ = [H,A_1]^\circ + [H,A_2]^\circ$ as an identity between
elements of $\cB(\cH_1;\cH_{-1})$.
\end{lemma}

\begin{lemma}\label{Lemma:SumsOfHs} Let $\{\cH_n\}_{n\in\NN}$ be a
  family of Hilbert spaces, and suppose that for each $n\in\NN$ we are given
two self-adjoint operators $H_n$ and $A_n$ on $\cH_n$, with $H_n$ of class $C^1(A_n)$.
Then $H = \oplus_{n=1}^\infty H_n$ is of class $C^1(A)$, with
$A=\oplus_{n=1}^\infty A_n$, as self-adjoint operators on 
$\cH=\oplus_{n=1}^\infty \cH_n$.
Furthermore $[H,A]^\circ = \oplus_{n=1}^\infty [H_n,A_n]^0$ under the identification $\cB(\cH_1;\cH_{-1})= \oplus_{n=1}^\infty \cB(\cH_{n;1};\cH_{n;-1})$.
\end{lemma}

In the proposition below $\tv\in C_0^\infty(\RR^\nu\backslash)$,
$\ta = (\tv\cdot \ri \nabla_k + \ri\nabla_k \cdot \tv)/2$ and $\tA\p{\ell} = A\otimes \one_{\cF\p{\ell}}
+\one_\cF\otimes\,\D\Gamma\p{\ell}(\ta)$. The tilde-free versions are as usual constructed using $v$.
The extended operators being discussed in this subsection were
introduced in Subsection~\ref{sec-ext}.

\begin{proposition}\label{Prop:ExtendedC1} Suppose
  $(\omega,\Omega,\coup,v)$ satisfies a $C^1$-condition. 
Then for all $\ell\in\NN$ and $\xi\in\RR^\nu$ the following holds:
$H\p{\ell}(\xi)$ is of class $C^1(A\otimes \one_{\cF\p{\ell}})$, 
$C^1(\one_\cF\otimes \,\D\Gamma\p{\ell}(\ta))$ and $C^1(\tA\p{\ell})$,
with
\[
\begin{aligned}
& \ri \bigl[H\p{\ell}(\xi),A\otimes \one_{\cF\p{\ell}}\bigr]^\circ 
=  \int^\oplus_{\RR^{\ell\nu}} \ri \bigl[H\p{\ell}(\xi;\uk),A\bigr]^\circ \D k,\\
& \ri \bigl[H\p{\ell}(\xi),\one_\cF\otimes\, \D\Gamma\p{\ell}(\ta)\bigr]^\circ =
\int^\oplus_{\RR^{\ell \nu}}\sum_{j=1}^\ell \tv(k_j)\cdot\Bigl( 
\nabla\omega(k_j)- \nabla\Omega\bigl(\xi-\textstyle\sum_{j=1}^\ell k_j -\D\Gamma(k)\bigr)\Bigr)\D k,\\
&\ri \bigl[H\p{\ell}(\xi),\tA\p{\ell}\bigr]^\circ 
= \ri \bigl[H\p{\ell}(\xi),A\otimes \one_{\cF\p{\ell}}\bigr]^\circ
+\ri \bigl[H\p{\ell}(\xi),\one_\cF\otimes\, \D\Gamma\p{\ell}(\ta)\bigr]^\circ.
\end{aligned}
\]
Furthermore, $H^\x(\xi)$ is of class $C^1(A^\x)$ and
\begin{align*}
\ri \bigl[H^\x(\xi),A^\x\bigr]^\circ & = [H(\xi),A]^\circ 
\oplus\Bigl\{\bigoplus_{\ell=1}^\infty \ri \bigl[H\p{\ell}(\xi),A\p{\ell}\bigr]^\circ\Bigr\}\\
& = \D\Gamma^\x(v\cdot\nabla\omega) - \D\Gamma^\x(v)\cdot\nabla\Omega\big(\xi-\D\Gamma^\x(k)\big) - \phi(\ri a \coup)\otimes\one_\cF.
\end{align*}
\end{proposition}

\begin{remark}
For the purpose of the proof below we abbreviate $k\p{\ell}
=k_1+\cdots+ k_\ell$, for vectors $\uk =
(k_1,\dots,k_\ell)\in\RR^{\ell\nu}$.
Note that $[H^{(\ell)}(\xi;\uk),A]^\circ = [H(\xi-k\p{\ell}),A]^\circ$ can be computed using
Proposition~\ref{Prop:C2}~\ref{Item:HisC1}.
\remarkQED\end{remark}

\begin{proof} We only prove that $H^{(\ell)}(\xi)$ is of class
  $C^1(A\otimes\one_{\cF\p{\ell}})$ and of class $C^1(\one_\cF\otimes\,\D\Gamma\p{\ell}(\ta))$. 
The $C^1(\tA\p{\ell})$ property then follows from Lemma~\ref{Lemma:SumsOfAs} 
and that $H\ext(\xi)$ is of class $C^1(A\ext)$ follows from 
Lemma~\ref{Lemma:SumsOfHs} after choosing $\tv = v$. The expressions can subsequently be easily
confirmed by computations on a suitable core for $H^{(\ell)}(\xi)$. 

Let $\ell\in\{1,2,3,\dotsc\}$ and $\xi\in \RR^\nu$.
We begin by showing that $H\p{\ell}(\xi)$ is of class $C^1(A\otimes\one_{\cF\p{\ell}})$,
where we identify $\cF \otimes  \cF\p{\ell}$ with $\symL^2(\RR^{\ell\nu};\cF)$.

Let
\[
\cC\p{\ell} = \bigset{f\in C_{0,\rmsym}(\RR^{\ell\nu};\cF) }{ \forall \uk\in \RR^{\ell\nu}: \ f(\uk)\in \cC }.
\]
Here $C_{0,\rmsym}(\RR^{\ell\nu};\cF)$ denotes the continuous and
compactly supported  $\cF$-valued functions,
symmetric under permutation of the $\ell$ variables. 
Clearly $\cC\p{\ell}$ is a core for $A\otimes\one_{\cF\p{\ell}}$. Pick a $\lambda< \Sigma_0$, cf. \eqref{Sigma0}.
Since
\[
(H\p{\ell}(\xi)-\lambda)^{-1} = \int^\oplus_{\RR^{\ell\nu}} (H\p{\ell}(\xi;\uk) - \lambda)^{-1} \D k,
\]
we observe that for $f\in \cC\p{\ell}$ we have
\[
\bigl((H\p{\ell}(\xi)-\lambda)^{-1}f \bigr)(\uk) = 
\bigl(H(\xi-k\p{\ell})+\sum_{j=1}^\ell \omega(k_j)-\lambda\bigr)^{-1} f(\uk).
\]
Hence by Lemma~\ref{lemma:C1equiv}~\ref{item:C1charextend} we conclude that
\begin{equation}\label{core-l-tilde}
\begin{aligned}
 (H\p{\ell}(\xi)-\lambda)^{-1}\cC\p{\ell} & \subset \bigset{f\in
   C_{0,\rmsym}(\RR^{\ell\nu};\cF)}{ \forall \uk\in \RR^{\ell\nu}: \
   f(\uk)\in \cD(A)\cap \cD }\\ & =: \tcC\p{\ell}.
\end{aligned}
\end{equation}
For $f,g\in \tcC\p{\ell}$ we compute using Proposition~\ref{Prop:C2}~\ref{Item:HisC1}
\begin{align*}
\big\la f,\bigl[H^{(\ell)}(\xi),A\otimes\one_{\cF^{(\ell)}}\bigr]g\big\ra 
&= \int_{\RR^{\ell\nu}} \big\la f(\uk), \bigl[H(\xi-k\p{\ell})+\sum_{j=1}^\ell \omega(k_j),A\bigr]g(\uk)\big\ra\, \D k\\
& = \int_{\RR^{\ell\nu}} \big\la f(\uk), \bigl[H(\xi-k\p{\ell}),A\bigr]^\circ g(\uk)\big\ra\,\D k.
\end{align*}
Since
\begin{equation}\label{FiniteMs}
\begin{aligned}
& M_1 := \sup_{\eta\in\RR^\nu}\bigl\|(H(\eta)-\lambda)^{-1}[H(\eta),A]^\circ(H(\eta)-\lambda)^{-1}\bigr\| <\infty,\\
& M_2 := \sup_{\xi\in\RR^\nu,\uk\in\RR^{\ell\nu}} 
\bigl\|(H(\xi-k\p{\ell})-\lambda)(H\p{\ell}(\xi;\uk)-\lambda)^{-1}\bigr\| < \infty,
\end{aligned}
\end{equation}
we can finally estimate for $f,g\in\cC\p{\ell}$
\begin{align*}
& \bigl|\big\la f,\bigl[(H\p{\ell}(\xi)-\lambda)^{-1},A \otimes \one_{\cF\p{\ell}}\bigr]g\big\ra\bigr| \\
& \quad \leq \int_{\RR^{\ell\nu}}\bigl|\big\la
(H\p{\ell}(\xi;\uk)-\lambda)^{-1}f(\uk),
[H(\xi-k\p{\ell}),A]^\circ(H\p{\ell}(\xi;\uk)-\lambda)^{-1}g(\uk)\big\ra\bigr| \D k\\
& \quad \leq M_1 M_2^2 \|f\|\|g\|.
\end{align*}
That $\cC\p{\ell}$ is a core for $A\otimes \one_{\cF\p{\ell}}$ now
implies that $H\p{\ell}(\xi)$ is of class $C^1(A \otimes \one_{\cF\p{\ell}})$.

By Lemmata~\ref{Lemma:SumsOfAs} and~\ref{Lemma:SumsOfHs} it now suffices to 
show that $H\p{\ell}(\xi)$ is of class $C^1(\one_\cF\otimes \,\D\Gamma\p{\ell}(\ta))$.
Denote by $\tw_t$ the group $\e^{\ri t\ta}$ generated by $\ta$. 
Then $\tw_t\p{\ell} = \one_\cF\otimes \Gamma\p{\ell}(\tw_t)$ is the group
generated by $\one_\cF\otimes\,\D\Gamma\p{\ell}(\ta)$. 
If we denote by $\tpsi_t$ the globally defined flow generated by the ODE $\dot{\psi}_t = \tv(\psi_t)$
we can write $(\tw_t f)(y) = \sqrt{\tJ_t} f(\tpsi_t(y))$, where $\tJ_t$ is the Jacobi determinant. See \eqref{OneBodyGroup} and \eqref{Jacobiant}.

We introduce a bit of notation. Given $\uk\in\RR^{\ell\nu}$ we write
$\tpsi_t\p{\ell}(\uk) = \tpsi_t(k_1) + \cdots + \tpsi_t(k_\ell)$.
We compute as a form on $\tcC\p{\ell}$, cf. \eqref{core-l-tilde},
\begin{align}\label{CommWithEvol}
\nonumber &\tw\p{\ell}_{-t}\bigl[H\p{\ell}(\xi),\tw\p{\ell}_t\bigr] \\
\nonumber&=
\tw\p{\ell}_{-t} \Bigl(\Omega\bigl(\xi-k\p{\ell}-\D\Gamma(k)\bigr) +\sum_{j=1}^\ell \omega(k_j)\Bigr)\tw\p{\ell}_t 
 - \Omega\bigl(\xi-k\p{\ell}-\D\Gamma(k)\bigr) - \sum_{j=1}^\ell \omega(k_j) \\
\nonumber & = \Omega\bigl(\xi-
\tpsi_t\p{\ell}(\uk)-\D\Gamma(k)\bigr)-\Omega\bigl(\xi-k\p{\ell}-\D\Gamma(k)\bigr) 
+\sum_{j=1}^\ell\bigl(\omega\bigl(\tpsi_t(k_j)\bigr)-\omega(k_j)\bigr) \\
\nonumber & = -\int_0^t
\nabla\Omega\bigl(\xi-\psi_t\p{\ell}(\uk)-\D\Gamma(k)\bigr)\cdot
\sum_{j=1}^\ell \tv\bigl(\tpsi_s(k)\bigr)\,\D s\\
& \qquad  + \sum_{j=1}^\ell\int_0^t \nabla\omega\bigl(\tpsi_s(k_j)\bigr)\cdot\tv\bigl(\tpsi_s(k_j)\bigr)\,\D s.
\end{align}
Estimate, as a fiber operator pointwise in $\uk$,
\begin{align*}
& \Bigl\|\nabla\Omega\bigl(\xi-\psi_s\p{\ell}(\uk)-\D\Gamma(k)\bigr)
\bigl(\Omega\bigl(\xi-k\p{\ell}-\D\Gamma(k)\bigr)+1\bigr)^{-\frac12}\Bigr\|\\
& \qquad  \leq C \sup_{\eta\in\RR^\nu}\frac{\big\la \xi-\psi_s\p{\ell}(\uk)-\eta \big\ra^{s_\Omega-1}}{\big\la \xi-k\p{\ell}-\eta \big\ra^{s_\Omega/2}}\leq (1+|s|)\tC
\end{align*}
uniformly in $s$ and $k$. Here we used \eqref{FinitePropSpeed} in the last step.
Appealing to
\eqref{FinitePropSpeed} again and the $C^1$-condition, cf. Condition~\ref{Ck-Condition}, we observe that
the right-hand side of \eqref{CommWithEvol} is $(H(\xi-k\p{\ell})-\lambda)^{1/2}$-bounded
uniformly in $\uk$.
From this observation it is now clear that as a form on $\cC^{(\ell)}$
\[
\frac1{t}\bigl[(H\p{\ell}(\xi)-\lambda)^{-1},\tw\p{\ell}_t\bigr]
=
-(H\p{\ell}(\xi)-\lambda)^{-1}\Bigl\{\frac1{t}\bigl[H\p{\ell}(\xi),\tw\p{\ell}_t\bigr] 
(H\p{\ell}(\xi)-\lambda)^{-1}\Bigr\},
\]
and the term in brackets extends to a bounded operator uniformly bounded in $0<|t|\leq 1$. Cf. \eqref{FiniteMs}.
It thus follows from Lemma~\ref{lemma:C1equiv}~\ref{item:C1charlim}
that $H\p{\ell}(\xi)$ is of class $C^1(\one_\cF\otimes \,\D\Gamma\p{\ell}(\ta))$.
\end{proof}

We end this subsection by formulating and proving a virial theorem which will be used to
extract the effective free dynamics induced by mass shells. Similar virial theorems were used in \cite{DRFP,GeNi1}. In the following
$\cU\subset \RR^m$ is open and $\cH$ a complex separable Hilbert space, with dense subspace $\cD_0$.
Suppose  $\{H(x)\}_{x\in \cU}$ is a family
of N-measurable operators, essentially self-adjoint on $\cD_0$,
cf. Appendix~\ref{App-Fiber}.
Then $H := \int_\cU^\oplus H(x) \D x$, a priori defined on
 $\set{f\in L^2(\cU;\cH)}{ x\in \cU: f(x)\in \cD_0 \ \textup{a.e.}}$, is essentially self-adjoint.
Let $A$ be a self-adjoint operator on $\cH$, and  $a = \one\otimes
\frac12\{v\cdot \ri\nabla_x + \ri\nabla_x\cdot v\}$, 
with $v\in C^\infty_0(\cU;\RR^m)$.
Then $A\p{1} = A\otimes \one_{L^2(\cU)} + \one_\cH\otimes \, a$ 
is self-adjoint as an operator on $\cH\otimes L^2(\cU)$, which we identify as usual with
$L^2(\cU;\cH)$.

\begin{theorem}\label{VirialTheorem} Let $E\in C^1(\cU)$, with $E(x) \in \spp(H(x))$ for all $x\in \cU$.
Suppose $H$ is of class
  $C^1(A\p{1})$ and that the commutator fibers, i.e. 
$\ri[H,A\p{1}]^\circ = \int_\cU^\oplus \ri [H,A\p{1}]^\circ(x)\D x$. 
Then for almost every $x\in \cU$
\begin{equation}\label{VirialFormula}
\one_{\{E(x)\}}(H(x)) \ri \bigl[H,A\p{1}\bigr]^\circ(x)\one_{\{E(x)\}}(H(x)) = v(x)\cdot \nabla E(x) \one_{\{E(x)\}}(H(x)).
\end{equation}
\end{theorem}

\begin{remark} By the assumption $\ri[H,A\p{1}]^\circ$ being fibered is meant the existence
of a family of operators $x\to \ri [H,A\p{1}]^\circ(x)\in \cB(\cH_{x,1};\cH_{x,-1})$, with
\[
\cU\ni x\to B(x) = (H(x)-\ri)^{-1} \ri \bigl[H,A\p{1}\bigr]^\circ(x) (H(x)-\ri)^{-1}
\]
 weakly measurable,
and $(H-\ri)^{-1} \ri[H,A\p{1}]^\circ (H-\ri)^{-1} = \int^\oplus_\cU
B(x)\D x$. Note that it follows from the discussion in
Appendix~\ref{App-Fiber} that both sides of \eqref{VirialFormula} are
weakly measurable.
\remarkQED\end{remark}

\begin{proof} Let $\psi,\tpsi\in C_0^\infty(\cU)$ with $\psi\tpsi = \psi$
and observe, cf. Lemma~\ref{Lemma:CommutatorFromGroup}, that in the sense of forms on $\cD(H)$
we have
\[
\ri \bigl[H,A\p{1}\bigr]^\circ = \lim_{t\to 0} \frac1{t} [H,W_t],
\]
with $W_t = \e^{\ri t A\p{1}} = \e^{\ri tA}\otimes \e^{\ri ta}$. Abbreviate
\[
P_\psi := \int_\cU^\oplus \psi(x)\one_{\{E(x)\}}(H(x))\,\D x.
\]
Note that $x\to \one_{\{E(x)\}}(H(x))$ is weakly measurable, and hence
strongly measurable, cf. Appendix~\ref{App-Fiber}.
Since $P_\psi$ preserves -- in fact has range inside -- $\cD(H)$ we can compute as a form on $\cD(H)$
\[
P_\psi  [H, W_t]  P_\psi = P_\psi [E,W_t] P_\psi = P_\psi [\tpsi E,W_t]P_\psi,
\]
where $E$ and $\tpsi E$ should be read as multiplication operators in
the base space, 
or equivalently as $\one_\cH\otimes E$ and $\one_\cH\otimes \tpsi E$.
Since $\tpsi E\in C^1_0(\cU)$ we clearly have  $\tpsi E\in
C^1(A\p{1})$  
with $\ri[\tpsi E,A\p{1}]^\circ = \ri[\tpsi E,a]^\circ = v\cdot\nabla (\tpsi E)$. Hence
\begin{align*}
P_\psi \ri \bigl[H,A\p{1}\bigr]^\circ P_\psi & = 
\lim_{t\to 0} \frac1{t}  P_\psi [H,W_t] P_\psi =  
\lim_{t\to 0} \frac1{t}  P_\psi [\tpsi E,W_t] P_\psi \\
& = v\cdot\nabla (\tpsi E) P_\psi^2 = v\cdot\nabla E P_\psi^2.
\end{align*}

We conclude the theorem since $\psi$ was arbitrary and when fibered the above identity reads
\[
\begin{aligned}
& \int^\oplus_{\cU} \psi(x)^2 \one_{\{E(x)\}}(H(x)) \ri \bigl[H,A\p{1}\bigr]^\circ(x)\one_{\{E(x)\}}(H(x))\, \D x \\
&\quad = \int^\oplus_{\cU} \psi(x)^2 v(x)\cdot\nabla E(x) \one_{\{E(x)\}}(H(x))\, \D x.
\end{aligned}
\]
\end{proof}

\section{The Commutator Estimate} \label{Chap-Mourre}

In this section we analyze the geometry of the threshold set $\thr\p{1}$,
construct vector fields $v\in C_0^\infty(\RR^\nu)$ going into the one-body conjugate operator $a$, cf. \eqref{a-GeneralForm},
and finally prove a Mourre estimate for the fiber Hamiltonians below
the two-boson threshold $\Spectrum_0\p{2}$ and away from threshold
energies (and the set $\Exc$).

We remark that in the literature, this type of analysis
\cite{Am,DGe1,DGe2,FGSch2,MAHP} 
has made essential use of the property $\omega(k)\to\infty$, $|k|\to\infty$,
something we do not want to assume here in view of the polaron model. In \cite{MRMP} this assumption was avoided,
by instead using that for bounded $\omega$ the gap between the ground state energy $\Spectrum_0(\xi)$
and the bottom of the essential spectrum $\Spectrum_\ess(\xi)$ closes at large total momentum. More precisely, under Condition~\ref{Cond:MC},
and the additional assumption $\sup_{k\in\RR^\nu} \omega(k)<\infty$, the second case in \ref{Item:NoHolesInSpec}, we have
\begin{equation}\label{ClosingGap}
\lim_{|\xi|\to\infty} \Spectrum_\ess(\xi) - \Spectrum_0(\xi) = 0.
\end{equation}
We refer the reader to \cite{MRMP} for a proof.
This result is crucial for treating the polaron model, and its importance is encoded in Lemma~\ref{Lemma-CompactKandSigma} below.

\subsection{Structure of the Threshold Set}\label{Sec-Thr-Structure}

Recall from \eqref{EM-Region} the notation $\cE\p{1}$ for the
energy-momentum region between the $1$- and $2$-boson thresholds.

\begin{lemma}\label{Lemma-CompactKandSigma} Assume Condition~\ref{Cond:MC}.
Let $C\subset \cE\p{1}$ be a compact set and $\cK\subset \RR^\nu$. The following holds
\begin{enumerate}[label=\textup{(\roman*)},ref=(\roman*)]
\item\label{Item-CompactKC} $\cK_C := \set{k\in\RR^\nu}{\exists (p,e)\in\Spectrum_\iso,
 \ \textup{s.t.} \  (p+k,e+\omega(k))\in C}$ is compact.
 \item\label{Item-CompactSC} If $\cK\cap\cK_C$ is closed, then the set
$\Spectrum_C(\cK)  = \bigset{(p,e)\in\Spectrum_\iso}{\exists k\in \cK \ \textup{s.t.} \
(p+k,e+\omega(k))\in C}$ is compact.
\end{enumerate}
\end{lemma}

\begin{remark} Observe that $\Spectrum_C(\cK)=\Spectrum_C(\cK_C\cap \cK)$.
In particular we abbreviate $\Spectrum_C := \Spectrum_C(\RR^\nu) = \Spectrum_C(\cK_C)$.

The set $\cK_C$ consists of asymptotic momenta available to states
localized in $C$ for one-boson emission, due to energy and
momentum conservation. The set $\Spectrum_C$ ($\Spectrum_C(\cK)$) contains the interacting bound states
reachable from states localized in $C$ after emission of one boson 
(with asymptotic momentum in $\cK$).
\remarkQED\end{remark}

\begin{proof} We divide the proof into three steps.

\noindent\emph{Step I:} Reducing the problem to compactness of a single set.
Let $X = \RR^\nu\times\Spectrum_\iso$. 
Define a map
$\Psi\colon X\to \RR^{\nu+1}$ by 
\[
\Psi(k,p,e) = (p+k,e+\omega(k)).
\]
Denote by $\Pi_1\colon X\to\RR^\nu$ the projection onto the $k$
coordinate and
by $\Pi_2\colon X\to \RR^{\nu+1}$ the projection onto the $(p,e)$
coordinate. With this notation we can write $\cK_C =
\Pi_1(\Psi^{-1}(C))$ and $\Sigma_C(\cK) = \Pi_2(\Psi^{-1}(C)\cap
(\cK\times\RR^{\nu+1}))$.
Hence it suffices to prove that $C'= \Psi^{-1}(C)$ is a compact subset
of $\RR^{2\nu+1}$. 

\noindent\emph{Step II:} There exists $\epsilon>0$ such that 
$C'\subset \RR^\nu\times\set{(p,e)\in\Sigma_\iso}{e\leq \Spectrum_\ess(p)-\epsilon}$.
Indeed, let $\epsilon= d(C,\Spectrum\p{2}_0)>0$, the distance from $C$ to
the two-boson threshold. Suppose $(k,p,e)\in C'$ satisfies that $e\in
(\Spectrum_\ess(p)-\epsilon,\Spectrum_\ess(p))$.
Then 
\[
e+\omega(k) > \Spectrum_\ess(p)+\omega(k) -\epsilon =
\Spectrum\p{1}_0((p+k)-k)+\omega(k)-\epsilon\geq \Spectrum\p{2}_0(p+k)-\epsilon.
\]
This contradicts the choice of $\epsilon$, since $(p+k,e+\omega(k))\in
C$.

\noindent\emph{Step III:} $C'$ is compact. Since $\Psi$ is continuous, the preimage $C'$ is closed
as a subset of $X$. By Step II, it is in fact closed as a subset of
$\RR^{2\nu+1}$ as well. It remains to argue that $C'$ is bounded.

Assume $C'$ is unbounded. Then there must exist a sequence
$(k_n,p_n,e_n)\in C'$ with $|k_n|+|p_n| \to \infty$.  Since $(p_n+k_n,e_n+\omega(k_n))$
is in the compact set $C$, $p_n+k_n$ is a bounded sequence. Hence
$|k_n|$ and $|p_n|$ both diverge to $\infty$. Let $\epsilon =
d(C,\Spectrum\p{2})$ as in Step II.
By \eqref{ClosingGap} there exists $N$ such that for $n\geq N$ we have
$\Spectrum_0(p_n) > \Spectrum\p{1}_0(p_n) - \epsilon$.
For $n\geq N$ we can now estimate
\begin{align*}
e_n +\omega(k_n) & \geq \Spectrum_0(p_n)+\omega(k_n) >
\Spectrum\p{1}_0(p_n) +\omega(k_n)-\epsilon\\
& = \Spectrum\p{1}_0((p_n+k_n)-k_n)+\omega(k_n)-\epsilon \geq \Spectrum\p{2}_0(p_n+k_n)-\epsilon.
\end{align*}
This contradicts the choice of $\epsilon$ and we are done.
\end{proof}

The remainder of this section is devoted to the geometry of the
threshold sets $\thr\p{1}_\shells$, $\thr\p{1}_\parallel$ and
$\thr\p{1}_{\nparallel}$, cf. \eqref{LC-Shells}, \eqref{LC-Parallel}
and \eqref{LC-Angular}.

\begin{lemma}\label{Lemma-ParaAndNonPara} Assume Conditions~\ref{Cond:MC}
  and~\ref{Cond:MT}, with $n_0=0$. We have the following two properties
\begin{enumerate}[label=\textup{(\roman*)},ref=(\roman*)]
\item\label{Item-CrossingsSmall} The sets $\thr\p{1}_\parallel(\xi)\cap\cE\p{1}(\xi)$ and $\thr\p{1}_\nparallel(\xi)\cap\cE\p{1}(\xi)$
are locally finite, with possible accumulation points only at the upper
boundary $\Spectrum_0\p{2}(\xi)$,
the $2$-boson threshold.
\item\label{Item-ZeroesSmall}
The sets $\thr\p{1}_\parallel\cap\cE\p{1}$ and $\thr\p{1}_\nparallel\cap\cE\p{1}$
are (relatively) closed subsets of $\cE^{(1)}$.
\end{enumerate}
\end{lemma}

\begin{remark}
The set $\thrC$ is precisely the union of radial graphs of $\omega$ centered above each crossing point,
i.e. union of the graphs $\set{(p+ru,e+\omega(p+ru))}{r\in\RR}$ for each
$(p,e)\in\cross$. Here $u$ is a unit vector collinear with $p$.
\remarkQED\end{remark}

\begin{proof}
We begin with \ref{Item-CrossingsSmall} and take first the set $\thrC(\xi)$.
 Fix $\xi\in\RR^\nu$ and a matching collinear unit vector $u$. Let
 $r\in \RR$ be such that $\xi = r u$. 

Suppose
$\{E_n\}\subset\thrC(\xi)$, $E_n<\Spectrum\p{2}_0(\xi)$, with $E_n\to
E<\Spectrum\p{2}_0(\xi)$.
We need to argue that the sequence $\{E_n\}$ is eventually constant.
Let $E' = (E+\Spectrum\p{2}_0(\xi))/2$, such that
\begin{equation}\label{ChoiceOfC}
C=\{\xi\}\times
\bigl[\Spectrum\p{1}_0(\xi),E'\bigr]
\end{equation}
is a compact subset of $\cE\p{1}$.
For $n$ large enough we have $(\xi,E_n)\in C$.

There exists $r_n\in\RR$, for each $n\in\NN$,  such that
$(\xi-r_n u,E_n-\omega(r_n u))\in\cross$ for
all $n$. Observe that $(\xi-r_n u,E_n-\omega(r_n u))\in\cross \cap
\Spectrum_C$  for large $n$.
Since $\Spectrum_C$ is compact, cf. Lemma~\ref{Lemma-CompactKandSigma},
and the set $\cross$ consists of isolated $S^{\nu-1}$-spheres
centered at $\xi=0$, we conclude that $|\xi-r_nu| = |r-r_n||u|$, and
hence also
$r_n$, only take finitely many values. But then  $E_n-\omega(r_n u)$ must also take only
finitely many values and hence $E_n$ is eventually constantly equal to
$E$.

As for the set $\thrN(\xi)$ we assume again that
$\{E_n\}\subset \thrN(\xi)$, $E_n<\Spectrum\p{2}_0(\xi)$, and $E_n\to E<
\Spectrum_0\p{2}(\xi)$.
There exists $k_n$, for each $n$,
such that $(\xi-k_n,E_n-\omega(k_n))\in\cross$ and
$\nabla\omega(k_n)=0$.

Let $\epsilon>0$ and the compact subset $C$ of $\cE\p{1}$  be as
before, cf. \eqref{ChoiceOfC}.
The sequence $\{(\xi-k_n,E_n-\omega(k_n))\}$ must, for $n$ large, again
run inside the compact set $\Spectrum_C$ and thus since $\cross$ consists
of isolated spheres, we must have a subsequence $\{k_{n_j}\}$ such that
$|\xi-k_{n_j}| = R_c$ and $E_{n_j}-\omega(k_{n_j})=E_c$ are constant, signifying
that we are on the same level crossing $\partial B(0;R_c)\times\{E_c\}\subset\cross$.
If $\xi=0$ or $\omega$ is constant, we
are done since in either case
$\omega(k_{n_j})$  is a constant sequence, and
hence $E_{n_j} = E$.

If we are in dimension $\nu=1$  we are also done, since this will
force $k_{n_j}$ to only attain the two values $\xi+R_c$ and $\xi-R_c$.
Hence $E_{n_j} = E$, for $j$ sufficiently large.

We can thus assume that $\nu\geq 2$ and $\omega$ is not a
constant function. Fix another unit
vector $v$, with $v\cdot u =0$. By symmetry in the hyperspace
orthogonal to $u$, we can assume that $k_{n_j}\in\huelle\{u,v\}$.
Using that $|\xi-k_{n_j}|=R_c$, we can write the momenta as
\[
k_{n_j} =  k(\theta_j)=\xi - R_c\bigl(\cos(\theta_j) u +\sin(\theta_j)v\bigr),
\]
with $\theta_j\in \RR$ a bounded sequence. But since $\theta \to
|\nabla \omega(k(\theta))|^2$ is a non-zero real analytic function
the sequence $\theta_j$ can only attain finitely many values.
Again we conclude that $E_{n_j}$ can only attain finitely many values,
and hence must be constantly equal to $E$ for $j$ large.

As for \ref{Item-ZeroesSmall}, let  $\{(\xi_n,E_n)\}_{n\in\NN}\subset
\thrC\cap\cE\p{1}$ be a convergent sequence with
$(\xi_n,E_n)\to(\xi,E)\in\cE\p{1}$. Let $0<R< d((\xi,E),\Spectrum\p{2}_0)$, i.e. $R$ is chosen
smaller than the distance from $(\xi,E)$ to the upper boundary of
$\cE\p{1}$. With this choice
\begin{equation}\label{CompactC}
C = \bigl(\overline{B(\xi;R)}\times [E-R,E+R]\bigr)\cap\cE\p{1}
\end{equation}
is a compact subset of $\cE\p{1}$.
For $n$ large enough we have $(\xi_n,E_n)\in C$.
 By rotational symmetry we can assume that $\xi$ and all the $\xi_n$'s
are collinear with a unit vector $u$. Write $\xi = ru$ and $\xi_n= r_n u$.
 There exist a sequence of momenta $s_n u$, with $s_n u\in\cK_C$ for $n$
 large enough, such that
$((r_n-s_n)u,E_n-\omega(s_n u))\in \Spectrum_C\cap \cross$, for $n$ large
enough. By compactness of $\cK_C$  we can
extract a convergent subsequence $s_{n_j}$ converging to $s\in\RR$.
Then
\[
\bigl(\xi -su, E-\omega(su)\bigr) = \lim_{j\to\infty} 
\bigl((r_{n_j}-s_{n_j})u,E_{n_j}-\omega(s_{n_j} u)\bigr)\in\Spectrum_C\cap\cross,
\]
since the set on the right-hand side is closed. Hence $(\xi,E)\in
\thrC$,
which implies that $\thrC$ is closed as a subset of $\cE\p{1}$.

We now turn to $\thrN$. We again take a sequence
$\{(\xi_n,E_n)\}_{n\in\NN}\subset\thrN$ converging to
$(\xi,E)\in\cE\p{1}$. As above we can assume
that there exists a unit vector $u$ such that $\xi = ru$ and $\xi_n =
r_n u$, with $r_n\to r$.

Since $(r_n u,E_n)\in\thrN$ there must exist
$k_n\in\RR^\nu$ such that $(r_nu-k_n,E_n-\omega(k_n))\in\cross$ and
$\nabla\omega(k_n)=0$. By the, by now, standard argument, there exists a
convergent subsequence $\{k_{n_j}\}$. Denote by $k$ its limit.

We can now argue as for $\thrC$ that
\[
\bigl(\xi -k, E-\omega(k)\bigr) = \lim_{j\to\infty} \bigl(r_{n_j}u-k_{n_j},E_{n_j}-\omega(k_{n_j})\bigr)\in\cross
\]
and $\nabla\omega(k)=\lim_{j\to\infty}\nabla\omega(k_{n_j})=0$. Hence
$(\xi,E)\in\thrN$, which establishes the remaining part of \ref{Item-ZeroesSmall}.
\end{proof}

\begin{proof}[Proof of Theorem~\ref{Thm-thr}]
Abbreviate for the purpose of this proof  
\begin{equation}\label{ThresholdShellOnly}
\begin{aligned}
&\cU  = \bigset{(\xi,E)\in\thr\p{1}_\shells}{(\xi,E)\not\in \thr\p{1}_\parallel\cup\thr\p{1}_\nparallel}\\
&\cU(\xi) = \bigset{E\in\RR}{ (\xi,E)\in\cU}.
\end{aligned}
\end{equation}
The sets $\cU$ and $\cU(\xi)$ are subsets of $\thr\p{1}_\shells$ and
$\thr\p{1}_\shells(\xi)$ respectively.

Given Lemma~\ref{Lemma-ParaAndNonPara} it remains to prove the following two statements:
\begin{align}
\label{Thr-Remain}
& \overline{\cU}\cap \cE\p{1} \subset \thr\p{1}\\
\label{Thr-Remain-Fiber}
& \cU(\xi) \cap \cE\p{1}(\xi)\quad \textup{is locally finite.}
\end{align}

To prove \eqref{Thr-Remain}, let $(\xi_n,E_n)\in\cU\cap \cE\p{1}$ and
assume
$(\xi_n,E_n)\to (\xi,E)$ with $E<\Spectrum\p{2}_0(\xi)$. 
We need to argue that $(\xi,E)\in \thr\p{1}$.

Construct a compact set $C\subset\cE\p{1}$ containing $(\xi,E)$ as in \eqref{CompactC}.
For $n$ large enough we have $(\xi_n,E_n)\in C$. For each (large) $n$
we can find a $k_n\in\cK_C$, a mass shell
$(\Annul_n,\Shell_n)\in\shells$,
such that $\xi_n-k_n\in\Annul_n$, $E_n = S_n(\xi_n-k_n)+\omega(k_n)$
and
$\nabla S_n(\xi_n-k_n) =\nabla \omega(k_n)$. Here we used that $(\xi_n,E_n)\in\thr\p{1}_\shells$.

Since $\cK_C$ is compact we can pass to a convergent subsequence
$\{k_{n_\ell}\}$
with $k:=\lim_{\ell\to\infty} k_{n_\ell} \in\cK_C$. 
Abbreviate 
\begin{equation}\label{Thr-p-ell}
p_\ell:= \bigl(\xi_{n_\ell}-k_{n_\ell},E_{n_\ell}-\omega(k_{n_\ell})\bigr)\in\Spectrum_C.
\end{equation}
Since $\Sigma_C$ is closed we have
\begin{equation}\label{Thr-p}
\lim_{\ell\to\infty} p_\ell = p := \big(
\xi-k,E-\omega(k)\big) \in\Spectrum_C\subset \Spectrum_\iso.
\end{equation}
Recall that level crossings, as $S^{\nu-1}$-spheres inside $\cross$, are
isolated and only finitely many mass shells emanate from each
crossing. Hence we can assume that there exists a distinguished  mass shell
$(\Annul,\Shell)\in\shells$ such that $p_\ell\in\cG_\Shell$,
cf.~\eqref{Graph-Of-Shell}, for all $\ell$.
We can furthermore assume that
we are in one of the two following cases 
\begin{equation}\label{Thr-Cases}
\begin{aligned}
&\textup{\textbf{Case A}}\qquad \forall \ell: \quad  \nabla\omega(k_{n_\ell}) \neq 0\\
&\textup{\textbf{Case B}} \qquad \forall \ell:\quad \nabla\omega(k_{n_\ell}) =0.
\end{aligned}
\end{equation}
In Case A we must have for each $\ell$ an $s_\ell\in\RR$ such that
$k_{n_\ell} = s_\ell u$ and  $\lim_{\ell\to\infty} s_\ell$ exists.
If $k\to\Shell\p{1}(\xi;k) = \Shell(\xi-k)+\omega(k)$ is not a constant function, the sequence $s_\ell$
must be eventually constant and hence $k\in \Annul+\xi$ and
$p\in\cG_\Shell$.
Here we used that $s\to \Shell\p{1}(\xi;su)$ continues
analytically through level crossings.
If on the other hand $k\to\Shell\p{1}(\xi;k)$ is a constant function, we can replace the $k_{n_\ell}$'s
by a constant $k\in\Annul+\xi$. Hence the new limit will satisfy $p\in\cG_\Shell$.
 
In Case B we have  $\nabla S(\xi_{n_\ell}-k_{n_\ell}) =
0$, so we must have either $\Shell$
constant, or $|\xi_{n_\ell}-k_{n_\ell}|$ eventually constant and equal
to $|\xi-k|$.
In the latter case $k\in\Annul+\xi$ and $p\in\cG_\Shell$.
We now assume that $\Shell$ is a constant function.

If $\omega$ is also constant we can redefine the $k_{n_\ell}$'s as
above and again arrive at $p\in\cG_\Shell$. If $\omega$ is  not a
constant,
$|k_{n_\ell}|$ is eventually constant and equal to $r\geq 0$.
First of all we observe that $r$ is strictly smaller than the outer
radius of $\Annul$. This is due to the choice of
$(\xi_{n_\ell},E_{n_\ell})$ away from $\thr\p{1}_{\nparallel}$, cf. \eqref{LC-Angular}.
We can thus replace the $k_{n_\ell}$'s with possibly different
$k_{n_\ell}$'s in $rS^{\nu-1}$ such that the limit
$k\in\Annul+\xi$.

Summing up, we have argued that either $p\in\cG_\Shell$, or we can
make a different choice of sequence $k_{n_\ell}$ such that $p$
ends up inside $\cG_\Shell$.
Then, by continuity, 
we must have $E=S(\xi-k)+\omega(k)$ and  $\nabla S(\xi-k)=\nabla\omega(k)$.
Hence $(\xi,E)\in\thr\p{1}_\shells$. This proves \eqref{Thr-Remain}.

To verify \eqref{Thr-Remain-Fiber}, let $(\xi,E_n)\in \cU(\xi)$, with $E_n<\Spectrum\p{2}_0(\xi)$, such that
$E_n\to E\in\cU(\xi)$, with  $E< \Spectrum\p{2}_0(\xi)$.
We have to prove that the sequence $E_n$ is eventually constant.
Assume towards a contradiction that it is not eventually
constant. Hence we can assume, possibly passing to a subsequence, that it is
strictly monotone.

Let $C\subset\cE\p{1}$ compact, be as in \eqref{ChoiceOfC}.
For $n$ sufficiently large we
have $(\xi,E_n)\in C$. By the choice of $E_n$ we can to each $n$
identify a $k_n\in \cK_C$  and a mass shell
$(\Annul_n,\Shell_n)\in\shells$
such that $\xi-k_n\in\Annul_n$,
$E_n = S_n(\xi-k_n)+\omega(k_n)$ and $\nabla S_n(\xi-k_n) =
\nabla\omega(k_n)$.

As in the verification of \eqref{Thr-Remain} we can extract a
subsequence $\{E_{n_\ell}\}$ together with a convergent sequence of momenta
$\{k_{n_\ell}\}$, and a distinguished mass shell
$(\Annul,\Shell)\in\shells$
such that $\xi-k,\xi-k_{n_\ell}\in\Annul$ and $p_\ell,p\in\cG_\Shell$,
cf. \eqref{Thr-p-ell} and \eqref{Thr-p}. Here
$k=\lim_{\ell\to\infty} k_{n_\ell}$. We can furthermore assume that we
are in either Case A or Case B, cf.~\eqref{Thr-Cases}.
Here we used that $E\in\cU(\xi)$, cf.~\eqref{ThresholdShellOnly}, to
rule out the possibility that $p\in\cross$.

In Case A we reach a contradiction with $E_{n_\ell}$ being strictly
monotone as follows. Write $\xi = ru$ for some unit vector $u$ and
$r\in\RR$. If $\xi\neq 0$ the demand that $\nabla S(\xi-k_{n_\ell})
=\nabla\omega(k_{n_\ell})$, together with rotation invariance, forces
all the $k_{n_\ell}$'s to be collinear with $u$. If $\xi=0$, we can
again use rotation invariance and
simply replace all the $k_{n_\ell}$'s by $|k_{n_\ell}|u$ and thus arrive at the
same situation.
Hence the map $t\to S\p{1}(\xi;tu) =
S((r-t)u)+\omega(tu)$ is analytic and vanishes along a sequence
with accumulation point inside its domain of analyticity. Hence it is
constant,
i.e. $E_{n_\ell} = S\p{1}(\xi;k_{n_\ell})$ is constant.

In Case B we reach a contradiction as follows. 
Since $\nabla\omega(k_{n_\ell}) = 0 = \nabla S(\xi-k_{n_\ell})$ we
can conclude that: Either $\omega$ is constant or $|k_{n_\ell}|$ is
eventually constant. Furthermore, either $S$ is constant or $|\xi-k_{n_\ell}|$
is eventually constant. Regardless of which of the $4$ possible
combinations we find ourselves in, we conclude again that
$E_{n_\ell} = S(\xi-k_{n_\ell})+\omega(k_{n_\ell})$ is eventually
constant.
\end{proof}

\subsection{Some Geometric Considerations}

The goal of this subsection and the next, is to analyze the set of
momenta $k$ available to boson emission, which are not collinear with
$\xi$, and for which the remaining
interacting system ends up at or near a level crossing.

In one dimension or at total momentum $\xi=0$, we can avoid this
situation completely by staying away from the threshold set $\cT\p{1}$.
For this reason the reader should, for the purpose of this subsection
and the next, think of  $\nu\geq 2$ and $\xi\neq 0$. 
Finally, the reader trying to get a feel for the basic ideas of the
construction can safely skip these two subsections on a first reading.
We remark that for the polaron model as well, for the same reason, these considerations also
do not play a role.

Given a point in energy-momentum space
\begin{equation}\label{ChoiceOfStart}
(\xi,E)\in\cE\p{1}\backslash\bigl(\thr\p{1}\cup\Exc\bigr),
\end{equation}
we  wish to be able to choose a compact interval $\cJ = [E-\delta,E+\delta]$ with
\begin{equation}\label{ChoiceOfO}
(\xi,E)\in\{\xi\}\times\cJ \subset \cE\p{1}\backslash\bigl(\thr\p{1}\cup\Exc\bigr),
\end{equation}
such that states localized in $\cJ$ (at sharp total momentum $\xi$)
can only break up into
channels with non-zero breakup velocity. The exceptional set $\Exc$ was defined in \eqref{ExcSet}.

Given $\cJ\subset \cE\p{1}(\xi)$ and $\cK\subset \RR^\nu$, we
associate the sets
\[
\cK_\cJ := \cK_{\{\xi\}\times\cJ},\quad \Spectrum_\cJ :=
\Spectrum_{\{\xi\}\times\cJ}  
\quad \textup{and} \quad  \Spectrum_\cJ(\cK) := \Spectrum_{\{\xi\}\times\cJ}(\cK).
\]
Recall from Lemma~\ref{Lemma-CompactKandSigma} the notation for the sets $\cK_C$ and
$\Spectrum_C$, for $C\subset\RR^{\nu+1}$.
The set $\cK_\cJ$ contains the momenta available for boson emission starting from
a state localized with respect to energy in $\cJ$, whereas
$\Spectrum_\cJ$ labels the available interacting bound states the system can
relax to. Observe that if $\cJ\cap(\sigma_\pp(H(\xi))+\omega(0))=\emptyset$, then $0\not\in\cK$. 

With the choice \eqref{ChoiceOfO} of $\cJ$, some of
the $k$'s in $\cK_\cJ$ may correspond to elements $(\xi-k,E-\omega(k))$  on level crossings, but only if
$\xi$ is non-zero, and then $k$ is linearly independent of $\xi$. Here
$E\in\cJ$.
We introduce the notation
\begin{equation}\label{Set-KO-cross}
\cK_\cJ^\cross := \bigset{k\in\cK_\cJ}{\exists \lambda\in\cJ: \quad \bigl(\xi-k,\lambda-\omega(k)\bigr)\in\cross}
\end{equation}
for the subset of $\cKJ{\xi}{\cJ}$ corresponding to level
crossings. In addition we write, for $\cK\subset \RR^\nu$,
\[
\Spectrum_\cJ^\cross = \Spectrum_\cJ\cap \cross \quad \textup{and} \quad \Spectrum_\cJ^\cross(\cK) = \Spectrum_\cJ(\cK)\cap \cross
\]
for the reachable interacting bound states at level crossings.

For $\xi\neq 0$, we use the notation
$\Ortho_\nu(\xi)$ for the subgroup of the orthogonal group consisting
of orthogonal matrices $O$ satisfying $O\xi=\xi$.
 It is convenient, given
$\xi\neq 0$, to introduce a change of coordinates. Let
\[
\Pi_\xi \colon \RR^{\nu-1} \to \bigset{\eta\in\RR^\nu}{\xi\cdot\eta=0}
\]
be a (linear) isometric isomorphism of $\RR^{\nu-1}$ onto the orthogonal
complement of $\xi$. We define a change of coordinates
$\bok\colon (0,\infty)\times [0,\pi]\times S^{\nu-2}\to \RR^\nu\backslash
\{\xi\}$
by
\begin{equation}\label{ChangeOfCoordinates}
\bok(s,\theta,w) := \xi - s \cos(\theta) \frac{\xi}{|\xi|}
+s\sin(\theta) \Pi_\xi(w),
\end{equation}
with the standard convention that $S^0 = \{-1,+1\}\subset \RR$.
Observe that $\Pi_\xi$ induces a group isomorphism $\Ortho_{\nu-1} \to
\Ortho_\nu(\xi)$
by mapping $O\in \Ortho_{\nu-1}$ to $O_\xi\in \Ortho_{\nu}(\xi)$,
determined by the two relations $O_\xi \xi = \xi$ and $ O_\xi \Pi_\xi =
\Pi_\xi O$.
For  $O\in \Ortho_{\nu-1}$ we have
\begin{equation}\label{OIntertwiner}
O_\xi \bok(r,\theta,w) = \bok(r,\theta,Ow).
\end{equation}
We will use the function $\bok$ defining the change of coordinates also
beyond angles confined to $[0,\pi]$. Finally, note that
\begin{equation}\label{ThetaReflected}
\bok(r,2\pi-\theta,w)=\bok(r,\theta,-w)
\end{equation}
and the points $\bok(r,0,w)$ and $\bok(r,\pi,w)$, the poles of a sphere with
$\RR \xi $ as the axis of rotation, do not depend on $w$.

\begin{lemma}\label{FinitelyManyCrossings} Assume Conditions~\ref{Cond:MC}
  and~\ref{Cond:MT}, with $n_0=0$. Suppose  $\nu\geq 2$. Let
  $(\xi,E)\in \cE\p{1}\backslash\thr\p{1}_\parallel$. 
There exists a finite number of radii
  $\{R_i\}_{i=1}^M$, with $R_i>0$,  and for each $i=1,\dots,M$, a
  finite set of angles
$\{\theta_{i,j}\}_{j=1}^{M_i}$, with $\theta_{i,j}\in(0,\pi)$,
such that
\[
\cK_{\{E\}}^\cross  = \bigcup_{i=1}^M\bigcup_{j=1}^{M_i} \bok(R_i,\theta_{i,j},S^{\nu-2}).
\]
If $\xi=0$ or $\omega$ is constant the set $\cK_{\{E\}}^\cross$ is empty.
\end{lemma}

\begin{remark} 1) The set described above is a finite union of
non-empty  $S^{\nu-2}$-spheres, all centered along a line through the origin in the
direction of $\xi$. They sit inside $S^{\nu-1}$-spheres of crossings
centered at $\xi$ with radius $R_i$.

2) In dimension $\nu=1$ the set $\cK_{\{E\}}^\cross$ is empty by the choice of
$(\xi,E)$. In dimension $2$ the set consists of finitely many points
placed symmetrically around the line through the origin and $\xi$,
with no points on the line through $0$ and $\xi$.
\remarkQED\end{remark}

\begin{proof} First we observe that if $\xi=0$
we have $\cK_{\{E\}}^\cross=\emptyset$.
This is due to the assumption that $\omega$ is rotation invariant.
For a similar reason, the set is also empty if $\omega$ is a
constant function regardless of $\xi$. From now on we assume that
$\xi\neq 0$ and that $\omega$ is not constant.

From Lemma~\ref{Lemma-CompactKandSigma} we know a priori that the sets
$\cK_{\{E\}}^\cross$ and
$\Spectrum_{\{E\}}^\cross$ are compact. In particular, there exist
finitely many radii $R_1,\dots,R_M$, and energies $\lambda_1,\dots,\lambda_M$ such that
\begin{equation}\label{SharpInsideSpheres}
\Spectrum_{\{E\}}^\cross \subset \bigcup_{i=1}^M R_i S^{\nu-1}
\times\{\lambda_i\}\subset \cross.
\end{equation}
The choice of $(\xi,E)$ ensures that $R_i>0$ for all $i=1,\dots, M$.

Clearly the set $\cK_{\{E\}}^\cross$ is invariant under rotations from
the group $\Ortho_\nu(\xi)$. Fix a unit vector $v$ orthogonal
to $\xi$. Take for example
$v = \Pi_\xi(e_1)$. Put $u=\xi/|\xi|$.

What we need to show is that $\cK = \cKJ{\xi}{\{E\}}^\cross \cap \huelle\{u,v\}$
is a finite set. The choice of $(\xi,E)$ ensures that the intersection
above does not contain any elements in $\RR u = \RR \xi$.
The orbit under $\Ortho_\nu(\xi)$ of $k\in \cKJ{\xi}{\{E\}}^\cross\cap
\huelle\{u,v\}$ are exactly the $S^{\nu-2}$-spheres in the lemma,
cf. \eqref{OIntertwiner}.

Aiming for a contradiction we assume that there exists an
infinite sequence $\{k_n\}_{n\in\NN}\subset \cK$ consisting of distinct momenta.
Observe that $(\xi-k_{n},E-\omega(k_{n}))\in
\Spectrum_{\{E\}}^\cross$, and hence by \eqref{SharpInsideSpheres} 
there must exist $1\leq i\leq M$ and a subsequence $\{k_{n_j}\}$,
with $|\xi-k_{n_j}| = R_i > 0$ and $E-\omega(k_{n_j}) = \lambda_i$ for
all $j$.
 We can now write
$k_{n_j} = \bok(R_i,\theta_j,e_1)$, for a sequence of distinct angles
$\theta_j\in[0,2\pi)$.

Observe that $E-\omega(k_{n_j})$, and consequently
$\omega(k_{n_j})$, is a constant sequence.
Since the map $\RR\ni \theta \to f(\theta) = \omega(\bok(R_i,\theta,e_1))$ is a real
analytic function, constant along a sequence $\theta_j$ that has a
cluster point, we conclude that $f$ must be a constant function.
Since $\omega$ is not a constant function, this can only happen if
$\RR\ni \theta \to |\bok(R_i,\theta,e_1)|$ is constant.
But this is impossible because we assumed that $\xi\neq 0$, cf. \eqref{ChangeOfCoordinates}.
Hence $\cK$ does not contain
a countable sequence of distinct momenta and we conclude the lemma.
Observe that \eqref{ThetaReflected} ensures that we can restrict the angles to $(0,\pi)$.
\end{proof}

Let $(\xi,E)$ be chosen as in \eqref{ChoiceOfStart}. We
construct torus neighborhoods $\TT_{i,j}$, in the $(r,\theta,w)$ coordinate
system, around the finitely many $S^{\nu-2}$-spheres in
$\cKJ{\xi}{\{E\}}^\cross$ identified in Lemma~\ref{FinitelyManyCrossings}. We can label these sets by
radii and angles $(R_i,\theta_{i,j})$, $i=1,\dots,M$ and $j=1,\dots,M_i$, with $R_i>0$ and $\theta_{i,j}\in
(0,\pi)$. We define
\begin{equation}\label{Torus-Def}
\TT_{i,j}(\epsilon_\theta,\epsilon_r):= \bigset{\bok(r,\theta,w)}{|R_i-r|<\epsilon_r,
  |\theta-\theta_{i,j}|<\epsilon_\theta, w\in S^{\nu-2}},
\end{equation}
where $\epsilon_\theta$ measures the angular thickness of the torus,
and $\epsilon_r$ the radial thickness.

In order to pick an appropriate angular and radial thickness for the tori
we proceed in steps to ensure that a number of properties are satisfied.
We first pick $0<\epsfour_r, \epstwo_\theta \leq 1$ such that
\begin{equation}\label{TorusHasHole}
\begin{aligned}
\epsfour_r & < \min_{i=1,\dots,M} R_i\\
\epstwo_\theta & <\frac12 \min_{1\leq i\leq M}\min_{1\leq j\leq M_i}\{\theta_{i,j},\pi-\theta_{i,j}\}.
\end{aligned}
\end{equation}
With this choice we have ensured that the tori will have their holes,
with a little angular room to spare. 

By the choice $E\not\in \thrN(\xi)$ we know that
$\nabla\omega(\bok(R_i,\theta_{i,j},w))\neq 0$, for every $i=1,\dots, M$, $j=1,\dots,M_i$, and $w\in S^{\nu-2}$.
In addition, by rotation invariance of $\omega$, for any $i$ and $j$ the norm $|\nabla\omega(\bok(R_i,\theta_{i,j},w))|$
does not depend on $w$.
By continuity of $\nabla\omega$, we can choose $0<\epsone_\theta\leq \epstwo_\theta$
and $0<\epsthree_r\leq \epsfour_r$ such that
\begin{equation}\label{nablaomegaPos}
\forall i,j: \quad \inf_{k\in \TT_{i,j}(\epsone_\theta,\epsthree_r)} |\nabla\omega(k)|>0.
\end{equation}
The choice of $\epsone_\theta$ and $\epsthree_r$ implies that
$\TT_{i,j}(\epsone_\theta,\epsthree_r)$ are topological tori
and they contain no $k$'s parallel with $\xi$, nor are there $k$'s
with $\nabla\omega(k)=0$.

Since $\nabla\omega$ does not vanish on the tori $\TT_{i,j}(\epsone_\theta,\epsthree_r)$,
 and $\omega$ is rotation invariant, we find that
 $k\cdot\nabla\omega(k)$ does not vanish on the tori either.
 Recall that $k=0$, being \myquote{collinear} with $\xi$, is not in any of the tori. Hence $k\cdot\nabla\omega(k)$
 has a sign, which we denote by $\sigma_{i,j}\in\{-1,+1\}$, for each $i=1,\dots,M$ and $j=1,\dots,M_i$.
We note the identity
\begin{equation}\label{SignOf-kdotnablaomega}
\forall k\in \TT_{i,j}\bigl(\epsone_\theta,\epsthree_r\bigr):\quad |\nabla\omega(k)| = \sigma_{i,j}\frac{k}{|k|}\cdot\nabla\omega(k).
\end{equation}

Unfortunately the above choice of $\epsone_\theta$ and
$\epsthree_r$  does not quite
suffice. At the center of the torus, i.e. for $k$'s in the set
$\bok(R_i,\theta_{i,j},S^{\nu-2})$, we know that
$(\xi-k,E-\omega(k))\in\cross$. In fact for such $k$  we always end at the same
level crossing
\[
\cross_{i,j} := \bigl(R_i S^{\nu-1}\bigr)\times\bigl\{E-\omega(\bok(R_i,\theta_{i,j},e_1))\bigr\},
\]
due to rotation invariance. (A $j'\neq j$ may a priori give rise to
a different
crossing $\cross_{i,j'}\neq \cross_{i,j}$ if there at different energies sit level
crossings with the same radius $R_i$.)
For other $k$'s in the torus  we need to be sure that
$(\xi-k,E-\omega(k))$ does not land on a \emph{different}
crossing. That is, we have to identify $\epsilon_\theta\leq\epsone_\theta$
and $\epstwo_r\leq\epsthree_r$ such that
\begin{equation}\label{OnlyOneCrossingFromTori}
\Spectrum_{\{E\}}^\cross \Bigl(\bTT_{i,j}\bigl(\epsilon_\theta,\epstwo_r\bigr)\Bigr) \subset \cross_{i,j},
\end{equation}
where $\bTT_{i,j}$ denotes the closure of the torus.
Here we can use that level crossings are isolated and that we only
consider finitely many tori, to ensure that
\begin{equation}\label{ChoiceOfd}
d=d(\xi,E):=\min_{i,j} d(\cross_{i,j},\cross\backslash \cross_{i,j})>0.
\end{equation}

For $k\in \bTT_{i,j}(\epsilon_\theta,\epstwo_r)$ we write
first $k = \bok(r,\theta,w)$ with $|r-R_i|\leq\epstwo_r$,
$|\theta-\theta_{i,j}|\leq\epsilon_\theta$ and $w\in S^{\nu-2}$.
Then we compute
\begin{equation}\label{Eq-DistToCross}
\begin{pmatrix} \xi-k \\ E-\omega(k)\end{pmatrix} =
\begin{pmatrix}\xi-\bok(R_i,\theta_{i,j},w) \\
  E-\omega(\bok(R_i,\theta_{i,j},w))\end{pmatrix}
+ \begin{pmatrix}\bok(R_i,\theta_{i,j},w)-\bok(r,\theta,w) \\ \omega(\bok(R_i,\theta_{i,j},w))-\omega(\bok(r,\theta,w))\end{pmatrix}
\end{equation}
and estimate
\[
\left|\begin{pmatrix} \bok(R_i,\theta_{i,j},w)-\bok(r,\theta,w) \\
    \omega(\bok(R_i,\theta_{i,j},w))-\omega(\bok(r,\theta,w))\end{pmatrix} \right|
\leq C_{i,j}\max\bigl\{\epsilon_\theta,\epstwo_r\bigr\},
\]
using that $\nabla \omega$ is bounded to argue for the existence of the
constant $C_{i,j}$. Put $C=\max_{i,j} C_{i,j}$. Since the first term
on the right-hand side of \eqref{Eq-DistToCross} is an element of $\cross_{i,j}$ we observe that if
we choose $\epsilon_\theta,\epstwo_r \leq d/(2C)$ we can conclude that
\[
\forall i,j \ \textup{and} \ k\in\bTT_{i,j}\bigl(\epsilon_\theta,\epstwo_r\bigr):
\quad d((\xi-k,E-\omega(k)),\cross\backslash\cross_{i,j})\geq \frac{d}2.
\]
The constant $d$ was defined in \eqref{ChoiceOfd}.
We now make the choice
\[
\epsilon_\theta = \min\bigl\{\epsone_\theta,d/(2C)\bigr\},\quad
\epstwo_r = \min\bigl\{\epsthree_r,d/(2C)\bigr\},
\]
and emphasize that with this choice
the desired inclusion  \eqref{OnlyOneCrossingFromTori}
holds true.

Our next task is to pick $\delone$ small enough such that
$\cKJ{\xi}{[E-\delone,E+\delone]}^\cross$, cf.~\eqref{Set-KO-cross},
is contained inside the union over $i$ and $j$ of the tori $\TT_{i,j}(\epsilon_\theta,\epstwo_r)$,
and such that the inclusion \eqref{OnlyOneCrossingFromTori} 
remains valid when $\{E\}$ is replaced by the interval $[E-\delone,E+\delone]$.

\begin{lemma}\label{Lem-delta-Choice} Assume Conditions~\ref{Cond:MC}
  and~\ref{Cond:MT}, with $n_0=0$, and  let $(\xi,E)\in\cE\p{1}\backslash\bigl(\thr\p{1}\cup\Exc\bigr)$. There exists
$\delta'>0$ such that $\cJ'=[E-\delta',E+\delta'] \subset \cE\p{1}(\xi)\backslash\bigl(\thr\p{1}(\xi)\cup\Exc(\xi)\bigr)$ and
\begin{align}
\label{deltaprime1}
& \cKJ{\xi}{\cJ'}^\cross  \subset \bigcup_{i=1}^M\bigcup_{j=1}^{M_i}
\TT_{i,j}\bigl(\epsilon_\theta,\epstwo_r\bigr)\\
\label{deltaprime2}
& \Spectrum_{\cJ'}^\cross \Bigl(\bTT_{i,j}\bigl(\epsilon_\theta,\epstwo_r\bigr)\Bigr)  \subset \cross_{i,j}.
\end{align}
\end{lemma}

\begin{remark}\label{Rem-delta-NoTori} Included in the conclusion of the lemma is that if $\cK_{\{E\}}^\cross = \emptyset$,
then $\delta'$ can be chosen such that $\cK_{\cJ'}^\cross = \emptyset$. 
\remarkQED\end{remark}

\begin{proof} That $\delta'>0$ can be chosen such that $\cJ'\subset \cE\p{1}(\xi)\backslash\bigl(\thr\p{1}(\xi)\cup\Exc(\xi)\bigr)$ follows from Theorem~\ref{Thm-thr}, cf.~also \eqref{ExcSet}.
Assume the inclusion \eqref{deltaprime1} is false. Then there exists
a sequence $\lambda_n\in [E-1/n,E+1/n] =: \cJ_n$
 and $k_n\in \cK_{\cJ_n}^\cross$ with $(\xi-k_n,\lambda_n-\omega(k_n))\in\cross$
and $k_n \not\in  \cup_{i=1}^M\cup_{j=1}^{M_i}
\TT_{i,j}(\epsilon_\theta,\epstwo_r)$.

By Lemma~\ref{Lemma-CompactKandSigma}, we can extract a convergent subsequence $\{k_{n_j}\}$ converging to
a  momentum  $k$. Since $\lambda_{n_j}\to E$ we must have $(\xi-k,E-\omega(k))\in\cross$,
and hence $k\in\cK_{\{E\}}^\cross$.
Since the tori are open we conclude furthermore
that $k \not\in  \cup_{i=1}^M\cup_{j=1}^{M_i}
\TT_{i,j}(\epsilon_\theta,\epstwo_r)$. But this
 contradicts Lemma~\ref{FinitelyManyCrossings} and we have thus established \eqref{deltaprime1}.

 As for \eqref{deltaprime2} we proceed in a similar fashion assuming that for any $n$
 there exists $\lambda_n\in [E-1/n,E+1/n]$ and $k_n\in\bTT_{i,j}(\epsilon_\theta,\epstwo_r)$
 such that $(\xi-k_n,\lambda_n-\omega(k_n))\in\cross\backslash \cross_{i,j}$.

 Again, by Lemma~\ref{Lemma-CompactKandSigma}, we must have a subsequence $k_{n_j}$
 converging to a momentum $k\in
 \bTT_{i,j}(\epsilon_\theta,\epstwo_r)$. 
 For this $k$ we must have $(\xi-k,E-\omega(k))\in \cross\backslash\cross_{i,j}$,
 and hence $(\xi-k,E-\omega(k))\in\Sigma_{\{E\}}^\cross$. But this contradicts \eqref{OnlyOneCrossingFromTori}.
\end{proof}

We identify the mass shells available for scattering channels, starting at momentum $\xi$ and energy in $\cJ'$ to be
\begin{equation}\label{Sprime}
\shells' := \bigset{(\Annul,\Shell)\in\shells}{\cG_\Shell\cap \Spectrum_{\cJ'}\neq\emptyset}.
\end{equation}
Here $\cG_\Shell$ denotes the graph of $\Shell$, cf. \eqref{Graph-Of-Shell}.
By compactness of $\Spectrum_{\cJ'}$ this set is finite. We list the radii of the spheres forming $\partial \Annul$, where $(\Annul,\Shell)\in\shells'$,
as $R_1',\dots, R'_{M'}$. We \emph{exclude} from the list of radii those already included in $R_1,R_2,\dots,R_M$.
With this choice we find that
\[
\forall 1\leq \ell\leq M':\quad r'_\ell := d\bigl( \cK_{\cJ'}, \partial B(\xi,R_\ell')\bigr) >0.
\]
The next thing we need to do is to ensure that the set $\cK_{\cJ'}$ approaches level crossings,
or more precisely the spheres $\partial B(\xi,R_i)$,  through the
radial face of the tori. Define for $\epsilon>0$ the compact set
\begin{equation}\label{Kepsilon}
\cK(\epsilon) := \cK_{\cJ'} \backslash \Bigl(\bigcup_{i,j}\TT_{i,j}(\epsilon_\theta,\epsilon)\Bigr)
\end{equation}
and for $i=1,\dots,M$ subsets
\[
\cK_{i} := \bigset{\bok(r,\theta,w)\in \cK(\epstwo_r) }{
|\theta-\theta_{i,j}|\geq\epsilon_\theta, j=1,\dots,M_i},
\]
which again are compact sets. Let
\[
r_{i} :=  d\bigl(\cK_i, \partial B(\xi,R_i)\bigr)>0,
\]
where strict positivity follows from \eqref{deltaprime1}.
Finally, we define $r_\Exc = +\infty$ if $\Exc = \emptyset$, and $r_\Exc = d(0,\cK_{J'})>0$ if $\Exc \neq\emptyset$. We now pick
an upper bound for the radial thickness $\epsilon_r$ to be
\begin{equation}\label{epsone-Choice}
\epsone_r := \min\bigl\{\epstwo_r,\min_{1\leq i\leq M} r_{i},\min_{1\leq \ell\leq M'} r_{\ell}',r_\Exc\bigr\}.
\end{equation}
This choice ensures that for $\epsilon\leq\epsone_r$ the
set $\cK(\epsilon)$ approaches the spheres $\partial B(\xi,R_i)$, $i=1,\dots,M$, through the radial
faces of the tori $\TT_{i,j}(\epsilon_\theta,\epsilon)$, not through their angular faces.
In addition, $k$'s in $\cK(\epsilon)$ stay at least a distance $\epsilon$ away from boundaries
of annuli in which the relevant mass shells in $\shells'$ are defined.
To summarize: For all $0<\epsilon\leq \epsone_r$, $i\in\{1,\dots,M\}$ and $\ell\in\{1,\dots,M'\}$ we have
\begin{equation}\label{DistToBoundary}
d\bigl(\cK(\epsilon), \partial B(\xi,R_i)\bigr) \geq \epsilon 
\quad \textup{and} \quad d\bigl(\cK(\epsilon), \partial B(\xi,R'_\ell)\bigr) \geq \epsilon.
\end{equation}

\subsection{An Analytic Consideration}

The next part of the construction is somewhat less obvious, in that it
anticipates the proof of the Mourre estimate to follow. We need to
construct a conjugate operator, i.e. a vector field, in a set like $\cKJ{\xi}{\cJ'}$, but we proceed
differently depending on whether we are inside or outside one of the tori
introduced in the previous subsection, cf.~\eqref{Torus-Def}. If we
are inside a torus, which is the situation we deal with in this subsection, we want the conjugate
operator to be a $\theta$-derivative. We now proceed to compute what
turns out to be the relevant commutator inside a torus and get
something positive on the crossing $\chi_{i,j}$ sitting at radius
$R_i$. Then we pick $\epsilon_r\leq \epsone_r$ small
enough for the expression to remain positive inside the torus. Note
that this subsection, as with the previous one, only comes into play when
$\nu\geq 2$ and $\xi\neq 0$.

We anticipate a conjugate operator of the form \eqref{A-GeneralForm},
cf.~also \eqref{a-GeneralForm}.
We require that the vector field $v\in C_0^\infty(\RR^\nu)$ entering into the construction of the one-body conjugate operator $a$, which remains to constructed,  satisfies
\begin{equation}\label{v-Constraint}
\|v\|_\infty\leq 2 + \max_{1\leq i\leq M} R_i.
\end{equation}
Furthermore, if $\Exc\neq\emptyset$ we demand that $0\not\in\supp{v}$.

We define, for $r>0$,  auxiliary Hamiltonians $G_\xi(r)$ on the Hilbert space
$L^2([0,\pi]\times S^{\nu-2};\cF)$ by the following direct integral
construction
\begin{equation}\label{Gxi-r}
G_\xi(r) = \int^\oplus_{\tS^{\nu-1}} G_\xi(r,\theta,w)\,\D
\theta \D w,
\end{equation}
where we abbreviated $\tS^{\nu-1} := [0,\pi]\times S^{\nu-2}$ and 
\begin{equation}\label{Gxi-rtw}
G_\xi(r,\theta,\omega) = H(\xi-\bok(r,\theta,w)) + \omega(\bok(r,\theta,w)) \one_\cF.
\end{equation}
For a $\rho_{i,j}\in C_0^\infty(\RR_\theta)$, we define a self-adjoint cutoff angular derivative
\[
\ta =\sigma_{i,j}
\frac{\ri}{2}\bigl\{ \partial_\theta \rho_{i,j} + \rho_{i,j} \partial_\theta\bigr\}.
\]
We  fix our choice of $\rho_{i,j}$ to be compactly supported in 
$(\theta_{i,j}-2\epsilon_\theta,\theta_{i,j}+2\epsilon_\theta)\subset(0,\pi)$,
equal to $1$ on
\begin{equation}\label{Thetaij}
\Theta_{i,j} := (\theta_{i,j}-\epsilon_\theta,\theta_{i,j}+\epsilon_\theta)
\end{equation}
and satisfying that $0\leq\rho_{i,j}\leq 1$.
Observe that $\ta$ only acts on the base space, not on the fiber $\cF$.
We now stitch $A$ and $\ta$ together to get a conjugate operator on
$\cF\otimes L^2(\tS^{\nu-1})$
\[
\tA\p{1} = A\otimes \one_{L^2(\tS^{\nu-1})} + \one_\cF\otimes \,\ta,
\]
where we appeal to the identification $\cF\otimes L^2(\tS^{\nu-1})\simeq L^2(\tS^{\nu-1};\cF)$.
One can verify that $G_\xi(r)$ is of class $C^1(\tA\p{1})$ and
\[
\ri\bigl[G_\xi(r),\tA\p{1}\bigr]^\circ =  \int^\oplus_{\tS^{\nu-1}}
\ri\bigl[G_\xi(r),\tA\p{1}\bigr]^\circ(\theta,w)
\, \D\theta \D w,
\]
where, as an identity on $\cF$,
\begin{equation}\label{CommWithGxi}
\begin{aligned}
& \ri\bigl[G_\xi(r),\tA\p{1}\bigr]^\circ(\theta,w)   = \ri \bigl[H(\xi-\bok(r,\theta,w)),A\bigr]^\circ\\
& \quad +\sigma_{i,j}\rho_{i,j}(\theta)\bov_\xi(r,\theta,w)
\cdot\Bigl(\nabla\omega(\bok(r,\theta,w))\one_\cF -\nabla\Omega\bigl(\xi-\bok(r,\theta,w)-\D\Gamma(k)\bigr)\Bigr),
\end{aligned}
\end{equation}
and
\begin{equation}\label{vectorfieldvxi}
\bov_\xi(r,\theta,w) := \frac{\partial
  \bok}{\partial\theta}(r,\theta,w)=r\bigl(\cos(\theta)\Pi_\xi(w) + \sin(\theta)\frac{\xi}{|\xi|} \bigr).
\end{equation}
See also Proposition~\ref{Prop:ExtendedC1}, with $\ell=1$, for a similar commutator formula.
Let $k=\bok(r,\theta,w)\neq 0$.
By rotation invariance of $\nabla\omega$ we find that
\[
\bov_\xi(r,\theta,w)\cdot\nabla\omega(k) = (\bov_\xi(r,\theta,w)\cdot k)(k\cdot\nabla\omega(k))/|k|^2,
\]
which taken together with \eqref{SignOf-kdotnablaomega} and \eqref{ChangeOfCoordinates} enables us to establish that:
\begin{equation*}
\forall k=\bok(r,\theta,w)\in\TT_{i,j}\bigl(\epsilon_\theta,\epsone_r\bigr): \
\sigma_{i,j}\bov_\xi(r,\theta,w)\cdot\nabla\omega(k) = r\sin(\theta)|\xi|\frac{|\nabla\omega(k)|}{|k|}.
\end{equation*}
This identity in conjunction with \eqref{nablaomegaPos} implies, for all $i$ and $j$, the crucial property
\begin{equation}\label{PositiveRadialDerivative}
c_{i,j} := \inf_{\theta\in\Theta_{i,j}, w\in S^{\nu-2}}\sigma_{i,j}\bov_\xi(R_i,\theta,w)\cdot\nabla\omega(\bok(R_i,\theta,w)) >0.
\end{equation}
Note that $\sin(\theta)>0$ for $\theta\in \overline{\Theta}_{i,j}$.
The set $\Theta_{i,j}$ was defined in \eqref{Thetaij}.

Now we pick and fix a $\chi''\in C_0^\infty((E-\delta',E+\delta'))$.
We choose $\chi''$ such that
$\chi''=1$ on $[E-3\delta'/4,E+3\delta'/4]$. Introduce bounded operators
\begin{equation}\label{Bdprime}
B''(r) := \chi''(G_\xi(r))\ri \bigl[G_\xi(r),\tA^{(1)}\bigr]^\circ\chi''(G_\xi(r)).
\end{equation}
We have

\begin{lemma}\label{Lemma-ContOfComm} Assume Conditions~\ref{Cond:MC}
  and~\ref{Cond:MT}, with $n_0=1$. 
The maps $(0,\infty)\ni r\to\chi''(G_\xi(r))$ and  $(0,\infty)\ni r\to B''(r)$ are locally Lipschitz, and
furthermore: For any $0< \bar{r} <\infty$ there exists $L>0$ such that
the following holds
\[
\forall r,r'\in (0,\bar{r}]:  \quad \left\{\begin{aligned}
& \|\chi''(G_\xi(r))-\chi''(G_\xi(r'))\|\leq L |r-r'|\\
& \|B''(r)-B''(r')\|\leq L|r-r'|
\end{aligned}\right.
\]
where $L$ does not depend on $v$'s satisfying the constraint \eqref{v-Constraint}.
\end{lemma}

\begin{proof} For the purpose of this proof we abbreviate
\[
\bok = \bok(r,\theta,w) \quad \textup{and} \quad \bok' = \bok(r',\theta,w).
\]

We estimate first for $r,r'\geq 0$, using that for multi-indices $\alpha$
with $1\leq |\alpha|\leq 2$, the functions $\partial^\alpha \omega$ are
bounded:
\begin{equation}\label{GlobalLip-omega}
|\omega(\bok) - \omega(\bok')|\leq C_1|\bok-\bok'| \quad \textup{and}
\quad |\nabla\omega(\bok)-\nabla\omega(\bok')|\leq C_2|\bok-\bok'|.
\end{equation}
Here $C_1,C_2>0$ are some constants independent of
$\bok,\bok'\in \RR^\nu$.  But
\begin{equation}\label{kminuskprime}
|\bok-\bok'| =
|r-r'|\bigl|\cos(\theta)\frac{\xi}{|\xi|} + \sin(\theta)\Pi_\xi(w)\bigr|
\leq 2|r-r'|,
\end{equation}
so that when inserting into \eqref{GlobalLip-omega} we get, uniformly in $\theta$ and $w$,
\begin{equation}\label{wminuswprime}
|\omega(\bok) - \omega(\bok')|\leq 2C_1|r-r'| \quad \textup{and} \quad
|\nabla\omega(\bok)-\nabla\omega(\bok')|\leq 2C_2|r-r'|.
\end{equation}

Next we compute as an identity between operators on $\cC$
\begin{equation*}
 \Omega\bigl(\xi-\bok-\D\Gamma(k)\bigr) - \Omega\bigl(\xi-\bok'-\D\Gamma(k)\bigr)
 = \int_0^1 \nabla\Omega\bigl(\xi - s\bok - (1-s)\bok' -
 \D\Gamma(k)\bigr)\D s\, (\bok - \bok').
\end{equation*}
Appealing to \eqref{kminuskprime}, \ref{Item:GrowthOfOmegas} and~\ref{Item:BasicDerOfOmega}, we arrive at
\begin{equation}\label{OmminusOmprime}
\bigl\|\bigl(\Omega(\xi-\bok-\D\Gamma(k)) -
\Omega(\xi-\bok'-\D\Gamma(k))\bigr)\bigl(1+\Omega(\xi-\D\Gamma(k))\bigr)^{-\frac12}
  \bigr\| \leq C_3|r-r'|,
\end{equation}
which holds uniformly in $0<r,r'\leq \bar{r}$,  $\theta$ and $w$.

Using
\eqref{wminuswprime} and \eqref{OmminusOmprime}, we estimate for  $z\in \CC$ with $\im(z)\neq 0$:
\[
\bigl\|\bigl((G_\xi(r)-z)^{-1} - (G_\xi(r')-z)^{-1}\bigr)\psi\bigr\| 
\leq C |r-r'| \wgt{z}|\im(z)|^{-2}\bigl\|(H_0(\xi)+1)^{-\frac12}\psi\bigr\|,
\]
where $C>0$ does not depend on $0<r,r'\leq \bar{r}$, nor on $z$.
Representing $\chi''(G_\xi(r))$ using an almost analytic extension of
$\chi''$
now yields the estimate
\begin{equation}\label{ContOfCommStep0}
\bigl\|\bigl(\chi''(G_\xi(r))-\chi''(G_\xi(r'))\bigr)\psi \bigr\| \leq
C|r-r'|\bigl\|(H_0(\xi)+1)^{-\frac12}\psi\bigr\|.
\end{equation}
Here one should read $(H_0(\xi)+1)^{-\frac12}$ as a
$(\theta,w)$-independent operator
acting on each fiber $\cF$ by the same operator. This in particular proves
that the map $r\to \chi''(G_\xi(r))$ is locally Lipschitz and that the claimed bound holds.

We proceed to estimate the difference between the commutators, cf.~\eqref{CommWithGxi},
\begin{equation}\label{ContOfCommStep1}
\begin{aligned}
\ri\bigl[G_\xi(r),\tA^{(1)}\bigr]^\circ  - &
\ri\bigl[G_\xi(r'),\tA^{(1)}\bigr]^\circ 
 = \int^\oplus_{\tS^{\nu-1}}\Bigl\{ \ri [H(\xi-\bok),A]^\circ - \ri [H(\xi-\bok'),A]^\circ \\
 & \quad +
\sigma_{i,j}\rho_{i,j}(\theta)\bov_\xi(r,\theta,w)\cdot\bigl( \nabla\omega(\bok)
- \nabla\Omega(\xi-\bok -
\D\Gamma(k)) \bigr)\\
&\quad - \sigma_{i,j}\rho_{i,j}(\theta)\bov_\xi(r',\theta,w)\cdot\bigl(\nabla\omega(\bok') - \nabla\Omega(\xi-\bok'-\D\Gamma(k))\bigr) \Bigr\}\,
\D\theta \D w.
\end{aligned}
\end{equation}
The above equation should be read as an identity between forms on $L^2(\tS^{\nu-1};\cD)$.
Appealing to \eqref{vectorfieldvxi} and \eqref{GlobalLip-omega} we find that
\begin{equation}\label{ContOfCommStep2}
\bigl|\rho_{i,j}(\theta)\bigl(\bov_\xi(r,\theta,w)\cdot \nabla\omega(\bok)
-\bov_\xi(r',\theta,\rho)\cdot\nabla\omega(\bok')\bigr)\bigr|\leq C|r-r'|,
\end{equation}
for some $C=C(\bar{r})>0$, uniformly in  $0<r,r'\leq \bar{r}$.
Using an argument similar to the one that gave \eqref{OmminusOmprime}
we conclude the bound
\[
\bigl\|\nabla\Omega(\xi-\bok-\D\Gamma(k))-\nabla\Omega(\xi-\bok'-\D\Gamma(k))\bigr\|\leq C_4|r-r'|,
\]
for some $C_4= C_4(\bar{r})>0$ and uniformly in $r,r'\geq 0$. In
conjunction with \eqref{vectorfieldvxi}, \ref{Item:GrowthOfOmegas} and~\ref{Item:BasicDerOfOmega},
 we arrive at the estimate 
\begin{equation}\label{ContOfCommstep2.5}
\bigl\|\bigl\{\bov_\xi(r,\theta,w)\cdot\nabla\Omega(\xi-\bok-\D\Gamma(k))-\bov_\xi(r',\theta,w)\cdot\nabla\Omega(\xi-\bok'-\D\Gamma(k))
\bigr\}(\Omega(\xi- \D\Gamma(k))^{-\frac12} \bigr\|\leq C|r-r'|,
\end{equation}
valid for some $C= C(\bar{r})>0$, uniformly in $r,r'\in [0,\bar{r}]$.

It remains to deal with the first $(\theta,w)$-independent term on the
right-hand side of \eqref{ContOfCommStep1}. Compute using
Proposition~\ref{Prop:C2}~\ref{Item:HisC1} in the sense of forms on $\cD$
\begin{equation*}\label{ContOfCommStep3}
\ri [H(\xi-\bok),A]^\circ - \ri [H(\xi-\bok'),A]^\circ =
\D\Gamma(v)\cdot\bigl\{\nabla\Omega(\xi-\bok-\D\Gamma(k))
-\nabla\Omega(\xi-\bok'-\D\Gamma(k))\bigr\}.
\end{equation*}
Arguing as for \eqref{OmminusOmprime} we get
\begin{equation}\label{ContOfCommStep4}
\bigl\|(H_0(\xi)+1)^{-\frac12}\bigl(\ri [H(\xi-\bok),A]^\circ -
\ri [H(\xi-\bok'),A]^\circ\bigr)(H_0(\xi)+1)^{-1}\bigr\|\leq C |r-r'|\|v\|_\infty,
\end{equation}
where $C>0$ does not depend on $0<r,r'\leq \bar{r}$ or on $v$.

Putting together \eqref{ContOfCommStep1}, \eqref{ContOfCommStep2},
\eqref{ContOfCommstep2.5} and \eqref{ContOfCommStep4} we
arrive at
\[
\bigl\|(H_0(\xi)+1)^{-\frac12}\bigl(\ri \bigl[G_\xi(r),\tA^{(1)}\bigr]^\circ -
\ri \bigl[G_\xi(r'),\tA^{(1)}\bigr]^\circ\bigr)(H_0(\xi)+1)^{-1}\bigr\|\leq C
|r-r'|(1+\|v\|_\infty),
\]
The lemma now follows from the bound above together with
\eqref{ContOfCommStep0}.
\end{proof}

\begin{lemma}\label{Lemma-AngularSep} Assume Conditions~\ref{Cond:MC}
  and~\ref{Cond:MT}, with $n_0=1$. For any integers  $i\in\{1,\dots,M\}$ and $j\in\{1,\dots,M_i\}$
we have
\[
\one_{\Theta_{i,j}} B''(R_i) = \one_{\Theta_{i,j}}  \Bigl(\int^{\oplus}_{\tS^{\nu-1}} \rho_{i,j}(\theta)
\bov_\xi(R_i,\theta,w)\cdot\nabla\omega(\bok(R_i,\theta,w))\,\D\theta \D w \Bigr) \chi''(G_\xi(R_i))^2.
\]
Furthermore, we have
\[
 \one_{\Theta_{i,j}} B''(R_i) \geq c''  \one_{\Theta_{i,j}} \chi''(G_\xi(R_i))^2,
\]
for some $c''>0$, which does not depend on $v$'s satisfying the constraint \eqref{v-Constraint}.

\end{lemma}

\begin{remark} The operator $\one_{\Theta_{i,j}}$ should be read
as the operator
$\int^\oplus_{\tS^{\nu-1}}\one_{\Theta_{i,j}}(\theta)\one_\cF \D\theta
\D w$. 
A similar notation is used for $\rho_{i,j}$ in the proof below.
\remarkQED\end{remark}

\begin{proof} We begin by writing
\begin{equation*}
\one_{\Theta_{i,j}}  \chi''(G_\xi(R_i)) = \int_{\tS^{\nu-1}} \one_{\Theta_{i,j}}(\theta)\chi''(G_\xi(R_i,\theta,w))\,
\D\theta \D w.
\end{equation*}
When $\theta$ is confined to the neighborhood
$\Theta_{i,j}$ we have
\begin{align*}
\chi''\bigl(G_\xi(R_i,\theta,w)\bigr) &= \chi''\bigl(H(\xi-\bok(R_i,\theta,w)) + \omega(\bok(R_i,\theta,w))\bigr)\\
& = \one_{E-\omega(\bok(R_i,\theta_{i,j},e_1))}\bigl(H(\xi-\bok(R_i,\theta,w)\bigr).
\end{align*}
This is due to the choice of $\delone$, cf. \eqref{deltaprime2},
which ensures that we can at most land on one energy level, namely on the
crossing $\cross_{i,j}$ sitting at height $E$ in energy-momentum space. By the virial theorem, cf. Theorem~\ref{VirialTheorem},
this implies that
\begin{align*}
&\int^\oplus_{\tS^{\nu-1}} \one_{\Theta_{i,j}}(\theta)\chi''\bigl(G_\xi(R_i,\theta,w)\bigr)
\ri\bigl[G_\xi(R_i),\tA^{(1)}\bigr]^\circ(\theta,w)\chi''\bigl(G_\xi(R_i,\theta,w)\bigr)\, \D\theta \D w\\
& = \int^\oplus_{\tS^{\nu-1}}  \one_{\Theta_{i,j}}(\theta)\chi''\bigl(G_\xi(R_i,\theta,w)\bigr)\\
&\qquad \qquad \quad \times \ri\bigl[\bigl(E -\omega(\bok(R_i,\theta_{i,j},e_1)\bigr)\one_\cF +
\omega(\bok(R_i,\theta,w))\one_\cF,\tA^{(1)}\bigr]^\circ(\theta,w)\\
& \qquad \qquad \quad \times \chi''\bigl(G_\xi(R_i,\theta,w)\bigr)\,
\D\theta \D w\\
& = \int^\oplus_{\tS^{\nu-1}}\one_{\Theta_{i,j}}(\theta)
\bov_\xi(R_i,\theta,w) \cdot\nabla\omega(\bok(R_i,\theta,w))
\chi''\bigl(G_\xi(R_i,\theta,w)\bigr)^2\,\D\theta \D w.
\end{align*}
This proves the first part. The second
statement clearly follows from the first together with \eqref{PositiveRadialDerivative}. One can choose
\[
c'':= \inf_{|\theta_{i,j}-\theta|\leq \epsilon_\theta} \sigma_{i,j}
\bov_\xi(R_i,\theta,w)\cdot\nabla\omega(\bok(R_i,\theta,w))
> 0,
\]
which is independent of $w$.
\end{proof}

We now fix $\chi'\in C_0^\infty((E-3\delta'/4,E+3\delta'/4))$, with
$\chi'=1$ on $[E-\delta'/2,E+\delta'/2]$, and write
\begin{equation}\label{Bprime}
B'(r) = \chi'(G_\xi(r))
\ri\bigl[G_\xi(r),\tA^{(1)}\bigr]^\circ \chi'(G_\xi(r)) =\chi'(G_\xi(r)) B''(r)\chi'(G_\xi(r)).
\end{equation}
The operator $B''(r)$ was defined in \eqref{Bdprime}.

\begin{proposition}\label{Prop-PosInTorus} Assume Conditions~\ref{Cond:MC}
  and~\ref{Cond:MT}, with $n_0=1$. There exists $c'>0$ and $0<\epsilon_r\leq\epsone_r$,
independent of $v$'s satisfying \eqref{v-Constraint},
such that for all $i\in\{1,\dots,M\}$ and $j\in\{1,\dots,M_i\}$, we have
\[
\forall r\in [R_i-\epsilon_r,R_i+\epsilon_r]:\quad \one_{\Theta_{i,j}}B'(r)
\geq c' \one_{\Theta_{i,j}}\chi'(G_\xi(r))^2.
\]
\end{proposition}

\begin{proof} Apply Lemmata~\ref{Lemma-ContOfComm}  and~\ref{Lemma-AngularSep}.
This yields the bounds
\begin{align*}
\one_{\Theta_{i,j}} B''(r)
 & \geq \one_{\Theta_{i,j}}B''(R_i) - L|R_i-r|\\
&  \geq c'' \one_{\Theta_{i,j}} \chi''(G_\xi(R_i))^2 - L|R_i-r|\\
& \geq c'' \one_{\Theta_{i,j}} \chi''(G_\xi(r))^2 - L(1+c'')|R_i-r|.
\end{align*}
Here $L$ is the Lipschitz constant coming from Lemma~\ref{Lemma-ContOfComm} 
applied with $\bar{r} = \max_{1\leq i\leq M}R_i+\epsone_r$.
Choose $\epsilon_r = \min\{\epsone_r,c''(2L(1+c''))^{-1})\}$.
Multiplying both sides first by $\chi'(G_\xi(r))$ from the left and the right, and subsequently by
$\one_{\Theta_{i,j}}$,
yields the result with $c' = c''/2$. Recall that all the operators are fibered,
i.e. they are functions of $\theta$ and $w$.
\end{proof}

\subsection{The Conjugate Operator}

The task at hand in this subsection is the construction of the vector field $v_\xi\colon\RR^\nu\to\RR^\nu$, used
to define the conjugate operators
\begin{equation}\label{A-xi}
a_\xi = \frac12 \bigl(v_\xi\cdot  \ri\nabla_k + \ri\nabla_k \cdot
v_\xi\bigr),\quad \textup{and} \quad A_\xi = \D\Gamma(a_\xi).
\end{equation}
These are operators of the form considered in Section~\ref{Chap-Regularity},
cf. \eqref{a-GeneralForm} and~\eqref{A-GeneralForm}.
The vector field $v_\xi$ will depend both on the total momentum $\xi$
and on the energy localization we choose. 

Our first ingredient is a partition
of unity in momentum space subordinate to an appropriately chosen open
covering of $\cK_{\cJ'}$. The first sets in the covering were
constructed in the previous section, namely the disjoint open tori
$\TT_{i,j}(\epsilon_\theta,\epsilon_r)$. If the dimension $\nu$ is $1$, $\xi=0$, or 
$\omega$ is constant (polaron model), there are no tori and the
construction simplifies. In particular, the considerations of the
previous two subsections are superfluous.

Define a set of momenta
\[
\cK  := \cK(\epsilon_r),
\]
which is a compact subset of $\RR^\nu$. Recall from \eqref{Kepsilon} the definition of the set $\cK(\epsilon_r)$.

If there are no tori, i.e. if $\cK_{\{E\}}^\cross = \emptyset$, then $\cK = \cK_{\cJ'}$ (and $\epsilon_r$ is not defined). Note that $\cJ'$ is still given by Lemma~\ref{Lem-delta-Choice}.
The key property of the set $\cK(\epsilon_r)$ is that it is separated from $\partial B(\xi,R_i)$ and $\partial B(\xi,R_\ell')$
by a distance at least $\epsilon_r$, cf.~\eqref{DistToBoundary}. To ensure this property also holds if there are no tori, we define in that case $\epsilon_r =  \min_{1\leq \ell\leq M'} d(\cK_{\cJ'},\partial B(\xi,R_\ell'))>0$,
where positivity follows from Remark~\ref{Rem-delta-NoTori}.
Here $R_\ell'$ lists the inner and outer radii of $\Annul$, for $(\Annul,\Shell)\in\shells'$, where $\shells'$ is the collection of mass shells that a state localized in $\cJ'$ can relax to, cf.~\eqref{Sprime}.
If furthermore $\Exc\neq\emptyset$, we pick $\epsilon_r$ possibly smaller such that $\epsilon_r < d(0,\cK_{\cJ'})$,
which is possible since in this case $0\not\in \cK_{\cJ'}$. See also \eqref{epsone-Choice}.

We proceed to pick a $0<\delta\leq \delone$ with the property that the choice
$\cJ=[E-\delta,E+\delta]$ ensures that
$\Spectrum_{\cJ}(\cK)$ is a graph, i.e. the projection onto momentum
space $\RR^\nu$ is injective. To do this we define
the energy distance between mass shells $(\cA,\Shell)\in\shells'$, away from $\epsilon_r/2$ neighborhoods
of their boundaries $\partial \Annul$: 
By compactness of $\Spectrum_{\cJ'}\subset \Spectrum_\iso$, there exist
$P>0$ and $\sigma>0$ such that
\[
\Spectrum_{\cJ'}(\cK)\subset\Spectrum_{\cJ'}\subset
\bigset{(p,\lambda) }{|p|\leq P \ \textup{and} \ \lambda\leq
  \Spectrum_\ess(p)-\sigma}.
\]
Hence we define for $(\Annul,\Shell)\in\shells'$
\[
\begin{aligned}
& \delta_{(\Annul,\Shell)} = \inf \bigset{ d(\Shell(p),\sigma_\pp(H(p))\backslash\{\Shell(p)\}) }{
p\in \Annul_{\epsilon_r,P}, \Shell(p)\leq\Spectrum_\ess(p)-\sigma },\\
&\textup{where} \ \Annul_{\epsilon_r,P} := \bigset{p\in\Annul}{ d(p,\partial A)\geq \epsilon_r/2}\cap \overline{B(0,P)}.
\end{aligned}
\]
Again, by compactness, $\delta =
\min_{(\Annul,\Shell)\in\shells'}\delta_{(\Annul,\Shell)}>0$. 
Recall that $\shells'$ denotes the finite collection of mass shells
available for scattering, cf. \eqref{Sprime}.

We split $\cK$ into compact components pertaining to shells $(\Annul,\Shell)\in\shells'$
\[
\cK_{(\Annul,\Shell)} =  \cK\cap (\Annul+\xi)=\bigset{ k\in\cK }{
  \xi-k \in \Annul }.
\]
We pick open neighborhoods, using \eqref{DistToBoundary} to verify the inclusion,
\[
\cV_{(\Annul,\Shell)} = \bigset{ k\in\RR^\nu }{
  d(k,\cK_{(\Annul,\Shell)})<\epsilon_r/2 }\subset \Annul+\xi
\]
which inherit the property of $\cK$ that the distance from $\Shell(\xi-k)$ to the nearest eigenvalue in
$\sigma_\pp(H(\xi-k))\backslash\{\Shell(\xi-k)\}$ is at least $\delta$.
In addition we remark that the sets $\cV_{(\Annul,\Shell)}$ are pairwise disjoint but overlap with possibly existing tori
$\TT_{i,j}(\epsilon_\theta,\epsilon_r)$, provided $\partial B(0,R_i)$ is a boundary of $\Annul$.
If $\Exc\neq\emptyset$, then $0\not\in\cV_{(\Annul,\Shell)}$ for any $(\Annul,\Shell)\in\shells'$.

To make a partition of unity we choose first $\varphi_{i,j}\in
C_0^\infty(\TT_{i,j}(2\epsilon_\theta,2\epsilon_r))$
such that $\varphi_{i,j} =1 $ on
$\TT_{i,j}(\epsilon_\theta,\epsilon_r)$. It will be convenient to use
a product construction such that
\[
\varphi_{i,j}(\bok(r,\theta,w)) = \rho_{i,j}(\theta)\trho_i(r),
\]
where $\rho_{i,j}$ was introduced in the beginning of the previous subsection and $\trho_i\in C_0^\infty((-2\epsilon_r,2\epsilon_r))$,
with $\trho_i = 1$ on $[R_i-\epsilon_r,R_i+\epsilon_r]$ and satisfying that $0\leq \trho_i\leq 1$.
Note that $0\not\in\supp{\varphi_{i,j}}$, cf.~\eqref{TorusHasHole}.

Using the smooth Urysohn lemma on the pairs
$\cK_{(\Annul,\Shell)}\subset \cV_{(\Annul,\Shell)}$, with $(\Annul,\Shell)\in\shells'$,  yields smooth functions
$\tvarphi_{(\Annul,\Shell)}$
with compact support in $\cV_{(\Annul,\Shell)}$ and equal to $1$ on
$\cK_{(\Annul,\Shell)}$. By a standard construction we can replace these by
(possibly) smaller functions $\varphi_{(\Annul,\Shell)}$ with the same two properties
and the additional property that
\begin{equation}\label{PartitionProperty}
\forall k\in \cK_{\cJ'} : \quad \sum_{i,j}\varphi_{i,j}(k) +
\sum_{(\Annul,\Shell)\in\shells'} \varphi_{(\Annul,\Shell)}(k)=1.
\end{equation}
If $\Exc\neq\emptyset$, we observe that $0\not\in\supp{\varphi_{(\Annul,\Shell)}}$ as required in order to deal with what is in this case an infrared singular coupling.

We can now construct our vector field
\begin{equation}\label{VectorField}
v_\xi = \sum_{i,j} v_{i,j}^\xi + \sum_{(\Annul,\Shell)\in\shells'}
v^\xi_{(\Annul,\Shell)},
\end{equation}
where
\begin{equation}\label{VectorFieldComp}
v^\xi_{i,j}(k) = \sigma_{i,j}\varphi_{i,j}(k)\bov_\xi(k) \quad
\textup{and} \quad
v^\xi_{(\Annul,\Shell)}(k)= \varphi_{(\Annul,\Shell)}(k) \frac{\nabla\Shell\p{1}(\xi;k)}{|\nabla\Shell\p{1}(\xi;k)|}.
\end{equation}
Recall from \eqref{EffectiveShell} and \eqref{vectorfieldvxi} the construction of the dispersion relation
$S\p{1}(\xi;\cdot)$ and the vector field $\bov_\xi$. The terms $v^\xi_{i,j}$ can only appear if $\xi\neq 0$
and $\nu\geq 2$, and one should then read $\bov_\xi(k) =
\bov_\xi(r,\theta,w)$, where $k = \bok(r,\theta,w)$.

We have thus finished  the construction of the
conjugate operator $A_\xi$, cf. \eqref{A-xi}.
We remark that the construction of $v_\xi$ is consistent with the constraint \eqref{v-Constraint}, cf. \eqref{vectorfieldvxi}.
The signs $\sigma_{i,j}$ and the vector field $\bov_\xi$ were introduced in
\eqref{SignOf-kdotnablaomega} and \eqref{vectorfieldvxi} respectively.

 In the following we make use of the notation 
\begin{equation}
A\ext_\xi = \D\Gamma\ext(a_\xi) = A_\xi\otimes\,\one_\cF + \one_\cF\otimes\, A_\xi\qquad
\textup{on} \ \cF\ext
\end{equation}
and observe the direct sum decomposition 
\begin{equation}\label{A-xi-ell}
A\ext_\xi = \oplus_{\ell=0}^\infty A\p{\ell}_\xi,\qquad\textup{where}
\ A\p{\ell}_\xi = A_\xi\otimes\one_{\cF\p{\ell}} + \one_\cF\otimes
\,\D\Gamma(a_\xi)_{|\cF\p{\ell}} \ \textup{on} \  \cF\otimes\cF\p{\ell}.
\end{equation}
In particular, $A\p{0}_\xi = A_\xi$. See also Subsection~\ref{sec-ext}
for notation and constructions pertaining to extended objects.

The following proposition is a Mourre estimate for $H\p{1}(\xi)$, with
conjugate operator $A\p{1}_\xi$, stating that a composite system consisting of a dressed matter particle and a free boson
at total momentum $\xi$, localized in energy in the interval $\cJ$, has non-zero breakup velocity.
This is the source of positivity in the Mourre estimate for $H$ we
prove in the following subsection.

\begin{proposition}\label{Prop-MEOne} Assume Conditions~\ref{Cond:MC}
  and~\ref{Cond:MT}, with $n_0=1$. Let $\chi\in C_0^\infty(\cJ)$, with $\chi\geq 0$. Then there exists $c>0$ such that
\[
 \chi\bigl(H\p{1}(\xi)\bigr)\ri \bigl[H\p{1}(\xi),A\p{1}_\xi\bigr]^\circ
 \chi\bigl(H\p{1}(\xi)\bigr)
\geq c \chi\bigl(H\p{1}(\xi)\bigr)^2.
\]
\end{proposition}

\begin{remark} It is in fact only the $\one_\cF\otimes\, a_\xi$ part of
  $A_\xi\p{1}$ which is important for positivity. In fact, the
  proposition remains true if $A_\xi\p{1}$ is replaced by 
$\tA\otimes\one_\hph + \one_\cF\otimes\, a_\xi$, where $\tA=
\D\Gamma(\ta)$, with $\ta =\ri \{\tv\cdot \nabla_k + \nabla_k\cdot \tv\}/2$ and $\tv\in C_0^\infty(\RR^\nu)$ 
satisfying the constraint \eqref{v-Constraint}. In particular the choice $\tA = 0$ works.  
\remarkQED\end{remark}

\begin{proof}
Let
\[
B(k) = \chi\bigl(H\p{1}(\xi;k)\bigr)\ri \bigl[H\p{1}(\xi),A\p{1}_\xi\bigr]^\circ(k)  \chi\bigl(H\p{1}(\xi;k)\bigr)
\]
and write
\begin{align*}
B & := \chi\bigl(H\p{1}(\xi)\bigr)\ri \bigl[H\p{1}(\xi),A\p{1}_\xi\bigr]^\circ \chi\bigl(H\p{1}(\xi)\bigr)\\
& = \int^\oplus_{\RR^\nu} B(k) \, \D k\\
  & = \sum_{i,j}\int^\oplus_{\RR^\nu} \one_{\TT_{i,j}(\epsilon_\theta,\epsilon_r)}(k) B(k)\, \D k+ \sum_{(\Annul,\Shell)\in \shells'}\int^\oplus_{\RR^\nu} \one_{\cK_{(\Annul,\Shell)}}(k) B(k) \, \D k,
\end{align*}
where the summation over $i$ and $j$ is understood to be over $i=1,\dotsc,M$ and $j=1,\dotsc,M_i$.

We split the operator $A\p{1}_\xi$ into the sum
\[
A\p{1}_\xi = A_\xi\otimes \one_\hph + \sum_{i,j} \one_\cF\otimes \,a_{i,j} + \sum_{(\Annul,\Shell)\in\shells'}\one_\cF\otimes\, a_{(\Annul,\Shell)},
\]
corresponding to the construction of $v_\xi$,
cf.~\eqref{VectorField} and \eqref{A-xi-ell}. This induces a decomposition of $B(k)$
using Lemma~\ref{Lemma:SumsOfAs} and Proposition~\ref{Prop:ExtendedC1}:
\[
B(k) = B_0(k) + \sum_{i,j} B_{i,j}(k) + \sum_{(\Annul,\Shell)\in\shells'} B_{(\Annul,\Shell)}(k),
\]
where
\[
\begin{aligned}
&B_0(k) = \chi\bigl(H^{(1)}(\xi;k)\bigr)
\ri \bigl[H\p{1}(\xi),A_\xi\otimes\one_\hph\bigr]^\circ(k)  \chi\bigl(H\p{1}(\xi;k)\bigr), \\
&B_{i,j}(k) = \chi\bigl(H^{(1)}(\xi;k)\bigr)
\ri \bigl[H\p{1}(\xi),\one_\cF\otimes\, a_{i,j}\bigr]^\circ(k)  \chi\bigl(H\p{1}(\xi;k)\bigr), \\
&B_{(\Annul,\Shell)}(k) =\chi\bigl(H\p{1}(\xi;k)\bigr)
\ri \bigl[H\p{1}(\xi),\one_\cF\otimes \,a_{(\Annul,\Shell)}\bigr]^\circ(k)  \chi\bigl(H\p{1}(\xi;k)\bigr).
\end{aligned}
\]
Observe that that $k\to B_{i,j}(k)$ has support in the torus $\TT_{i,j}(2\epsilon_\theta,2\epsilon_r)$
and $k\to B_{(\Annul,\Shell)}(k)$ has support in $\cV_{(\Annul,\Shell)}$.
Using these support properties we compute
\[
\int^\oplus_{\RR^\nu} \one_{\TT_{i,j}(\epsilon_\theta,\epsilon_r)}(k) B(k) \, \D k
= \int^\oplus_{\RR^\nu} \one_{\TT_{i,j}(\epsilon_\theta,\epsilon_r)}(k)\bigl( B_0(k) +B_{i,j}(k)\bigr) \, \D k.
\]
Write $\tB_{i,j}(k) = B_0(k)  + B_{i,j}(k)$ and recall the notation
$\tS^{\nu-1}=[0,\pi]\times S^{\nu-2}$.  Observe that 
\[
\int^{\oplus}_{\tS^{\nu-1}}\tB_{i,j}(\bok(r,\theta,w))\D\theta\D w = \chi(G_\xi(r))B'(r)\chi(G_\xi(r)),
\]
where $B'(r)$ was introduced in \eqref{Bprime} and $G_\xi(r)$ in \eqref{Gxi-r} and~\eqref{Gxi-rtw}. 
Here we used that $\chi\chi' = \chi$. Denote by $\cU\colon
L^2(\RR^\nu;\cF)\to L^2([0,\infty)\times \tS^{\nu-1};\cF)$ the unitary
operator defined by
\[
(\cU\psi)(r,\theta,w) = r^{\frac{\nu-1}2}\sqrt{\sin(\theta)}\,\psi(\bok(r,\theta,w)).
\] 
We have obtained the identity
\begin{align*}
&\int^\oplus_{\RR^\nu} \one_{\TT_{i,j}(\epsilon_\theta,\epsilon_r)}(k) B(k)\, \D k\\
 & \quad = \cU^*\left(\int^\oplus_{[0,\infty)} \one_{(R_i-\epsilon_r,R_i+\epsilon_r)}(r)
\left\{\int^\oplus_{\tS^{\nu-1}} \one_{\Theta_{i,j}}(\theta)\tB_{i,j}(\bok(r,\theta,w))\, \D \theta \D w \right\} \D r\right)\cU\\
& \quad = \cU^*\left( \int^\oplus_{[0,\infty)} \one_{(R_i-\epsilon_r,R_i+\epsilon_r)}(r)
\one_{\Theta_{i,j}}\chi(G_\xi(r))B'(r)\chi(G_\xi(r)) \,\D r\right)\cU.
\end{align*}
From Proposition~\ref{Prop-PosInTorus} we thus find a $c'>0$,
independent of $i$ and $j$,  such  that
\begin{equation}\label{EstInTorus}
\int^\oplus_{\RR^\nu} \one_{\TT_{i,j}(\epsilon_\theta,\epsilon_r)}(k) B(k)\, \D k 
\geq c' \int^\oplus_{\RR^\nu} \one_{\TT_{i,j}(\epsilon_\theta,\epsilon_r)}(k)\chi\bigl(H\p{1}(\xi;k)\bigr)^2  \D k.
\end{equation}

To deal with the remaining contributions we compute, using the support properties of $B_{(\Annul,\Shell)}$,
\[
\one_{\cK_{(\Annul,\Shell)}}(k) B(k) = \one_{\cK_{(\Annul,\Shell)}}(k)
\Bigl\{ B_0(k) + B_{(\Annul,\Shell)}(k) + \sum_{i,j} B_{i,j}(k)\Bigr\}.
\]
Since $\one_{\cK_{(\Annul,\Shell)}}(k) \chi(H\p{1}(\xi;k)) = \one_{\cK_{(\Annul,\Shell)}}(k) \one_{\{\Shell\p{1}(\xi;k)\}}(H\p{1}(\xi;k))$
we can apply the virial theorem, cf. Theorem~\ref{VirialTheorem}, to compute for a.e. $k = \bok(r,\theta,w)\in \cK_{(\Annul,\Shell)}$
\[
\begin{aligned}
& B_0(k) = 0,\\
& B_{(\Annul,\Shell)}(k)= \vphi_{(\Annul,\Shell)}(k) |\nabla_k\Shell\p{1}(\xi;k)|\geq  c'_{(\Annul,\Shell)}\vphi_{(\Annul,\Shell)}(k),\\
& B_{i,j}(k) = \sigma_{i,j}\vphi_{i,j}(k) \bov_\xi(r,\theta,w)\cdot \nabla\omega(\bok(r,\theta,w))\geq c_{i,j}\vphi_{i,j}(k).
\end{aligned}
\]
Here the constants $c_{i,j}$ are defined in \eqref{PositiveRadialDerivative},
\[
c'_{(\Annul,\Shell)} := \inf_{k\in  \cK_{(\Annul,\Shell)}}   |\nabla_k\Shell\p{1}(\xi;k)| >0
\]
and positivity follows from $\cK_{(\Annul,\Shell)}$ being compact and $\cJ$ being chosen to not contain threshold energies.
Summing up, cf. \eqref{PartitionProperty} and recalling that for distinct shells $\one_{\cK_{(\Annul,\Shell)}}\vphi_{(\Annul',\Shell')} = 0$, we get
\begin{equation}\label{EstOnShell}
\int^\oplus_{\RR^\nu} \one_{\cK_{(\Annul,\Shell)}}(k) B(k) \,\D k 
\geq c_{(\Annul,\Shell)} \int^\oplus_{\RR^\nu}\one_{\cK_{(\Annul,\Shell)}}(k) \chi\bigl(H\p{1}(\xi;k)\bigr)^2\D k,
\end{equation}
with $c_{(\Annul,\Shell)} = \min\{c'_{(\Annul,\Shell)},\min_{i,j} c_{i,j}\}$.

Combining \eqref{EstInTorus} and \eqref{EstOnShell} we finally get
\begin{align*}
B & \geq c\int^\oplus_{\RR^\nu} \Bigl\{\sum_{i,j}\one_{\TT_{i,j}(\epsilon_\theta,\epsilon_r)}(k) +
 \sum_{(\Annul,\Shell)\in\shells'} \one_{\cK_{(\Annul,\Shell)}}(k)\Bigr\}\chi\bigl(H\p{1}(\xi;k)\bigr)^2\,\D k\\
 & = c\chi\bigl(H\p{1}(\xi)\bigr)^2,
\end{align*}
with $c= \min\{c',\min_{(\Annul,\Shell)\in\shells'} c_{(\Annul,\Shell)}\}$.
\end{proof}

\subsection{The Mourre Estimate}

We will make use of a geometric partition of unity,
introduced in \cite{DGe1}, and since used frequently to perform localization arguments in non-relativistic QFT
\cite{Am,DGe2,FGSch2,GGM2,MAHP}.

The input is a pair of smooth functions $j_0,j_\infty\colon\RR^\nu\to [0,1]$,
with the properties $j_0^2 + j_\infty^2 = 1$ and $j_0 = 1$ on $\set{x\in\RR^\nu}{\abs{x}\leq 1}$ and
$j_\infty=1$ on $\set{x\in\RR^\nu}{\abs{x}\geq 2}$. We scale these functions and define for $R>0$
localizations $j_\#^R(x) = j_\#(x/R)$ inside and outside of balls with a radius scaling like $R$.
Reading now  $x= \ri\nabla_k$ these operators become bounded self-adjoint operators on $\hph$
and we form the vector operator $j^R = (j_0^R,j_\infty^R)\colon \hph\to \hph\oplus\hph$.
It satisfies $(j^R)^*j^R = \one_\hph$. The operator $\Gamma(j^R)$ now maps
$\cF = \Gamma(\hph)\to\Gamma(\hph\oplus\hph)$, and composing with the
canonical identification operator 
$U\colon \Gamma(\hph\oplus\hph) \to \Gamma(\hph)\otimes\Gamma(\hph) = \cF\ext$ we get
\[
\cGam(j^R) := U\Gamma(j^R) \colon \cF\to\cF\ext.
\]
The operator is a \myquote{partition of unity} in that it is an isometry, i.e. $\cGam(j^R)^*\cGam(j^R) = \one_\cF$.

\begin{lemma}\label{cor:locerror} Assume Conditions~\ref{Cond:MC}
  and~\ref{Cond:MT}, with $n_0=1$.
Let $f\in C_0^\infty(\RR)$. Then
  \begin{enumerate}[label=\textup{(\roman*)},ref=(\roman*)]
  \item \label{item:loc1} $\check\Gamma(j^R)f(H(\xi))=f(H\ext(\xi))\check\Gamma(j^R)+o_R(1)$
  \item \label{item:loc2} $\check\Gamma(j^R) f(H(\xi)) \ri [H(\xi),A_\xi]^\circ f(H(\xi)) 
    =f(H\ext(\xi)) \ri [H\ext(\xi),A\ext_\xi]^\circ f(H\ext(\xi))\check\Gamma(j^R)+o_R(1)$
  \end{enumerate}
\end{lemma}

\begin{remark} We note that \ref{item:loc1} was already proved in
\cite{MAHP} in the case $s_\Omega\in\{0,1,2\}$. As the assumption of
$s_\Omega$ being integer is only used in the proof of this result in
\cite{MAHP}, this new proof now secures the validity of the results in
\cite{MAHP} for non-integer values of $s_\Omega$. 
\remarkQED\end{remark}

\begin{proof}
 In the following we fix a $\lambda<\Sigma_0$, cf. \eqref{Sigma0}.
  We will start by proving the following statements: For
  $p\in\{1,\dots,\nu\}$ and $w\in C_0^\infty(\RR^\nu)$ we claim that
  \begin{enumerate}[label=\textup{(\alph*)},ref=(\alph*)]
  \item \label{item:locerror1}
      $\cGam(j^R) f(H(\xi))\colon \cF\to\cD\ext_{1/2}\textup{ and
      }f(H\ext(\xi))\cGam(j^R)\colon{\cD_{1/2}}^*\to
      \cF\ext$ 
    for any $R>1$ and,
    \begin{align*}
      (H\ext(\xi)-\lambda)^{-\frac12}\bigl(\cGam (j^R) H(\xi) - H\ext(\xi) \cGam(j^R)\bigr)f(H(\xi))&=o_R(1),\\
      f(H\ext(\xi))\bigl(\cGam(j^R)H(\xi)- H\ext(\xi) \cGam(j^R)\bigr)(H(\xi)-\lambda)^{-\frac{1}{2}}&=o_R(1).
    \end{align*}
  \item
    \label{item:locerror2}$\big(\check\Gamma(j^R)\partial_p\Omega(\xi-\D\Gamma(k))- 
    \partial_p\Omega(\xi-\D\Gamma\ext(k))\check\Gamma(j^R)\big)f(H(\xi))=o_R(1)$.
  \item \label{item:locerror3}$f(H\ext(\xi))\big(
    \check\Gamma(j^R)\D\Gamma(w) 
    - \D\Gamma\ext(w) \check\Gamma(j^R)\big)(H_0(\xi)-\lambda)^{-\frac12}=o_R(1)$.
  \item \label{item:locerror4}$f(H\ext(\xi))\big(\check\Gamma(j^R)\phi(\ri a_\xi
    \coup)-\phi(\ri a_\xi
    \coup)\otimes\bbbone_\cF\check\Gamma(j^R)\big)f(H(\xi))=o_R(1)$.
  \end{enumerate}

 In the rest of the proof we abbreviate $H_0 =H_0(\xi)$, $H=H(\xi)$, $A=A_\xi$, $H\ext
 = H\ext(\xi)$, $A\ext = A\ext_\xi$, $\phi\ext(w) = \phi(w)\otimes\one_\cF$, and
 $\cGam = \cGam(j^R)$.
 For notational convenience, we write $M\olig N$ if
  $M=N+o_R(1)$.

  \ref{item:locerror1} We only prove half of the statement as the
  other half follows by a symmetric argument.
   Note that
  $(H\ext-\lambda)^{-1/2}(\cGam (H -\Omega(\xi-\D\Gamma(k))- (H\ext -\Omega(\xi-\D\Gamma\ext(k)) \cGam ) f(H)=o_R(1)$
  by (the proof of) \cite[Lemma~3.2]{MAHP}. Hence, to prove the statement,
  we need only show that
  \begin{equation}\label{loc-needtoshow}
  (H\ext-\lambda)^{-\frac12}\bigl(\cGam \Omega(\xi-\D\Gamma(k)) - \Omega(\xi-\D\Gamma\ext(k)\bigr) \cGam)f(H)=o_R(1).
  \end{equation}
  Recall that this was in fact already established in \cite{MAHP} for the
  particular cases $s_\Omega=0,1,2$. Hence we can assume $s_\Omega<2$. 

  In order to use a commutator expansion formula,
we find it useful to cast the statement differently by extending $\cGam$ and its adjoint to operators on $\cF\ext$.
Let, for the purpose of this proof only, $P\colon \cF\ext\to\cF$ be the
projection
\[
\cF\ext=\cF\oplus\Bigl(\bigoplus_{\ell=1}^\infty \cF\otimes \cF^{(\ell)}\Bigr)\ni(u,v)\mapsto
u\in\cF
\]
 and $I = P^*\colon\cF\to\cF\ext$ the injection
 \[
 \cF\ni u\mapsto I u = (u,0)\in\cF\oplus\Bigl(\bigoplus_{\ell=1}^\infty \cF\otimes\cF^{(\ell)}\Bigr).
 \]
Define $\cGam\ext\colon\cF\ext\to\cF\ext$ by
$\cGam\ext=\cGam P$. Note that $PI$ is the
identity on $\cF$ and that 
\[
\cGam\ext I=\cGam, \quad
H\ext I=IH,\quad  A\ext I=IA \quad\textup{and} \quad \phi\ext(\coup) I=I\phi(\coup).
\]

 We write, using \cite[Lemma~3.6]{MAHP},
 \begin{equation}\label{CommExpand1} 
 \begin{aligned}
    &(H\ext-\lambda)^{-\frac12}\bigl[\cGam\ext,\Omega(\xi-\D\Gamma\ext(k))\bigr]f(H\ext)\\
    &\quad =\bigl\{(H\ext-\lambda)^{-\frac12}\bigl[\cGam\ext(N\ext+1)^{-2},\Omega(\xi-\D\Gamma\ext(k))\bigr]\bigr\}(N\ext+1)^2f(H\ext).
  \end{aligned}
\end{equation}
  The estimate \eqref{loc-needtoshow} follows if the term in the
  brackets above is $o_R(1)$.
  The commutator $[\cGam\ext(N\ext+1)^{-2},\Omega(\xi-d\Gamma\ext(k))]$
  satisfies the assumptions of \cite[Theorem~3]{Ra} (if $s_\Omega<2$) with
  \[
 B=\cGam\ext(N\ext+1)^{-2}, \quad A=\xi-\D\Gamma\ext(k), \quad f_\lambda=\Omega, \quad  s=s_\Omega,
  \quad n_0=2 \quad  \textup{and} \quad n=1.
 \]
Hence, by \cite[Theorem~3]{Ra} we obtain the commutator expansion
  \begin{align}\label{CommExpand2}
\nonumber    \bigl[\cGam\ext(N\ext+1)^{-2},\Omega(\xi-\D\Gamma\ext(k))\bigr]
    & =\sum_{p=1}^\nu \partial_{p}\Omega(\xi-\D\Gamma\ext(k))
      \adjungeret_{\xi_p-\D\Gamma\ext(k_p)}\!\bigl(\cGam\ext(N\ext+1)^{-2}\bigr)\\
    &\quad+R_1\bigl(\xi-\D\Gamma\ext(k),\cGam\ext(N\ext+1)^{-2}\bigr).
  \end{align}
The remainder $R_1(A,B)$ satisfies for some $C>0$ the bound 
\begin{equation}\label{R1Bound}
\|R_1(A,B)\|\leq C \sum_{|\alpha|=2} \|\adjungeret_{A}^{\alpha}(B)\|.
\end{equation}
 Let $\psi\in \cF\p{m}$ and $\varphi\in \cF\otimes\cF\p{\ell}$, with
 $m\geq 0$ and $\ell\geq 1$. Note that $P (\psi\otimes\vacuum,\varphi)
 = \psi$.
 One can readily verify that
  \begin{equation}\label{FormOfMultipleComm}
  \adjungeret_{\D\Gamma\ext(k)}^\alpha(\cGam\ext)
    \begin{pmatrix} \psi\otimes\vacuum \\ \varphi\end{pmatrix} 
  =U\smashoperator[l]{\sum_{\sum\alpha\p{\ell}=\alpha}}\frac{\alpha!}{\prod_{\ell=1}^m\alpha\p{\ell}!}
\bigotimes_{\ell=1}^m
\begin{pmatrix} \adjungeret_{k}^{\alpha\p{\ell}}(j_0^R)\\
 \adjungeret_{k}^{\alpha\p{\ell}}(j_\infty^R)\end{pmatrix}\psi,
  \end{equation}
and
\begin{equation*}
\adjungeret_{k}^{\alpha\p{\ell}}\!\bigl(j_\#^R\bigr)=
i^{\abs{\alpha\p{\ell}}}R^{-\abs{\alpha\p{\ell}}}\bigl(\partial^{\alpha\p{\ell}}
j_\#\bigr)\!
\bigl(\frac{x}{R}\bigr) =  O\bigl(R^{-\abs{\alpha\p{\ell}}}\bigr),
\end{equation*}
where the sums are over all ordered sets of multi-indices
$\{\alpha\p{\ell}\}_{\ell=1}^m$ such that
$\sum_{\ell=1}^m\alpha\p{\ell}=\alpha$. The identity
\eqref{FormOfMultipleComm}
extends by linearity to an identity between bounded operators from $\cF\p{m}\otimes\cF$
to $\one[N\ext = m]\cF\ext$, where $N\ext = N\otimes\one_\cF +
\one_\cF\otimes N$ is the extended number operator.
Since
\begin{equation*}
  \sum_{\sum\alpha\p{\ell}=\alpha}\frac{\alpha!}{\prod_\ell \alpha\p{\ell}!}=m^{\abs{\alpha}},
\end{equation*}
it follows that
$\adjungeret_{\D\Gamma\ext(k)}^\alpha(\cGam\ext)(N\ext+1)^{-\abs\alpha}=O(R^{-\abs\alpha})$,
with respect to the norm on $\cB(\cF\ext)$,
and hence that
\begin{align*}
  \sum_{p=1}^\nu\bigl\|\adjungeret_{\xi_p-\D\Gamma\ext(k_p)}\bigl(\cGam\ext(N\ext\!+\!1)^{-2}\bigr)\bigr\|
+\bigl\|R_1\bigl(\xi-\D\Gamma\ext(k),\cGam\ext(N\ext\!+\!1)^{-2}\bigr)\bigr\|=O(R\inv).
\end{align*}
As $s_\Omega\le2$,
$(H-\lambda)^{-1/2}\partial_p\Omega(\xi-\D\Gamma\ext(k))$ is
bounded. These two observations together with \eqref{CommExpand1},
\eqref{CommExpand2} and \eqref{R1Bound}
imply the  claim  \ref{item:locerror1}. Note that one can include the
case $s_\Omega=2$ and make the argument self-contained by adding one
extra factor of $(N+1)^{-1}$ and taking $n_0=3$ and $n=2$ in the
expansion formula from \cite{Ra}. This would however
require the (very reasonable) assumption that
Condition~\ref{Cond:MC}~\ref{Item:BasicDerOfOmega} 
holds for $|\alpha|\leq 3$.

By an analogous argument we get \ref{item:locerror2}. The proof of
\ref{item:locerror3} and \ref{item:locerror4} can be found in the
proof of \cite[Lemma~3.2]{MAHP}.

We can now prove \ref{item:loc1}. Let $\chi\in C_0^\infty(\RR)$ be such
that $f=f\chi$. We pull the energy localization
through $\cGam$ in two steps, using both bounds in
\ref{item:locerror1} along the way,
\begin{align*}
f(H\ext)\cGam & = f(H\ext)\chi(H\ext)\cGam \\
 & = f(H\ext)\cGam\chi(H)  + \frac1{\pi}\int_{\CC}
 \bar{\partial}\tilde{\chi}(z)(H\ext-z)^{-1} f(H\ext)\bigl(\cGam H - H\ext \cGam\bigr) (H-z)^{-1} \D z\\
 & \olig f(H\ext)\cGam \chi(H)\\
& = \cGam f(H) +   \frac1{\pi}\int_{\CC}
 \bar{\partial}\tilde{f}(z)(H\ext-z)^{-1}\bigl(\cGam H - H\ext \cGam\bigr)\chi(H) (H-z)^{-1} \D z\\
 & \olig \cGam f(H).
\end{align*}
This computation establishes \ref{item:loc1}.

Finally we argue for the validity of \ref{item:loc2}. Let $\chi\in
C_0^\infty(\RR)$ be as in the proof of \ref{item:loc1}. By \ref{item:loc1} and \ref{item:locerror1} we see
that
  \begin{align}\label{PullingThroughComm}
   \nonumber & f(H\ext)[H\ext,A\ext]^\circ f(H\ext)\cGam\\
   \nonumber & \quad =f(H\ext)[H\ext,A\ext]^\circ f(H\ext)\cGam\chi(H)+f(H\ext)[H\ext,A\ext]^\circ f(H\ext)o_R(1)\\
   \nonumber & \quad \olig f(H)[H,A]^\circ \cGam f(H)+\frac1{\pi}\int_{\CC} \bar\partial\tilde
    f(z) f(H)[H\ext,A\ext]^\circ(H\ext-z)^{-1}\bigl(\cGam H - H\ext \Gamma\ext\bigr) \chi(H)(H-z)^{-1} \D z\\
    &\quad \olig f(H\ext)[H\ext,A\ext]^\circ \cGam f(H),
  \end{align}
  Here we used that $f(H\ext)[H\ext,A\ext]^\circ(H\ext-\lambda)^{-1/2}$ is bounded.
  The right-hand side of \eqref{PullingThroughComm} splits into three terms
  \begin{subequations}\label{eq:lem5.22}
    \begin{align}
    f(H\ext)[H\ext,A\ext]^\circ \cGam f(H) 
     & = f(H\ext)\D\Gamma\ext(v_\xi\cdot\nabla\omega)\cGam f(H)\label{eq:lem5.1}\\
      &\quad-f(H\ext)\D\Gamma\ext(v_\xi)\cdot\nabla\Omega(\xi-\D\Gamma\ext(k))\cGam
      f(H)\label{eq:lem5.2}\\
      &\quad-f(H\ext)\phi\ext(\ri a_\xi \coup)\cGam f(H).\label{eq:lem5.3}
    \end{align}
  \end{subequations}
  Now by \ref{item:locerror3} we find
$\eqref{eq:lem5.1}\olig f(H\ext)\cGam \D\Gamma(v_\xi\cdot\nabla\omega)f(H)$,
  by \ref{item:locerror2}, \ref{item:locerror3} and
  \ref{item:locerror1}
  \begin{align*}
    \eqref{eq:lem5.2}&\olig
    -f(H\ext)\sum_{p=1}^\nu \D\Gamma\ext(v_{\xi;p})\cGam \partial_p\Omega(\xi-\D\Gamma(k))f(H)\\
    &=-f(H\ext)\cGam
    \D\Gamma(v_\xi)\cdot\nabla\Omega(\xi-\D\Gamma(k))f(H)\\
&\quad-f(H)\sum_{p=1}^\nu \bigl(\D\Gamma\ext(v_{\xi;p})\cGam -\cGam
\D\Gamma(v_{\xi;p})\bigr)(H_0-\lambda)^{-\frac12}\partial_p\Omega(\xi-\D\Gamma(k))(H_0-\lambda)^{\frac12}f(H)\\
 &\olig -f(H\ext)\cGam
    \D\Gamma(v_\xi)\cdot\nabla\Omega(\xi-\D\Gamma(k))f(H),
  \end{align*}
  and by \ref{item:locerror4} we conclude that
  $\eqref{eq:lem5.3}\olig-f(H\ext)\cGam\phi(\ri a_\xi \coup)f(H)$.
  Putting this together -- and again using \ref{item:locerror1} and
  \ref{item:loc1} -- we see that
  \begin{align*}
  \eqref{eq:lem5.22} & \olig f(H\ext)\cGam {[H,A]}^\circ f(H)\\
    &=\chi(H\ext)\cGam f(H)[H,A]^\circ f(H)\\
    &\quad+\frac1{\pi}\int_\CC \bar\partial\tilde
    f(z)(H\ext-z)^{-1}\chi(H\ext)\bigl(\cGam H - H\ext \cGam\bigr) (H-z)^{-1} [H,A]^\circ f(H)\D z\\
    &\olig \cGam f(H)[H,A]^\circ f(H)+o_R(1)f(H)[H,A]^\circ f(H),
  \end{align*}
  as desired.
\end{proof}

We will make use of another partition of unity, this time in momentum space.
It has proved useful for the type of models studied here \cite{GGM2,MRMP}.
We take, for $r>0$,  sharp localizations $\one_0^r$ and $\one_\infty^r$  
onto sets  $\Lambda_r = \set{k\in\RR^\nu}{|k|\leq r}$ and $\Lambda_r^\mathrm{c}$ respectively.
 As multiplication operators they are projections
and this allows us to view the vector operators $\one^r = (\one_0^r,\one_\infty^r)$
as operators from $\hph$ to $\gothh_<^r\oplus \gothh_>^r$, 
with $\gothh_<^r = L^2(\Lambda_r)$ and $\gothh_>^r = L^2(\Lambda_r^\mathrm{c})$.
We can now lift this vector operator to a unitary operator
\[
\cGam(\one^r) = U \Gamma(\one^r)\colon \cF\to \cF_<^r\otimes \cF_>^r,
\]
where $\cF_<^r = \Gamma(\gothh_<^r)$ and $\cF_>^r = \Gamma(\gothh_>^r)$. Here
$U$ is again the canonical identification operator and we have abused notation by
using the same notation $\cGam$, although these operators map into a
smaller space than $\cF\ext$ and they are unitary, not merely isometric.

Since the Hamiltonian does not involve $k$-derivatives, we will not 
pick up localization errors, when applying a partition of unity in
momentum space. We introduce notation for the fiber Hamiltonians $H_r(\xi)$
restricted to $\cF_<^r$, the new \myquote{extended} Hamiltonian
$H\ext_r(\xi)$ and its building blocks $H\p{\ell}_r(\xi)$ and their
fiber operators $H\p{\ell}_r(\xi;\uk)$:
\begin{equation}\label{MomExtHam}
\begin{aligned}
& H_r(\xi) = \D\Gamma\bigl(\omega_{|\gothh_<^r}\bigr) +
\Omega\bigl(\xi-\D\Gamma(k_{|\gothh_<^r})\bigr) + \phi\bigl(\one_0^r\coup\bigr)\\
& H_r\p{\ell}(\xi;\uk) = H\bigl(\xi-\textstyle\sum_{j=1}^\ell k_j\bigr) +
\Bigl(\textstyle\sum_{j=1}^\ell \omega(k_j)\Bigr)\one_{\cF_{<}^r}\\
& H_r\p{\ell}(\xi) = \int^\oplus_{(\Lambda_r^c)^\ell}
H\p{\ell}_r(\xi;\uk)\, \D k\\
&  H\ext_r(\xi) = H_r(\xi)\oplus\Bigl(\bigoplus_{\ell=1}^\infty H\p{\ell}_r(\xi)\Bigr). 
\end{aligned}
\end{equation}
 The direct sum above is with respect to the splitting
  \begin{equation}\label{MomSpaceDecomp}
    \cF_<^r \otimes\cF_>^r
    =\cF_<^r \oplus\Bigl(\bigoplus_{\ell=1}^\infty\cF_<^r \otimes\Gamma\p{\ell}\bigl(\gothh_>^r\bigr)\Bigr),
  \end{equation}
where we identify $\cF_{<}^r\otimes \Gamma\p{\ell}(\gothh_>^r)$ with $L^{2}\sym((\Lambda_r^c)^\ell;\cF_{<}^r)$.

\begin{lemma}\label{lem:locerror1}
  Assume Conditions~\ref{Cond:MC} and~\ref{Cond:MT}, with $n_0=1$. Then
  \begin{equation*}
    \cGam(\one^r)f(H(\xi))=f(H_r\ext(\xi))\cGam(\one^r)+o_r(1).
  \end{equation*}
\end{lemma}

\begin{proof}
  Note that
  $H(\xi)=\cGam(\one^r)^*H_r\ext(\xi)\cGam(\one^r)+\phi(\one_\infty^r \coup)$.
  Composing with the unitary operator $\cGam(\one^r)$ from the left on both sides yields
  \begin{equation*}
    \cGam(\one^r)H(\xi)=H_r\ext(\xi)\cGam(\one^r)+\cGam(\one^r)\phi(\one^r_\infty \coup).
  \end{equation*}
  Subtracting $z\cGam(\one^r)$ on both sides and multiplying with
  $(H_r\ext(\xi)-z)\inv$ and $(H(\xi)-z)\inv$ from the left  and the
  right respectively, we get
  \begin{align*}
   (H_r\ext(\xi)-z)^{-1}\cGam(\one^r) & =\cGam(\one^r)(H(\xi)-z)^{-1} 
    +(H_r\ext(\xi)-z)^{-1}\cGam(\one^r)\phi(\one_\infty^r \coup)(H(\xi)-z)^{-1},
  \end{align*}
  where the term involving $\phi(\one_\infty^r\coup)$ is of order
  $\wgt{z}^{1/2}\abs{\im{z}}^{-2}o_r(1)$. The result is now obtained using the calculus
  of almost analytic extensions.
\end{proof}

Let $\kappa>0$ and $\lambda\in\RR$. Denote by $E_{0,1}\colon\RR\to \RR$ the indicator
function for the set $[-1,1]$. Abbreviate $E_{\lambda,\kappa}(t) =
E_{0,1}((t-\lambda)/\kappa)$, the indicator function for the set $[\lambda-\kappa,\lambda+\kappa]$.

\begin{theorem}[Mourre Estimate]\label{thm:Mourre} Assume
  Conditions~\ref{Cond:MC} and~\ref{Cond:MT}, with $n_0=1$.  Let
  $(\xi,\lambda)\in \cE\p{1}\backslash (\thr\p{1}\cup\Exc)$. Then there exist
  $\kappa>0$, $c>0$ and a compact self-adjoint operator $K$ such that
  \begin{align}\label{eq:Mourrethm}
    E_{\lambda,\kappa}(H(\xi))\ri [H(\xi),A_\xi]^\circ
    E_{\lambda,\kappa}(H(\xi))&\ge c E_{\lambda,\kappa}(H(\xi))+K.
  \end{align}
\end{theorem}

\begin{proof} For a pair $(\xi,\lambda)\in \cE\p{1}\backslash (\thr\p{1}\cup \Exc)$
we have in \eqref{VectorField} constructed a vector field giving rise to a conjugate operator $A_\xi$, cf.~\eqref{A-GeneralForm},
and hence extended conjugate operators $A\p{\ell}_\xi$ and $A\ext_\xi$, cf. Subsection~\ref{Sec-Extended}.
We recall from Propositions~\ref{Prop:C2} and~\ref{Prop:ExtendedC1}
that $H(\xi)$ is of class $C^1(A_\xi)$, $H\p{\ell}(\xi)$ is of class $C^1(A\p{\ell}_\xi)$,
and $H\ext(\xi)$ is of class $C^1(A\ext_\xi)$.

  Let $f\in  C_0^\infty\bigl(\bigl(\Sigma_0\p{1}(\xi),\Sigma_0\p{2}(\xi)\bigr)\bigr)$.
  Calculate using Lemma~\ref{cor:locerror}
  \begin{equation}\label{eq:Mourre1}
  \begin{aligned}
    &f(H(\xi))\,\ri [H(\xi),A_\xi]^\circ\,f(H(\xi))\\
    &\quad=\check\Gamma(j^R)^*\,\check\Gamma(j^R)\,
    f(H(\xi))\,\ri [H(\xi),A_\xi]^\circ\,f(H(\xi))\\
    &\quad=\check\Gamma(j^R)^*\,f(H\ext(\xi))\,
    \ri [H\ext(\xi),A\ext_\xi]^\circ\,f(H\ext(\xi))\,\check\Gamma(j^R)
    +o_R(1)
  \end{aligned}
  \end{equation}
  From Proposition~\ref{Prop:ExtendedC1} we know that
  \begin{equation}
    f(H\ext(\xi))\,\ri [H\ext(\xi),A\ext_\xi]^\circ\,f(H\ext(\xi))
    =\bigoplus_{\ell=0}^\infty f\bigl(H\p{\ell}(\xi)\bigr)\,
    \ri \bigl[H\p{\ell}(\xi),A\p{\ell}_\xi\bigr]^\circ\,f\bigl(H\p{\ell}(\xi)\bigr),\label{eq:splitofcomm}
  \end{equation}
  where $H\p{0}(\xi):=H(\xi)$ and $A_\xi\p{0}:=A_\xi$.
    Recalling \eqref{n-boson-threshold}, \eqref{MonotoneThresholds} and \eqref{GroundStateOfHn} we find that for $\ell\geq 2$ we have $H\p{\ell}(\xi)\ge
  \Spectrum_0\p{\ell}(\xi)\bbbone_{\cH\p{\ell}}\ge
  \Spectrum_0\p{2}(\xi)\bbbone_{\cH\p{\ell}}$. It follows that
  \begin{equation}\label{NoContrFrom-lgeq2}
   \forall \ell\geq 2:\quad f\bigl(H\p{\ell}(\xi)\bigr)=0.
  \end{equation}
  This takes care of the contributions to \eqref{eq:splitofcomm} with $\ell\geq 2$, where we can simply write
  \[
  f\bigl(H\p{\ell}(\xi)\bigr) \ri \bigl[H\p{\ell}(\xi),A\p{\ell}_\xi\bigr]^\circ
  f\bigl(H\p{\ell}(\xi)\bigr) =  f\bigl(H\p{\ell}(\xi)\bigr)^2,
 \] 
   both sides being elaborate zeroes.

     If we insert
  \eqref{eq:splitofcomm} into \eqref{eq:Mourre1} and look at the
  $\ell=0$ contribution, we get
  \begin{equation}\label{eq:Mourrecontribution0}
     \begin{aligned}
    & \Gamma(j_0^R)^*\,f(H(\xi))\,\ri [H(\xi),A_\xi]^\circ\,f(H(\xi))\,\Gamma(j_0^R)\\
      &\quad =\Gamma(j_0^R)^*\,f(H(\xi))\,\ri [H(\xi),A_\xi]^\circ\,h(H(\xi))\,f(H(\xi))\,\Gamma(j_0^R)=BK,
  \end{aligned}
  \end{equation}
  where
  \begin{equation*}
    B=\Gamma(j_0^R)^*\,f(H(\xi))\,\ri [H(\xi),A_\xi]^\circ\,h(H(\xi))\quad\text{and}\quad
    K=f(H(\xi))\,\Gamma(j_0^R).
  \end{equation*}
  Here $h\in C_0^\infty(\RR)$ equals $1$ on the support of $f$.  Note
  that $B$ is bounded, so to see that $BK$ is compact, it is enough to
  prove that $K$ is compact. Now by Lemma~\ref{lem:locerror1}
  \begin{equation}\label{LocalizedK}
    K=\cGam(\one^r)^*f(H_r\ext(\xi))\cGam(\one^r)\Gamma(j_0^R)+o_r(1).
  \end{equation}
    Like before, we split with respect to the direct sum decomposition
    \eqref{MomSpaceDecomp}, cf. also \eqref{MomExtHam}, and find
  \begin{equation}\label{eq:momentumcutoff}
    \cGam(\one^r)^*f(H_r\ext(\xi))\cGam(\one^r)\Gamma(j_0^R)
     =\cGam(\one^r)^*\Bigl\{f(H_r(\xi))\oplus\Bigl(\bigoplus_{\ell=1}^\infty f\bigl(H_r\p{\ell}(\xi)\bigr)\Bigr)\Bigr\}\cGam(\one^r)\Gamma(j_0^R).
  \end{equation}
  
  Observe now that by a variational argument we have
  $\inf\sigma(H_r(\xi))\geq \inf\sigma(H_0(\xi)+ \phi(\one_0^r\coup))$,
 and furthermore by monotonicity of $\Spectrum_0(\xi)$ as a function
 of the coupling $\coup$, cf. \cite[Corollary~2.5~(i)]{MRMP}, we get
  $\inf\sigma(H_0(\xi)+ \phi(\one_0^r\coup))\geq \Spectrum_0(\xi)$. 
  Hence 
\[
H\p{\ell}_r(\xi;\uk) \geq
\Spectrum\p{\ell}_0(\xi;\uk)\one_{\cF_{<}^r} \quad \textup{and} \quad
H\p{\ell}_r(\xi)\geq  \Spectrum\p{\ell}_0(\xi)\one_{\cF_<^r}.
\]
This ensures that $f(H_r\p{\ell}(\xi)) = 0$ for $\ell\geq 2$.
  
  Choose and $\epsilon>0$ with  $\epsilon < d(\supp(f),\Spectrum_0(\xi))$. 
 As for the term with $\ell=1$ we get similarly that
  $H_r\p{1}(\xi;k)\geq \Spectrum_0\p{1}(\xi;k)\one_{\cF_<^r}$ 
  and by \eqref{ClosingGap} (if $\omega$ is bounded) $r$ can thus be chosen so
  large that $\Spectrum_0\p{1}(\xi;k)  =
  \Spectrum_0(\xi-k)+\omega(k)\geq \Spectrum_0\p{2}(\xi)-\epsilon$, 
  for any $\abs{k}\geq r$. Hence, for $r$ large enough and $\abs{k}\geq r$ we have  $f(H_r\p{1}(\xi;k))=0$. 
  
  The only non-zero contribution to \eqref{eq:momentumcutoff} for large $r$ is thus the remaining $\ell=0$ term
  $\Gamma(\one_0^r)f(H_r(\xi))\Gamma(\one_0^r)\Gamma(j_0^R)$, which clearly is
  compact. Hence we see by letting $r\to\infty$ in \eqref{LocalizedK} that $K$ is compact.

  By \eqref{NoContrFrom-lgeq2} we only get one non-compact contribution
when inserting \eqref{eq:splitofcomm} into \eqref{eq:Mourre1}, namely one coming from the term $\ell=1$, which is
  \begin{align*}
    f\bigl(H\p{1}(\xi)\bigr)\,\ri \bigl[H\p{1}(\xi),A\p{1}_\xi\bigr]^\circ\,f\bigl(H\p{1}(\xi)\bigr).
  \end{align*}
  We can now apply Proposition~\ref{Prop-MEOne}. Let $\cJ = (\lambda-\delta,\lambda+\delta)$
  from Proposition~\ref{Prop-MEOne}, which we apply with an $f$ chosen to be equal to one
  on $[\lambda-\delta/2,\lambda+\delta/2]$. We thus get
  \[
  f(H(\xi))\,\ri [H(\xi),A_\xi]^\circ\,f(H(\xi)) \geq  c f(H(\xi))^2  - c\,\Gamma(j_0^R) f(H(\xi))^2\Gamma(j_0^R) + BK,
  \]
  which implies the theorem, with $\kappa =\delta/2$, 
 since $\Gamma(j_0^R) f(H(\xi))^2 \Gamma(j_0^R)$ was demonstrated to be compact above.
  \end{proof}

\begin{proof}[Proof of Theorem~\ref{Thm-pp}]
Items~\ref{Item-pp-3} and~\ref{Item-pp-4} are standard consequences of $H(\xi)$ being of class $C^1(A)$,
cf. Proposition~\ref{Prop:C2}~\ref{Item:HisC1}, the virial theorem
\cite{GG}, and the Mourre
estimate Theorem~\ref{thm:Mourre}. See e.g. \cite[Chapter~4.3]{DGBook}
and \cite[Section~VI]{HS}. The last statement \ref{Item-pp-5} follows
from Theorem~\ref{Thm-thr} once we have
observed that the Mourre estimate is continuous in $\xi$ and $E$:
Let $(\xi_0,E_0) \in \cE\p{1} \backslash (\Spectrum_\pp\cup\thr\p{1}\cup\Exc)$ be
given. Then the Mourre estimate, cf.~Theorem~\ref{thm:Mourre},
\[
f_{E,\kappa}(H(\xi))\ri [H(\xi),A_{\xi_0}]f_{E,\kappa}(H(\xi)) 
\geq c f_{E,\kappa}(H(\xi))^2
\]
holds true at $(\xi,E) = (\xi_0,E_0)$ for some $\kappa,c>0$. Here $f\in C_0^\infty(\RR)$, with $\supp
f\subset [-1,1]$, $0\leq f\leq 1$ and $f(t)=1$ for $|t|<1/2$. 
Finally, $f_{\lambda,\kappa}(t) = f((t-\lambda)/\kappa)$.
Recall that, being away from the point spectrum, one can squeeze away the compact
error in the Mourre estimate, by passing to a smaller $\kappa$.
We leave it to the reader to argue that both sides of the estimate
above are jointly continuous in $\xi$ and $E$, hence an estimate
of the same form, possibly with a smaller $\kappa$, will hold in a
small neighborhood of $(\xi_0,E_0)$. Hence, by the standard virial theorem \cite{GG}, there can be no point
spectrum in a small neighborhood of $(\xi_0,E_0)$. Taken together with relative
closedness of $(\thr\p{1}\cup\Exc)\cap\cE\p{1}$ in $\cE\p{1}$ we are done.   
\end{proof}

\begin{proof}[Proof of Theorem~\ref{Thm-sc}]
Under the assumptions of the theorem we have $H(\xi)$ of class
$C^2(A)$, cf. Proposition~\ref{Prop:C2}~\ref{Item:HisC2}. Hence we can conclude from Theorem~\ref{thm:Mourre} the limiting
absorption principle
\begin{equation}\label{LAP}
\sup_{\stackrel{z\in\CC,\, \im z\neq 0}{\re z\in \cJ}}
\left\|\wgt{A}^{-s}(H(\xi)-z)^{-1}\wgt{A}^{-s} \right\| <\infty,
\end{equation}
where $s>1/2$ and $\cJ\subset
\cE\p{1}(\xi)\backslash(\thr\p{1}(\xi)\cup \Exc(\xi)\cup\sigma_\pp(H(\xi)))$
is a compact interval. For a proof of this estimate, we refer the reader
to \cite{ABG,Ge2}. It is a well-known consequence of \eqref{LAP}, together with
Theorem~\ref{Thm-pp}, that the singular continuous part of
$\sigma(H(\xi))\cap \cE\p{1}(\xi)$ is empty. See e.g. \cite[Theorem~XIII.20]{RS4}.
\end{proof}

\appendix

 \section{Fibered Operators}\label{App-Fiber}
 
  Let $\cF$ be a separable Hilbert space and $(X,S)$ a measurable
  space, i.e. $S$ is a $\sigma$-algebra of subsets of $X$.

 Let $H(x)$, for $x\in X$, be a family of self-adjoint operators on $\cF$ with domain $\cD(x)$.
The family is said to be \emph{weakly resolvent measurable} if the map $x\to
(H(x) +\ri)^{-1}$ 
is weakly - hence strongly - measurable. Henceforth we simply write
measurable. 
This implies the same property for $x\to (H(x)-z)^{-1}$ for any $z\in\CC$ with $\im z\neq 0$.
We remark that if $X=\RR$ and $S$ is the Lebesgue measurable subsets
of $\RR$, then being weakly resolvent measurable
is equivalent to being measurable in the sense introduced by Nussbaum
in \cite{Nu}, a property called \myquote{N-measurable} in
\cite{GGST}. For equivalence of weak resolvent measurability and
N-measurability for self-adjoint families of operators see
\cite[Theorem~4.11]{GGST}.

By Stone-Weierstrass we can conclude that for any $f\in C_0(\RR)$,
the map $x\to f(H(x))$ is measurable. Choosing a sequence  
$f_n\in C_0(\RR)$ with $f_n(t)\to 0$ for $t\neq \lambda$ and $f_n(\lambda) = 1$ yields measurability
of eigenprojections $x\to \one_{\{\lambda\}}(H(x))$. Stone's formula now gives measurability
of $x\to \one_I(H(x))$ for any interval $I$. Since the collection of Borel sets $E$
for which $x\to \one_E(H(x))$ is measurable form a $\sigma$-algebra, we can conclude that
the property must hold true for all Borel sets. 

Equip the Cartesian product $X\times\RR$ with the product $\sigma$-algebra $S\times\mathrm{Borel}(\RR)$.
Let $F\subset \set{\psi\in\cF}{\|\psi\|=1}$ be a countable dense subset of the unit ball.
For $\psi\in F$ put
\[
f_\psi(x,\lambda) = \|(H(x)-\lambda)(H(x)+\ri)^{-1}\psi\|/\|(H(x)+\ri)^{-1}\psi\|. 
\]
Then $(x,\lambda)\to f_\psi(x,\lambda)$ is measurable.
Put $\Sigma_n = \cup_{\psi\in F}f_\psi^{-1}((-\infty,1/n))$.
Since the joint spectrum $\Sigma = \set{(x,\lambda)}{\lambda\in\sigma(H(x))}\subset X\times\RR$
can be written as $\cap_{n=1}^\infty \Sigma_n$ we conclude that  $\Sigma$ is measurable.

Similarly, for $\psi\in F$ and $n\in\NN$,  we can define $S\times\mathrm{Borel}(\RR)$ measurable
functions
\[
f_\psi\p{n}(x,\lambda) = \|(n(H(x)-\lambda) +\ri)^{-1}\psi\|. 
\]
By the spectral theorem together with Lebesgue's dominated convergence
theorem, we find that $f_\psi\p{n}(x,\lambda) \to
\|\one_{\{\lambda\}}(H(x))\psi\|$, which is thus a measurable function
of $x$ and $\lambda$. Taking supremum over $\psi\in F$, we conclude that
$(x,\lambda)\to\|\one_{\{\lambda\}}(H(x))\|$ is measurable and hence the joint point spectrum
$\Spectrum_\pp = \set{(x,\lambda)}{\lambda\in\sigma_\pp(H(x))}$ is an $S\times\mathrm{Borel}(\RR)$ 
measurable set.

Let now $\mu$ be a positive measure defined on the $\sigma$-algebra $S$.  
 Denote by $\cH$ the Hilbert space $L^2(X;\cF)\simeq \cF\otimes
 L^2(X)$, consisting of all measurable functions $X\ni x\to\psi(x)$ with $\int_X \|\psi(x)\|_\cF^2\D\mu(x)<\infty$.
The construction
\[
R(z) = \int^\oplus_X (H(x)-z)^{-1}\, \D \mu(x)
\]
yields a family of bounded operators on $\cH$ satisfying the first resolvent formula.
The operator $R(z)$ satisfies $R(z)^* = R(\bar{z})$ and has kernel
$\{0\}$. Hence, it is the resolvent family
of a self-adjoint operator $H$ densely defined on $\cH$. Its domain is
$\cD(H) = R(z)\cH$, which is independent of $z\in\set{z\in \CC}{\im z\neq 0}$.

Define a different domain by
\[
\tcD(H)  =  \Bigset{\psi\in\cH}{ \psi(x)\in\cD(x)\, \textup{a.e. and} \
\int_{X} \|H(x)\psi(x)\|_\cF^2\,\D\mu(x)<\infty}.
\]
We remark that for $\psi\in\cH$ and $\im z\neq 0$ we have
$(R(z)\psi)(x) = (H(x)-z)^{-1}\psi(x)$ a.e. and hence 
$\cD(H)\subset \tcD(H)$. Furthermore we can on $\tcD(H)$ define a symmetric
operator $\tH$ by $(\tH\psi)(x) = H(x)\psi(x)$. It is easy to see that
$H\subset \tH$ and since $H$ was self-adjoint we must have $H=\tH$ and
in particular $\cD(H) = \tcD(H)$. We remark that even with only weak
measurability of $x\to H(x)$ 
one can always construct $\tH$ as a closed operator, but without assumptions beyond
weak measurability one may not arrive at a densely defined operator,
cf. \cite[Remark~4.7]{GGST}.

 The spectral resolution $\one_E(H)$, with $E\subset \RR$ Borel, can be explicitly computed to be
\[
 \one_E(H) = \int^\oplus_X \one_E(H(x))\,\D\mu(x).
\]

If $\cD_0\subset \cD(x)$ for (almost) every $x$ is a common core for $H(x)$, then one can construct 
an essentially self-adjoint operator on the dense set of $\psi\in\cH$, with a.e. $\psi(x)\in\cD_0$.
The closure coincides with $H$ and is the situation we find ourselves
in with the fibered Nelson Hamiltonian,
where $X= \RR^\nu$ equipped with the Borel $\sigma$-algebra and
Lebesgue measure.

\providecommand{\bysame}{\leavevmode\hbox to3em{\hrulefill}\thinspace}
\providecommand{\MR}{\relax\ifhmode\unskip\space\fi MR }
\providecommand{\MRhref}[2]{%
 \href{http://www.ams.org/mathscinet-getitem?mr=#1}{#2}
}
\providecommand{\href}[2]{#2}

\end{document}